\documentclass[11pt,twoside]{article}
\usepackage{fancyhdr}
\pdfoutput=1
\usepackage[colorlinks,citecolor=blue,urlcolor=blue,linkcolor=blue,bookmarks=false]{hyperref}
\usepackage{amsfonts,epsfig,graphicx}
\usepackage{afterpage}
\usepackage{pgfplots}
\usepackage{bm}
\usepackage{tikz}
\usetikzlibrary{arrows.meta,calc,positioning}
\usepackage{nicefrac}
\usepackage{comment}
\usepackage{bbm}
\usepackage{mathtools}
\usepackage{amsmath,amssymb,amsthm} 
\usepackage{fullpage}
\usepackage[T1]{fontenc} 
\usepackage{epsf} 
\usepackage{cancel}

\usepackage{graphics} 
\usepackage{amsfonts,amsmath}
\usepackage[round]{natbib} 
\usepackage{psfrag,xspace}
\usepackage{color,etoolbox}
\usepackage{subcaption} 
\usepackage{listings}
\usepgfplotslibrary{groupplots}
\usepackage[noabbrev,capitalize, nameinlink]{cleveref}

\usepgfplotslibrary{fillbetween}

\usepackage{pgfplots}
\pgfmathdeclarefunction{poiss}{1}{%
  \pgfmathparse{(#1^x)*exp(-#1)/(x!)}%
}

\DeclareSymbolFont{bbold}{U}{bbold}{m}{n}
\DeclareSymbolFontAlphabet{\mathbbold}{bbold}

\newtheorem{theorem}{Theorem}
\newtheorem{lemma}{Lemma}
\newtheorem{corollary}{Corollary}
\newtheorem{proposition}{Proposition}

\theoremstyle{definition}

\theoremstyle{remark}

\newtheorem{condition}{Condition}

\let\hat\widehat
\definecolor{dkgreen}{rgb}{0,0.6,0}
\definecolor{gray}{rgb}{0.5,0.5,0.5}
\definecolor{mauve}{rgb}{0.58,0,0.82}
\definecolor{brightblue}{HTML}{00BFC4}

\newcommand\blfootnote[1]{%
  \begingroup
  \renewcommand\thefootnote{}%
  \footnote{#1}%
  \addtocounter{footnote}{-1}%
  \endgroup
}

\pgfplotsset{compat=1.18}
\begin{document}

\def\spacingset#1{\renewcommand{\baselinestretch}%
{#1}\small\normalsize} \spacingset{1}

\raggedbottom
\allowdisplaybreaks[1]

%%%%%%%%%%%%%%%%%%%%%%%%%%%%%%%%%%%%%%%%%%%

  \title{\vspace*{-.4in} {A Unified Framework for Rerandomization using Quadratic Forms}}
   \author{\\ $\text{Kyle Schindl}^{\dagger}$, $\text{Zach Branson}^{\ddag}$ \\ \\
    $^{\dag}$Department of Statistics \\
    Iowa State University \\
    \texttt{kschindl@iastate.edu} \\ \\ 
    $^\ddag$Department of Statistics \& Data Science \\
    Carnegie Mellon University \\
    \texttt{zach@stat.cmu.edu}
\date{}
    }

  \maketitle
  \blfootnote{Accompanying \texttt{R} code is available via \href{https://github.com/kyleschindl/rerandomization-quadratic-forms}{github.com/kyleschindl/rerandomization-quadratic-forms}}
  \thispagestyle{empty}
\spacingset{1.2}
\begin{abstract}
{\em When designing a randomized experiment, one way to ensure treatment and control groups exhibit similar covariate distributions is to randomize treatment until some prespecified level of covariate balance is satisfied; this strategy is known as rerandomization. Most rerandomization methods utilize balance metrics based on a quadratic form $\mathbf{v}^T \mathbf{A} \mathbf{v}$, where $\mathbf{v}$ is a vector of covariate mean differences and $\mathbf{A}$ is a positive semi-definite matrix. In this work, we derive general results for treatment-versus-control rerandomization schemes that employ quadratic forms for covariate balance. In addition to allowing researchers to quickly derive properties of rerandomization schemes not previously considered, our theoretical results provide guidance on how to choose $\mathbf{A}$ in practice. We find the Mahalanobis and Euclidean distances optimize different measures of covariate balance. Furthermore, we establish how the covariates' eigenstructure and their relationship to the outcomes dictate which matrix $\mathbf{A}$ yields the most precise difference-in-means estimator for the average treatment effect. We find the Euclidean distance is minimax optimal, in the sense that the difference-in-means estimator's precision is never too far from the optimal choice. We verify our theoretical results via simulation and a real data application, and demonstrate how the choice of $\mathbf{A}$ impacts the variance reduction of rerandomized experiments.}
\end{abstract}


\bigskip

\section{Introduction}

In the design stage of randomized experiments, it is important to address covariate imbalance between treatment and control groups. Large covariate imbalances can lead to increased standard errors when estimating causal effects --- this can reduce statistical power and make it more difficult to interpret results \citep{lachin1988statistical, senn1989covariate, lin2015pursuit, branson2022power}. Therefore, it is often preferable to ensure treatment groups exhibit covariate balance before the experiment is conducted. One classical experimental design strategy for ensuring covariate balance is blocking, where treatment is randomized within groups of subjects with similar categorical covariates \citep{box1978statistics, pashley2021insights}. However, blocking cannot be easily extended to accommodate many continuous variables. Instead, an experimental design tool that can handle categorical or continuous variables is rerandomization, where subjects are randomized until some prespecified level of covariate balance is achieved. Although the concept of rerandomization had been discussed as early as Fisher \citep{fisher1992arrangement}, it was not until \cite{morgan2012rerandomization} that a theoretical framework was established in which covariate balance is measured using the Mahalanobis distance between the covariate means of the treatment and control groups (often abbreviated as ``ReM''; hereafter, we simply use ``Mahalanobis Rerandomization''). Since then, there have been many extensions that also use the Mahalanobis distance, including those for tiers of covariates that vary in importance \citep{morgan2015rerandomization}, factorial designs \citep{branson2016improving, li2020rerandomization_2k}, sequential designs \citep{zhou2018sequential}, clustered experiments \citep{lu2023design}, and Bayesian designs \citep{liu2023bayesian}.

One key property of Mahalanobis Rerandomization is that it reduces the variance of all covariate mean differences by an equal amount. While this can be advantageous in some contexts, placing equal weight on all covariates can lead to poor performance when covariates are high-dimensional. One option is to define tiers of covariates based on variable importance, as suggested in \cite{morgan2015rerandomization}. However, it is necessary to specify which covariates are most important, which can be difficult to do in practice. Two methods that automatically place importance on particular covariates in a high-dimensional space are Ridge Rerandomization \citep{branson2021ridge} and PCA Rerandomization \citep{zhang2023pca}. Ridge Rerandomization uses a ridge penalty within the Mahalanobis distance to diagonalize the inverse covariance matrix, thereby placing more importance on top eigenvectors. Meanwhile, PCA Rerandomization uses the Mahalanobis distance only on the top $k$ principal components. This applies a greater amount of variance reduction to the top $k$ components than classical Mahalanobis Rerandomization, but no reduction to the bottom $d - k$ components, where $d$ is the dimension of the covariates.

Most rerandomization methods --- including Ridge Rerandomization and PCA Rerandomization --- use some kind of quadratic form $\mathbf{v}^T \mathbf{A} \mathbf{v}$, where $\mathbf{v} \in \mathbb{R}^d$ is a vector of covariate mean differences and $\mathbf{A}\in \mathbb{R}^{d \times d}$ is a positive semi-definite matrix. For example, if we let $\mathbf{\Sigma}$ denote the covariance matrix of the covariate mean differences, $\mathbf{A} = \mathbf{\Sigma}^{-1}$ corresponds to Mahalanobis Rerandomization and $\mathbf{A} = (\mathbf{\Sigma} + \lambda \mathbf{I}_d)^{-1}$ corresponds to Ridge Rerandomization. To our knowledge, this general framing of quadratic forms for rerandomization has only been noted in \cite{lu2023design} and remains largely unexplored. To fill this gap, we make several contributions. First, we derive general results for any treatment-versus-control rerandomization scheme that uses a quadratic form as its balance metric, as compared to a completely randomized experiment. This allows us to more quickly rederive previous results in the literature, as well as derive results for rerandomization schemes that have not been previously considered. Second, we establish guidance on how to optimally choose the matrix $\mathbf{A}$ in practice. While others such as \cite{liu2023bayesian} and \cite{lu2023design} have considered optimal rerandomization schemes for minimizing the variance of the difference-in-means estimator for the average treatment effect, their results require knowledge about how covariates are related to outcomes in order to be implemented in practice, which is often not available before the start of the experiment. Thus, one of our contributions is deriving optimal rerandomization schemes when outcome information is not available. We show that the Mahalanobis distance maximizes the total variance reduction across covariates and the Euclidean distance minimizes the Frobenius norm of the covariate mean differences' covariance matrix. Furthermore, we establish how the covariates' eigenstructure and their relationship to the outcomes dictate which matrix $\mathbf{A}$ yields the most precise difference-in-means estimator for the average treatment effect. We find that the Euclidean distance is minimax optimal, in the sense that the difference-in-means estimator's precision is never too far from the optimal choice, regardless of the relationship between covariates and outcomes. Consequently, our results are useful for practitioners who seek guidance on which rerandomization method is most appropriate for their data set.

That said, there are limitations to our work. First, not every rerandomization method can be expressed using quadratic forms. For example, \cite{zhao2021no} consider randomizing treatment until the $p$-values for statistical tests of covariate imbalance are above some threshold; while some of these tests can be written as a quadratic form (e.g.\ their joint acceptance rules), other tests cannot, making it unclear how they compare to the rerandomization procedures we consider in this work. Second, there are many experimental design strategies that cannot be framed as a rerandomization procedure; \cite{kallus2018optimal} suggest covariate balancing through kernel allocation, \cite{li2021covariate} suggest partitioning experimental units based on kernel density estimates, and \cite{Harshaw01102024} suggest optimizing a trade-off between robustness and covariate balance. Although we establish optimality results for rerandomization schemes involving quadratic forms, our results do not suggest whether alternative experimental design strategies may be preferable. Nonetheless, our work considers a broad class of rerandomization schemes, and furthermore provides guidance on how researchers can choose designs within this class in practice. 

The remainder of the paper is as follows. In \cref{notation} we define all important notation used in the paper. In \cref{review} we review existing rerandomization methods, with a focus on methods whose balance metrics can be expressed as a quadratic form. In \cref{QFR_section} we derive general results for the covariance of the covariate mean differences after rerandomization using quadratic forms, including optimality results for covariate balance and variance reduction for the difference-in-means estimator. In \cref{simulations} we validate our theoretical results via simulation and a real data application; we find that Euclidean Rerandomization's performance is robust across many different settings. In \cref{conclusion} we conclude with a discussion of our results and directions for future work.

\section{Notation} \label{notation}

Let $\mathbf{X} = (\mathbf{X}_1, \ldots, \mathbf{X}_n)^T \in \mathbb{R}^{n \times d}$ be the covariate matrix representing $n$ experimental units, where $\mathbf{X}_i \in \mathbb{R}^d$ denotes the vector of $d$ covariates for subject $i$. To make theoretical results notationally succinct, we assume that the columns of $\mathbf{X}$ have been centered and standardized such that each column has mean zero and variance one. Next, define $\mathbf{W} = (W_1, \ldots, W_n)^T \in \mathbb{R}^n$ to be the treatment assignment vector where $W_i = 1$ if unit $i$ has been assigned treatment and $W_i = 0$ otherwise. Let $n_1 = \sum^n_{i=1} W_i$ and $n_0 = \sum^n_{i=1}(1 - W_i)$ be the number of units in treatment and control, respectively, and $p = \frac{1}{n}\sum^n_{i=1}W_i$ be the proportion of treated subjects. We use the potential outcomes framework, where unit $i$ has fixed potential outcomes, $Y_i(1)$ and $Y_i(0)$, denoting their outcomes when $W_i = 1$ and $W_i = 0$. Thus, we assume the stable unit treatment value assumption holds, such that one subject's outcome does not depend on another subject's treatment assignment, i.e.\ no interference. Our goal is to estimate the average treatment effect, given by $\tau = \frac{1}{n} \sum^n_{i=1} (Y_i(1) - Y_i(0))$ using the difference-in-means estimator, $\widehat{\tau} = \frac{1}{n_1} \sum^n_{i=1} W_i Y_i(1) - \frac{1}{n_0} \sum^n_{i=1} (1 - W_i) Y_i(0)$. To measure covariate balance, we consider the covariate mean differences, defined as
\begin{align*}
    \widehat{\bm{\tau}}_{\mathbf{X}} :=  \bar{\mathbf{X}}_T - \bar{\mathbf{X}}_C = \frac{1}{n_1} ( \mathbf{X}^T \mathbf{W} ) - \frac{1}{n_0} ( \mathbf{X}^T(\mathbf{1}_n - \mathbf{W}) )
\end{align*}
where $\mathbf{1}_n \in \mathbb{R}^n$ is a vector whose coefficients are all equal to one. We consider the covariates $\mathbf{X}$, the proportion of treated subjects $p$, and potential outcomes as fixed, such that the only stochastic element is the treatment assignment $\mathbf{W}$. Next, let $\mathbf{\Sigma}$ denote the covariance matrix of the covariate mean differences when $\mathbf{W}$ is assigned according to complete randomization, where a random $n_1$ units are assigned to treatment and $n_0$ units are assigned to control. As shown in \cite{morgan2012rerandomization}, $\mathbf{\Sigma}$ is a fixed matrix that can be expressed in terms of the sample covariance of $(\mathbf{X}_1, \ldots, \mathbf{X}_n)$, 
\begin{align*}
    \mathbf{\Sigma} = \text{Cov}\left(\sqrt{n}\left(\bar{\mathbf{X}}_T - \bar{\mathbf{X}}_C \right) \mid \mathbf{X}\right) = \frac{\frac{1}{n-1} \sum^n_{i=1} (\mathbf{X}_i - \bar{\mathbf{X}}) (\mathbf{X}_i - \bar{\mathbf{X}})^T}{p(1-p)}.
\end{align*}
Following other works (e.g.\ \cite{li2018asymptotic}), we use the scaling term $\sqrt{n}$ in the definition of $\mathbf{\Sigma}$ to simplify asymptotic results in \cref{QFR_section}. Furthermore, throughout we condition on $\mathbf{X}$ to emphasize that the only random variable is treatment assignment. Because the potential outcomes are fixed, we implicitly condition on them as well. We define $\mathbf{S}^d_{+} = \{ \mathbf{M} \in \mathbb{R}^{d \times d} \mid \mathbf{M} \succeq 0, \mathbf{M} = \mathbf{M}^T \}$ to be the set of symmetric, positive semi-definite matrices. We use $\lambda_1, \ldots, \lambda_d$ to denote the eigenvalues of $\mathbf{\Sigma}$ and $\mathbf{\Gamma}$ to denote the orthogonal matrix of eigenvectors of $\mathbf{\Sigma}$. Similarly, we use $\eta_1, \ldots, \eta_d$ for the eigenvalues of $\mathbf{\Sigma}^{\nicefrac{1}{2}} \mathbf{A} \mathbf{\Sigma}^{\nicefrac{1}{2}}$ and $\mathbf{\Omega}$ for its eigenvectors. We use $\text{diag}\{(a_{j})_{1 \leq j \leq d} \}$ to refer to a diagonal matrix with elements $a_1, \ldots, a_d$. Finally, for $\mathbf{x} \in \mathbb{R}^d$ we define $||\mathbf{x}||^2_2 = \sum^d_{j=1}x^2_j$ to be the squared $L_2$ norm.

\section{Review of Rerandomization Methods} \label{review}

Here, we review existing rerandomization methods that utilize particular quadratic forms to balance treatment and control groups. While there are other rerandomization methods that do not rely on quadratic forms such as the marginal acceptance rules defined in \cite{zhao2021no}, we focus less on these as our goal is to establish results for distance metrics that can be written as quadratic forms. Along the way, we note how each method incorporates the principal components of $\mathbf{X}$ and the eigenstructure of $\mathbf{\Sigma}$, i.e.\ the covariance of $\sqrt{n} \widehat{\bm{\tau}}_{\mathbf{X}}$ under complete randomization. As we show in \cref{QFR_section}, the eigenstructure plays a crucial role in determining which rerandomization method is optimal for reducing the variance of the difference-in-means estimator $\widehat{\tau}$.

\subsection{Mahalanobis Rerandomization}

First investigated by \cite{morgan2012rerandomization} and extended to many other experimental design settings \citep{morgan2015rerandomization, branson2016improving, li2020rerandomization_2k, li2020rerandomization_regression, zhou2018sequential, wang2022rerandomization, shi2022rerandomization, lu2023design, wang2023rerandomization}, Mahalanobis Rerandomization balances treatment and control groups by randomizing until $M \leq a$, where $a$ is a prespecified threshold and $M = ( \sqrt{n}\widehat{\bm{\tau}}_{\mathbf{X}})^T \mathbf{\Sigma}^{-1} ( \sqrt{n}\widehat{\bm{\tau}}_{\mathbf{X}}) $ is the Mahalanobis distance between $\bar{\mathbf{X}}_T$ and $\bar{\mathbf{X}}_C$. \cite{morgan2012rerandomization} establish several key properties of Mahalanobis Rerandomization that serve as benchmarks for rerandomization research. First, they show that the difference-in-means estimator $\widehat{\tau}$ remains unbiased in finite samples conditional on $M \leq a$ provided that $\sum^n_{i=1} W_i = \sum^n_{i=1} (1 - W_i)$. Furthermore, the expectation of all observed and unobserved covariate mean differences is still zero under rerandomization. Next, the authors show that Mahalanobis Rerandomization applies an equal-percentage variance reduction to all covariates, in the sense that
\begin{align*}
    \text{Cov}(\sqrt{n}\widehat{\bm{\tau}}_{\mathbf{X}} \mid \mathbf{X}, M \leq a ) = v_a \text{Cov}( \sqrt{n}\widehat{\bm{\tau}}_{\mathbf{X}} \mid \mathbf{X})
\end{align*}
where $v_a = \mathbb{P}(\chi^2_{d+2} \leq a) / \mathbb{P}(\chi^2_d \leq a) \leq 1$ is the variance reduction term, i.e., the amount that rerandomization reduces the variances and covariances of $\widehat{\bm{\tau}}_{\mathbf{X}}$, compared to complete randomization. Finally, when the treatment effect is additive, Mahalanobis Rerandomization reduces the variance of $\widehat{\tau}$ by $1 - (1-v_a)R^2$ in finite samples, where $R^2$ is the squared multiple correlation between the potential outcomes and $\mathbf{X}$, i.e.
\begin{align*}
    R^2 = \frac{\text{Cov}(\widehat{\tau}, \widehat{\bm{\tau}}_{\mathbf{X}} \mid \mathbf{X}) \text{Cov}(\widehat{\bm{\tau}}_{\mathbf{X}} \mid \mathbf{X})^{-1} \text{Cov}(\widehat{\bm{\tau}}_{\mathbf{X}}, \widehat{\tau} \mid \mathbf{X})}{\text{Var}(\widehat{\tau} \mid \mathbf{X})}.
\end{align*}
\cite{li2018asymptotic} show that these results hold asymptotically under non-additivity and unequal sample sizes.

Despite these useful properties, placing equal priority on all covariates creates challenges in high-dimensional settings because the variance reduction term $v_a$ is increasing in $d$; as $d \to \infty$ then $v_a \to 1$, such that there is no variance reduction. Consequently, others have recommended balancing metrics that place higher priority on smaller-dimensional spaces.

\subsection{Ridge Rerandomization}

To address the problems that Mahalanobis Rerandomization suffers from in high-dimensional settings, \cite{branson2021ridge} introduce a ridge penalty $\lambda \geq 0$ to the Mahalanobis distance. They suggest randomizing until
\begin{align} \label{ridge_m}
    M_\lambda &= (\sqrt{n}\widehat{\bm{\tau}}_{\mathbf{X}})^T (\mathbf{\Sigma} + \lambda \mathbf{I}_d)^{-1} (\sqrt{n} \widehat{\bm{\tau}}_{\mathbf{X}}) 
\end{align}
is less than some prespecified $a_\lambda > 0$. Under Ridge Rerandomization, there is no longer an equal-percentage variance reduction to the covariance of the covariate mean differences. Instead, $\text{Cov}(\sqrt{n}\widehat{\bm{\tau}}_{\mathbf{X}} \mid \mathbf{X}, M_\lambda \leq a_\lambda ) = \mathbf{\Gamma} \mathbf{\Lambda}^{\nicefrac{1}{2}}( \text{diag}\{(d_{j, \lambda})_{1 \leq j \leq d} \}) \mathbf{\Lambda}^{\nicefrac{1}{2}} \mathbf{\Gamma}^T$ where $\mathbf{\Gamma}$ is the matrix of eigenvectors of $\mathbf{\Sigma}$, $\mathbf{\Lambda}$ is the diagonal matrix of eigenvalues of $\mathbf{\Sigma}$, and
\begin{align} \label{ridge_d}
    d_{j, \lambda} = \mathbb{E}\left[\mathcal{Z}^2_j \mid \mathbf{X}, \sum^d_{\ell=1} \frac{\lambda_\ell}{\lambda_\ell + \lambda} \mathcal{Z}^2_\ell \leq a_\lambda \right]
\end{align}
where $\mathcal{Z}_1, \ldots, \mathcal{Z}_d \overset{iid}{\sim} \mathcal{N}(0, 1)$ and $d_{1, \lambda} \leq \cdots \leq d_{d, \lambda} \leq 1$ are variance reduction terms. Importantly, because the covariance matrix under Mahalanobis Rerandomization can be written as $\text{Cov}(\sqrt{n}\widehat{\bm{\tau}}_{\mathbf{X}} \mid \mathbf{X}, M \leq a ) = \mathbf{\Gamma} \mathbf{\Lambda}^{\nicefrac{1}{2}}\big( v_a \mathbf{I}_d \big) \mathbf{\Lambda}^{\nicefrac{1}{2}} \mathbf{\Gamma}^T$ we can see that these two methods differ by the variance reduction they apply to the eigenvectors of $\mathbf{\Sigma}$ (or equivalently, the principal components of $\mathbf{X}$). Ridge Rerandomization applies a greater variance reduction to the top eigenvectors of $\mathbf{\Sigma}$, but a lesser variance reduction to the bottom eigenvectors as weighted by $\nicefrac{\lambda_j}{\lambda_j + \lambda}$ where $\lambda_1 \geq \cdots \geq \lambda_d$. This differs from Mahalanobis Rerandomization, which reduces the variance of each eigenvector equally. As a consequence, the variance of $\widehat{\tau}$ tends to be less under Ridge Rerandomization than under Mahalanobis Rerandomization in high-dimensional settings. However, Ridge Rerandomization does not strictly dominate Mahalanobis Rerandomization because, ultimately, the precision of $\widehat{\tau}$ after rerandomization depends on the relationship between the principal components and the potential outcomes. Furthermore, Ridge Rerandomization requires selecting the tuning parameter $\lambda$, which can be computationally intensive.

\subsection{PCA Rerandomization} \label{pcarerand}

PCA Rerandomization, introduced by \cite{zhang2023pca}, also attempts to improve upon Mahalanobis Rerandomization in high-dimensional settings. Instead of adding a penalization term, PCA Rerandomization only considers the top $k$ principal components. Let $\mathbf{X} = \mathbf{UDV}^T$ be the singular value decomposition of $\mathbf{X}$, where $\mathbf{U} \in \mathbb{R}^{n \times d}$ and $\mathbf{V} \in \mathbb{R}^{d \times d}$ are the orthogonal matrices of left and right singular vectors, and $\mathbf{D}$ is a diagonal matrix of singular values. Then, $\mathbf{Z} = \mathbf{UD}$ is the matrix of principal components of $\mathbf{X}$ and $\mathbf{Z}_k = \mathbf{U}_k \mathbf{D}_k = (\mathbf{Z}^{(k)}_1, \ldots, \mathbf{Z}^{(k)}_n)^T \in \mathbb{R}^{n \times k}$ is the matrix of the top $k$ principal components of $\mathbf{X}$. Then, PCA Rerandomization randomizes until
\begin{align} \label{mk_def_pos_df}
    M_k = \big\{\sqrt{n}(\bar{\mathbf{Z}}^{(k)}_T - \bar{\mathbf{Z}}^{(k)}_C) \big\}^T \mathbf{\Sigma}^{-1}_Z \big\{\sqrt{n}(\bar{\mathbf{Z}}^{(k)}_T - \bar{\mathbf{Z}}^{(k)}_C) \big\}
\end{align}
is less than some prespecified $a_k > 0$, where $ \mathbf{\Sigma}_Z = \text{Cov}(\mathbf{Z}_k) / p(1-p)$ is the covariance matrix of $\sqrt{n}(\bar{\mathbf{Z}}_T^{(k)} - \bar{\mathbf{Z}}_C^{(k)})$. The authors show that
\begin{align*}
    \text{Cov}(\sqrt{n}\widehat{\bm{\tau}}_{\mathbf{X}} \mid \mathbf{X}, M_k \leq a_k ) &= C_n \mathbf{V} \begin{pmatrix}
        v_{a_k} \mathbf{D}^2_k & 0 \\
        0 & \mathbf{D}^2_{d - k}
    \end{pmatrix} \mathbf{V}^T = \mathbf{\Gamma} \mathbf{\Lambda}^{\nicefrac{1}{2}} \begin{pmatrix}
        v_{a_k} \mathbf{I}_k & 0 \\
        0 & \mathbf{I}_{d - k}
    \end{pmatrix} \mathbf{\Lambda}^{\nicefrac{1}{2}} \mathbf{\Gamma}^T
\end{align*}
where $v_{a_k} = \mathbb{P}(\chi^2_{k+2} \leq a_k) / \mathbb{P}(\chi^2_k \leq a_k)$ and $C_n= [(n-1)p(1-p)]^{-1}$. Thus, PCA Rerandomization is still an equal-percentage variance reduction method, but only for the top $k$ principal components. This introduces a trade-off: for $k < d$ it follows that $v_{a_k} < v_{a_d}$, so there is a greater variance reduction to the top $k$ principal components than if the full set of covariates were included during rerandomization. However, there is a loss of reduction to the bottom $d - k$ components. Again, depending on the relationship between principal components and potential outcomes, different rerandomization schemes will be preferable. 

Although the variance of $\widehat{\tau}$ tends to be less under PCA Rerandomization than under Mahalanobis Rerandomization in high-dimensional settings, it also requires choosing the number of principal components to include. The authors suggest either choosing a fixed amount of variation explained by the top $k$ principal components (e.g.\ 50\%, 70\%, or 90\%) or using the Kaiser rule, which takes the top components whose variation is larger than the average amount of variance explained. In \cref{pca_k_section} we introduce new decision rules for the number of principal components to keep based on comparing the benefits in variance reduction to the cost of dropping principal components.

\subsection{Other Related Methods}

There are several other rerandomization methods that balance covariates based on some quadratic form. For example, \cite{zhao2021no} analyze covariate balance using $p$-values from various statistical tests. Their ``joint acceptance rules'' based on regressing the treatment vector onto the covariates can be written as a quadratic form using the Wald test statistic $\widehat{\bm{\beta}}^T \widehat{\mathbf{V}}^{-1} \widehat{\bm{\beta}}$ where $\widehat{\bm{\beta}}$ is the vector of estimated coefficients and $\widehat{\mathbf{V}}$ is the Huber-White robust covariance matrix. In \cite{liu2023bayesian} the authors suggest a prior-induced distance given by $d_\pi = \widehat{\bm{\tau}}_{\mathbf{X}}^T (\bm{\mu}_\pi \bm{\mu}_\pi^T + \mathbf{\Sigma}_\pi) \widehat{\bm{\tau}}_{\mathbf{X}}$ for some prior distribution with mean $\bm{\mu}_{\pi}$ and covariance $\mathbf{\Sigma}_{\pi}$. Finally, \cite{lu2023design} make use of weighted Euclidean distances for rerandomization, which they note are widely used in practice \citep{liweighted, hayescluster}. These, too, can be expressed as quadratic forms.

\section{Quadratic Form Rerandomization} \label{QFR_section}

The rerandomization methods in \cref{review} all balance some kind of quadratic form, i.e.\ $Q_{\mathbf{A}}(\sqrt{n} \widehat{\bm{\tau}}_{\mathbf{X}}) := (\sqrt{n} \widehat{\bm{\tau}}_{\mathbf{X}})^T \mathbf{A} (\sqrt{n} \widehat{\bm{\tau}}_{\mathbf{X}})$ for some positive semi-definite $\mathbf{A} \in \mathbb{R}^{d \times d}$. Therefore, the difference between each of these methods depends on the choice of $\mathbf{A}$, which we summarize in \cref{a_table}. Derivations for each table entry are provided in the supplementary material.
\begin{table}[h]
\centering
\begin{tabular}{r|c|c}
\textit{Method} & $\mathbf{A}$ & Reference\\ \hline
Mahalanobis & $\mathbf{\Sigma}^{-1}$  & \cite{morgan2012rerandomization, morgan2015rerandomization} \\[0.05in]
Ridge & $(\mathbf{\Sigma} + \lambda \mathbf{I}_d)^{-1}$ & \cite{branson2021ridge}  \\[0.05in]
PCA & $\propto \mathbf{V} \begin{psmallmatrix}
        \mathbf{D}^{-2}_k & \mathbf{0} \\
        \mathbf{0} & \mathbf{0}
    \end{psmallmatrix} \mathbf{V}^T$ & \cite{zhang2023pca} \\[0.05in]
$p$-value (joint test) & $\propto \mathbf{\Sigma}^{-1} \widehat{\mathbf{V}}^{-1} \mathbf{\Sigma}^{-1}$ & \cite{zhao2021no} \\[0.05in]
Weighted Euclidean & $\text{diag}\{a_1, \ldots, a_d\}$ & \cite{lu2023design}
\end{tabular}
\caption{Choices of $\mathbf{A}$ for various rerandomization methods.}
\label{a_table}
\end{table}

In this section, we establish formal results for rerandomization using any quadratic form $Q_{\mathbf{A}}(\sqrt{n} \widehat{\bm{\tau}}_{\mathbf{X}})$ as a balance metric where $\mathbf{A} \in \mathbb{R}^{d \times d}$ is a fixed, nonzero, positive semi-definite matrix. This allows us to quickly derive properties of rerandomization schemes not previously considered and determine which choice of $\mathbf{A}$ is optimal for balancing covariates and minimizing the variance of $\widehat{\tau}$. We follow the traditional rerandomization procedure: generate potential randomizations until $Q_{\mathbf{A}}(\sqrt{n} \widehat{\bm{\tau}}_{\mathbf{X}}) \leq a$. We choose $a$ based on a given acceptance probability $\alpha$, such that $\mathbb{P}(Q_{\mathbf{A}}(\sqrt{n} \widehat{\bm{\tau}}_{\mathbf{X}}) \leq a \mid \mathbf{X}) = \alpha$. One can quickly determine $a$ by Monte Carlo simulation; since asymptotically $Q_{\mathbf{A}}(\sqrt{n} \widehat{\bm{\tau}}_{\mathbf{X}}) \mid \mathbf{X} \sim \sum^d_{j=1} \eta_j \mathcal{Z}^2_j$ where $\mathcal{Z}_1, \ldots, \mathcal{Z}_d \overset{iid}{\sim} \mathcal{N}(0, 1)$ and $\eta_1, \ldots, \eta_d$ are the eigenvalues of $ \mathbf{\Sigma}^{\nicefrac{1}{2}}\mathbf{A} \mathbf{\Sigma}^{\nicefrac{1}{2}}$ \citep{mathai1992quadratic}, we can simulate from this distribution many times and then define $a$ as an empirical quantile of these draws. Alternatively, we can approximate the distribution of $Q_{\mathbf{A}}(\sqrt{n} \widehat{\bm{\tau}}_{\mathbf{X}}) \mid \mathbf{X}$ using an extension of the Welch–Satterthwaite method \citep{stewart2007simple}, which we discuss in the supplementary material. 

\subsection{Theoretical Properties} \label{theoretical_properties_sec}

First, we note that under Quadratic Form Rerandomization all covariate mean differences, whether they are observed or unobserved, are centered at zero, and $\widehat{\tau}$ is unbiased due to $Q_{\mathbf{A}}(\sqrt{n} \widehat{\bm{\tau}}_{\mathbf{X}})$ being symmetric in the treatment assignment $\mathbf{W}$, i.e.\ the acceptance rule $\varphi(\, \cdot \, , \mathbf{W})$ satisfies $\varphi(\, \cdot \, , \mathbf{W}) = \varphi(\, \cdot \, , \mathbf{1}_n - \mathbf{W})$ \citep{morgan2012rerandomization}. Clearly, treatment symmetry is satisfied since $Q_{\mathbf{A}}(\sqrt{n} \widehat{\bm{\tau}}_{\mathbf{X}}) = Q_{\mathbf{A}}(- \sqrt{n} \widehat{\bm{\tau}}_{\mathbf{X}})$. Throughout the paper we leverage results from \cite{li2018asymptotic} that establish the conditions required for us to characterize asymptotic distributions under Quadratic Form Rerandomization; we describe these conditions in detail in the supplementary material. In short, two conditions are required: \cref{asymptotic_norm_condition} states that the finite-population variance and covariance of potential outcomes, individual treatment effects, and covariates are well-defined and have limiting values asymptotically. Meanwhile, \cref{general_balance_condition} ensures that the rerandomization criterion is symmetric in the treatment assignment and places a non-zero probability on some treatment assignments. The following result establishes the covariance of the covariate mean differences under Quadratic Form Rerandomization.

\newpage 
\begin{theorem}\label{Theorem1}
    Let $Q_{\mathbf{A}}(\sqrt{n} \widehat{\bm{\tau}}_{\mathbf{X}}) = (\sqrt{n} \widehat{\bm{\tau}}_{\mathbf{X}})^T \mathbf{A} (\sqrt{n} \widehat{\bm{\tau}}_{\mathbf{X}})$ where $\mathbf{A} \in \mathbf{S}^d_{+}$ and $a > 0$. Then, under \cref{asymptotic_norm_condition} and \cref{general_balance_condition}, as $n \to \infty$,
\begin{align*}
     \text{Cov}\left(\sqrt{n}\widehat{\bm{\tau}}_{\mathbf{X}} \mid \mathbf{X},  Q_\mathbf{A}(\sqrt{n} \widehat{\bm{\tau}}_{\mathbf{X}}) \leq a\right) &= \mathbf{\Sigma}^{\nicefrac{1}{2}} \mathbf{\Omega} \Big( \text{diag}\{(\nu_{j, \eta})_{1 \leq j \leq d} \} \Big) \mathbf{\Omega}^T \mathbf{\Sigma}^{\nicefrac{1}{2}}
\end{align*}
where $\mathbf{\Omega} \in \mathbb{R}^{d \times d}$ is the orthogonal matrix of eigenvectors of $\mathbf{\Sigma}^{\nicefrac{1}{2}} \mathbf{A} \mathbf{\Sigma}^{\nicefrac{1}{2}}$ and
\begin{align} \label{q_definition}
    \nu_{j,\eta} = \mathbb{E}\left[\mathcal{Z}^2_j \mid  \sum^d_{\ell =1} \eta_\ell \mathcal{Z}^2_\ell \leq a \right] \leq 1
\end{align}
where $\mathcal{Z}_1, \ldots,\mathcal{Z}_d \overset{iid}{\sim} \mathcal{N}(0, 1)$ and $\eta_1 \geq \cdots \geq \eta_d \geq 0$ are the eigenvalues of $\mathbf{\Sigma}^{\nicefrac{1}{2}} \mathbf{A} \mathbf{\Sigma}^{\nicefrac{1}{2}}$. 
\end{theorem}  

Because the reduction terms $\nu_{1,\eta}, \dots, \nu_{d, \eta}$ are not analytically tractable, throughout the paper we make use of an approximation for $\nu_{1, \eta}, \ldots, \nu_{d, \eta}$ derived in the supplementary material of \cite{lu2023design}. They show that for $\mathbf{\Sigma}^{\nicefrac{1}{2}} \mathbf{A} \mathbf{\Sigma}^{\nicefrac{1}{2}} \succ 0$,
\begin{align} \label{q_definition_main_text}
    \nu_{j, \eta} = \mathbb{E}\left[ \mathcal{Z}^2_j \mid \sum^d_{\ell =1} \eta_\ell \mathcal{Z}^2_\ell \leq a \right] = \frac{p_d}{\eta_j} \text{det}(\mathbf{\Sigma}^{\nicefrac{1}{2}}\mathbf{A}\mathbf{\Sigma}^{\nicefrac{1}{2}})^{1/d} \alpha^{2/d} + o(\alpha^{2/d})
\end{align}
where $p_d = \frac{2\pi}{d+2}\left(\frac{2\pi^{\nicefrac{d}{2}}}{d \Gamma(\nicefrac{d}{2})} \right)^{-\nicefrac{2}{d}}$ and $\eta_1, \ldots, \eta_d$ are the eigenvalues of $\mathbf{\Sigma}^{\nicefrac{1}{2}} \mathbf{A} \mathbf{\Sigma}^{\nicefrac{1}{2}}$. Note that the remainder term $o(\alpha^{2/d})$ is taken for a fixed $d$ as $\alpha \to 0$. Thus, the remainder will be small for $\alpha$ close to zero. The terms $p_d$ and $\alpha^{\nicefrac{2}{d}}$ act as scaling factors depending on the number of covariates and acceptance probability; $\nu_{j, \eta}$ is increasing in $d$ and $\alpha$. Intuitively, this result tells us that the choice of $\mathbf{A}$ changes the values of $\nu_{1, \eta}, \ldots, \nu_{d, \eta}$ primarily through the inverse of its eigenvalues, $\frac{1}{\eta_1}, \ldots, \frac{1}{\eta_d}$ and the determinant of $\mathbf{\Sigma}^{\nicefrac{1}{2}}\mathbf{A}\mathbf{\Sigma}^{\nicefrac{1}{2}}$. 

It is useful to reflect on the geometry implied by \cref{Theorem1} in order to better understand the mechanism by which Quadratic Form Rerandomization operates. By \cref{Theorem1}, Quadratic Form Rerandomization simply rotates $\mathbf{\Sigma}^{\nicefrac{1}{2}}$ alongside the eigenvectors of $\mathbf{\Sigma}^{\nicefrac{1}{2}} \mathbf{A} \mathbf{\Sigma}^{\nicefrac{1}{2}}$ and scales them by $\nu_{1, \eta}, \ldots, \nu_{d, \eta}$ (before rotating back). Geometrically, this represents the set of all ellipsoidal constraints on the covariance of the covariate mean differences. While there are other possible shape-based constraints that could be applied (ones that cannot be written as a quadratic form), the ellipsoid is a natural choice as it directly manipulates the eigenstructure of $\mathbf{\Sigma}$. In \cite{morgan2012rerandomization}, the authors suggest that a benefit of Mahalanobis Rerandomization is that the covariance structure stays the same after rerandomization. Soon, we discuss the costs and benefits of changing the shape of the covariance matrix after rerandomization.

\subsection{Choosing the Optimal Quadratic Form for Covariate Balance} \label{opt_no_info}

Because Quadratic Form Rerandomization depends on the choice of the matrix $\mathbf{A}$, it is natural to wonder what choice is most preferable given the covariates $\mathbf{X}$. Here, we determine which choice of $\mathbf{A}$ is optimal, in the sense of constraining the covariance matrix of the covariate mean differences defined in \cref{Theorem1}. Importantly, we explore optimality without assuming that there is any information about the potential outcomes available before the experiment is conducted. Others, such as \cite{lu2023design} and \cite{liu2023bayesian}, have derived rerandomization methods that are optimal for variance reduction (under certain restrictions), but implementing these methods requires information about how the covariates are related to potential outcomes. 

We define two ways of quantifying covariance reduction. First, we consider minimizing a measure of the size of the covariance matrix defined in \cref{Theorem1}. A natural norm to consider is the Frobenius norm, given by $||\mathbf{M}||^2_F = \sum^d_{j=1} \sigma^2_j(\mathbf{M})$ for some $\mathbf{M} \in \mathbb{R}^{d \times d}$ where $\sigma^2_j(\mathbf{M})$ are the eigenvalues of $\mathbf{M}^T \mathbf{M}$. This is a natural norm to consider because $\mathbf{M}^T \mathbf{M}$ is proportional to the sample covariance matrix of $\mathbf{M}$. Second, we consider maximizing the total variance reduction $\sum^d_{j=1}(1 - \nu_{j, \eta})$. We will find that $\mathbf{A} = \mathbf{I}_d$ (Euclidean Rerandomization) is optimal for the former, and $\mathbf{A} = \mathbf{\Sigma}^{-1}$ (Mahalanobis Rerandomization) is optimal for the latter. In \cref{var_reduction_section} we consider how reductions to the covariance matrix impact the variance reduction of the difference-in-means estimator. There, we find that $\mathbf{A} = \mathbf{I}_d$ is minimax optimal, in the sense that the difference-in-means estimator's precision is never too far from the optimal choice, regardless of the relationship between covariates and outcomes. The following theorem establishes that the choice $\mathbf{A} = \mathbf{I}_d$ is optimal for minimizing the Frobenius norm of the covariance matrix.

\begin{theorem}\label{opnorm}
    For all $\mathbf{A} \in \mathbf{S}^{d}_{+}$, under \cref{asymptotic_norm_condition} and \cref{general_balance_condition}, as $\alpha \to 0$ and $n \to \infty$,
    \begin{align*}
        ||\text{Cov}(\sqrt{n} \widehat{\bm{\tau}}_{\mathbf{X}} \mid \mathbf{X}, Q_{\mathbf{I}_d}(\sqrt{n} \widehat{\bm{\tau}}_{\mathbf{X}}) \leq a)||_{F} \leq ||\text{Cov}(\sqrt{n} \widehat{\bm{\tau}}_{\mathbf{X}} \mid \mathbf{X}, Q_{\mathbf{A}}(\sqrt{n} \widehat{\bm{\tau}}_{\mathbf{X}}) \leq a^{\prime})||_{F} + o(\alpha^{\nicefrac{2}{d}})
    \end{align*}
     where $Q_{\mathbf{I}_d}(\sqrt{n} \widehat{\bm{\tau}}_{\mathbf{X}}) = ||\sqrt{n} \widehat{\bm{\tau}}_{\mathbf{X}}||^2_2$ and $(a, a^\prime)$ are chosen to have a common acceptance probability.
\end{theorem}
Note that while the Frobenius norm grows with the number of covariates, the remainder term will remain small for sufficiently small $\alpha$. Indeed, in \cref{simulations} we validate via simulation that the inequality in \cref{opnorm} without the remainder term holds across different covariate dimensions $d$. To our knowledge, Euclidean Rerandomization has not been previously studied in the literature, so we will take a moment to discuss this method. First, note that our definition of $Q_{\mathbf{I}_d}(\sqrt{n} \widehat{\bm{\tau}}_{\mathbf{X}})$ requires the covariates to be standardized such that they have unit variance. While other methods such as Mahalanobis Rerandomization are affine invariant, this is not the case for Euclidean Rerandomization --- thus, if the covariates have not been standardized, then Euclidean Rerandomization's performance may suffer. Second, because $Q_{\mathbf{I}_d}(\sqrt{n} \widehat{\bm{\tau}}_{\mathbf{X}}) \mid \mathbf{X} \sim \sum^d_{j=1} \lambda_j \mathcal{Z}^2_j$ where $\lambda_1, \ldots, \lambda_d$ are the eigenvalues of $\mathbf{\Sigma}$, the variance reduction factors under Euclidean Rerandomization are $\nu_{j, \lambda} = \mathbb{E}[\mathcal{Z}^2_j \mid \sum^d_{\ell=1} \lambda_\ell \mathcal{Z}^2_\ell \leq a ]$. This manifests in a spherical acceptance region where all directions of variation have equal magnitude after Euclidean Rerandomization. As a result, there is a greater variance reduction applied to the top eigenvectors of $\mathbf{\Sigma}$, and less applied to the bottom eigenvectors. \cref{opnorm} implies that in order to control the ``size'' of the covariance matrix after rerandomization, the best thing we can do is ensure all directions of variation are of equal magnitude.

Whereas Euclidean Rerandomization minimizes the Frobenius norm of the covariance matrix after Quadratic Form Rerandomization, the following theorem establishes that $\mathbf{A} = \mathbf{\Sigma}^{-1}$ maximizes the total variance reduction applied to the eigenvectors of $\mathbf{\Sigma}^{\nicefrac{1}{2}} \mathbf{A} \mathbf{\Sigma}^{\nicefrac{1}{2}}$.
\begin{theorem} \label{totalvarredux}
    For all $\mathbf{A} \in \mathbf{S}^d_{+}$, under \cref{asymptotic_norm_condition} and \cref{general_balance_condition}, as $\alpha \to 0$ and $n \to \infty$,
    \begin{align*}
        \sum^d_{j=1} v_a \leq \sum^d_{j=1}  \nu_{j, \eta} + o(\alpha^{\nicefrac{2}{d}})
    \end{align*}
    where $v_a = \mathbb{P}(\chi^2_{d+2} \leq a) /\mathbb{P}(\chi^2_d \leq a)$ and $\nu_{j, \eta}$ defined in \cref{q_definition}, represents the variance reduction factors for any other $\mathbf{A} \in \mathbf{S}^{d}_{+}$, under a common acceptance probability $\alpha$.
\end{theorem}

Taken together, \cref{opnorm} and \cref{totalvarredux} illustrate that Euclidean Rerandomization and Mahalanobis Rerandomization represent two extremes of a spectrum of rerandomization methods. On one side of the spectrum, Mahalanobis Rerandomization leaves the shape of the covariance matrix unchanged; it only scales the eigenvectors by $v_a$ such that there is an equal variance reduction for all eigenvectors. Meanwhile, Euclidean Rerandomization scales each eigenvector by a different factor $\nu_{j, \eta}$ such that each of the eigenvectors after rerandomization has the same magnitude.

Ultimately, the method that best reduces the variance of $\widehat{\tau}$ depends on the relationship between covariates and outcomes. In the next section, we derive the $\mathbf{A}$ that minimizes the variance of $\widehat{\tau}$ for a given relationship between covariates and potential outcomes. We then establish that Euclidean Rerandomization is minimax optimal, in the sense that the variance of $\widehat{\tau}$ after Euclidean Rerandomization is never too far from the minimum variance of $\hat{\tau}$, regardless of the covariates' relationship to the potential outcomes. 

\subsection{Variance reduction of the difference-in-means estimator $\widehat{\tau}$} \label{var_reduction_section}

In the previous section, we established that certain choices of $\mathbf{A}$ are optimal in terms of covariate balance. However, our goal is not just to obtain covariate balance, but also to estimate the average treatment effect (ATE) $\tau$ precisely. In this section, we establish how the relationship between covariates and potential outcomes suggests a choice of $\mathbf{A}$ that minimizes the variance of the difference-in-means estimator. To make this connection, we quantify how the covariates are related to the potential outcomes. 
\begin{proposition} \label{variance_of_tauhat}
    For all $\mathbf{A} \in \mathbf{S}^d_{+}$, under \cref{asymptotic_norm_condition} and \cref{general_balance_condition}, and as $n \to \infty$,
    \begin{align*}
        \text{Var}\left(\sqrt{n}(\widehat{\tau} - \tau) \mid \mathbf{X}, Q_{\mathbf{A}}(\sqrt{n} \widehat{\bm{\tau}}_{\mathbf{X}}) \leq a \right) = \bm{\beta}^T \text{Cov}\left( \sqrt{n} \widehat{\bm{\tau}}_{\mathbf{X}} \mid \mathbf{X}, Q_{\mathbf{A}}(\sqrt{n} \widehat{\bm{\tau}}_{\mathbf{X}}) \leq a  \right) \bm{\beta} + \text{Var}( \varepsilon )
    \end{align*}
    where $\bm{\beta}$ is the coefficient vector from the linear projection of $\sqrt{n}(\widehat{\tau}-\tau)$ onto $\sqrt{n}\widehat{\bm{\tau}}_{\mathbf{X}}$ under complete randomization, $\varepsilon \sim \mathcal{N}(0, V_{\tau \tau}(1 - R^2))$, and $V_{\tau \tau} = \text{Var}(\sqrt{n}(\widehat{\tau} - \tau) \mid \mathbf{X})$.
\end{proposition}
Here, $\varepsilon$ encompasses a residual term that is independent of and unaffected by rerandomization. The terms $\bm{\beta}$, $R^2$, and $V_{\tau \tau}$ have been used in other works establishing asymptotic properties of rerandomization (e.g.\ \cite{li2018asymptotic} and \cite{lu2023design}). In the supplementary material, we give a precise definition of these terms and a fuller discussion of asymptotic results for rerandomization. Note that \cref{variance_of_tauhat} does not assume that a linear model holds. Instead, because Quadratic Form Rerandomization only balances covariate means, the variance reduction is only impacted by the linear relationship between covariates and outcomes. The following theorem establishes that Quadratic Form Rerandomization always weakly reduces the variance of $\widehat{\tau}$ compared to complete randomization. When $\bm{\beta} \neq \mathbf{0}$, this reduction is strict except in degenerate cases where the rerandomization criterion does not constrain the covariate directions that are predictive of the outcome.
\begin{theorem} \label{diff_in_vars}
    Suppose $\mathbf{A} \in \mathbf{S}^{d}_{+}$. Then, under \cref{asymptotic_norm_condition} and \cref{general_balance_condition}, and as $n \to \infty$, the difference $\text{Var}(\sqrt{n}(\widehat{\tau} - \tau) \mid \mathbf{X}) - \text{Var}(\sqrt{n}(\widehat{\tau} - \tau) \mid \mathbf{X}, Q_{\mathbf{A}}(\sqrt{n} \widehat{\bm{\tau}}_{\mathbf{X}}) \leq a)$   is given by
    \begin{align} \label{diff_in_variance_theorem_eq}
         \sum^d_{j=1} \left((\mathbf{\Omega}^T   \mathbf{V})_j \mathbf{\Lambda}^{\nicefrac{1}{2}}  \bm{\beta}_Z \right)^2 (1 - \nu_{j, \eta}) \geq 0
    \end{align}
    where $\bm{\beta}_Z = \mathbf{V}^T \bm{\beta}$ and $\mathbf{V}$ is the matrix of singular vectors of $\mathbf{X}$, and $(\mathbf{\Omega}^T \mathbf{V})_j$ denotes the $j$th row of $\mathbf{\Omega}^T \mathbf{V}$. If $\mathbf{\Sigma}$ and $\mathbf{\Sigma}^{\nicefrac{1}{2}} \mathbf{A} \mathbf{\Sigma}^{\nicefrac{1}{2}}$ share an eigenbasis, then \cref{diff_in_variance_theorem_eq} simplifies to $\sum^d_{j=1} \beta^2_{Z,j} \lambda_j (1 - \nu_{j, \eta})$ where $\lambda_1, \ldots, \lambda_d$ are the eigenvalues of $\mathbf{\Sigma}$. 
\end{theorem}

This result clarifies that the variance reduction of $\widehat{\tau}$ depends on how each principal component is related to the outcomes, given by $\bm{\beta}_Z$, as well as the eigenvalues $\lambda_j$ and variance reduction factors $\nu_{j, \eta}$. Thus, rerandomization methods that place more importance on reducing the variance of the top $k$ principal components (e.g.\ Euclidean, Ridge, and PCA Rerandomization) will result in small $\text{Var}(\sqrt{n}(\widehat{\tau} - \tau) \mid \mathbf{X}, Q_{\mathbf{A}}(\sqrt{n} \widehat{\bm{\tau}}_{\mathbf{X}}) \leq a)$ only if these top components are indeed strongly related to the potential outcomes. On the other hand, if the bottom principal components are relatively important predictors of the potential outcomes, rerandomization methods that do not penalize these terms (such as Mahalanobis Rerandomization) perform better than other methods. We explore this phenomenon further via simulation in \cref{simulations}. 
\cref{diff_in_vars} quantifies how much Quadratic Form Rerandomization reduces the variance of $\widehat{\tau}$, compared to complete randomization. Meanwhile, the following theorem establishes the choice of $\mathbf{A}$ that maximizes this reduction.

\begin{theorem} \label{optimal_a_outcomes}
     For all $\mathbf{A} \in \mathbf{S}^{d}_{+}$ and $\bm{\beta} \neq \mathbf{0}$, under \cref{asymptotic_norm_condition} and \cref{general_balance_condition} and a common acceptance probability $\alpha$, as $n \to \infty$, 
    \begin{align*}
        \text{Var}(\sqrt{n}(\widehat{\tau} - \tau)\mid \mathbf{X}, Q_{\mathbf{A}^*}(\sqrt{n} \widehat{\bm{\tau}}_{\mathbf{X}}) \leq a) \leq \text{Var}(\sqrt{n}(\widehat{\tau} - \tau) \mid \mathbf{X}, Q_{\mathbf{A}}(\sqrt{n} \widehat{\bm{\tau}}_{\mathbf{X}}) \leq a^\prime)
    \end{align*}
    where $\mathbf{A}^* = \bm{\beta}\bm{\beta}^T$. Furthermore, the percent reduction in the variance of $\widehat{\tau}$ relative to complete randomization is $100(1 - \nu^*_{1, \beta})R^2$ where $\nu^*_{1, \beta} = \mathbb{E}[\chi^2_1 \mid  \left(\bm{\beta}^T \mathbf{\Sigma} \bm{\beta}\right) \! \chi^2_1 < a]$.
\end{theorem}
Unlike our previous optimality results, \cref{optimal_a_outcomes} holds without any kind of remainder. Intuitively, we can see that $\mathbf{A}^*$ directly weights the covariate mean differences by how closely they are related to the potential outcomes. Moreover, $\nu^*_{1, \beta}$ establishes a direct connection between the variance reduction applied to $\widehat{\tau}$ and the eigenstructure of $\mathbf{\Sigma}$ via the weighted eigenvalues $\beta^2_{Z, 1} \lambda_1, \ldots, \beta^2_{Z, d}\lambda_d$. \cref{optimal_a_outcomes} is similar to Theorem 4 of \cite{lu2023design}; however, \cite{lu2023design} require that the covariates are orthogonal and $\mathbf{A}$ is diagonal, whereas we only require that covariates are standardized and that $\mathbf{A}$ is positive semi-definite. Our result lines up exactly with Theorems 1 and 2 of \cite{liu2023bayesian}, who derived the optimal rerandomization strategy when utilizing a Bayesian prior to inform which covariates to balance, among all rerandomization schemes with asymptotic acceptance probability $\alpha$. However, our \cref{optimal_a_outcomes} is derived by optimizing over the class of quadratic forms, whereas \cite{liu2023bayesian} optimize over a general class of rerandomization schemes and find that a quadratic form is optimal.

The optimal matrix $\mathbf{A}^*$ in \cref{optimal_a_outcomes} can only be computed if $\bm{\beta}$ is known before the experiment is conducted, and thus typically selecting $\mathbf{A}^*$ is infeasible. However, there are cases where outcome information may be available before the start of a rerandomized experiment, e.g.\ if baseline outcome measures are recorded, or if a smaller pilot experiment was conducted beforehand. In order to provide guidance on how to choose $\mathbf{A}$ when $\bm{\beta}$ is unknown, it is useful to quantify the difference between the variance reduction of Quadratic Form Rerandomization using the optimal $\mathbf{A}^*$ and using any other $\mathbf{A}$. The following theorem shows that among all positive semi-definite choices of $\mathbf{A}$ when $\bm{\beta}$ is unknown, Euclidean Rerandomization is minimax optimal.

\begin{theorem} \label{euc_minimax}
    Suppose $\mathbf{A} \in \mathbf{S}^{d}_{+}$, $\bm{\beta} \neq \mathbf{0}$, and that \cref{asymptotic_norm_condition} and \cref{general_balance_condition} hold. Then, define
    \begin{align*}
        \Delta_{\bm{\beta}}(\mathbf{A}) = \text{Var}(\sqrt{n}(\widehat{\tau} - \tau) \mid \mathbf{X}, Q_{\mathbf{A}^*}(\sqrt{n} \widehat{\bm{\tau}}_{\mathbf{X}}) \leq a) - \text{Var}(\sqrt{n}(\widehat{\tau} - \tau) \mid \mathbf{X}, Q_{\mathbf{A}}(\sqrt{n} \widehat{\bm{\tau}}_{\mathbf{X}}) \leq a^\prime)
    \end{align*}
    to be the difference in variances under the oracle versus any other quadratic form, assuming a common acceptance probability $\alpha$. Then, as $n \to \infty$ and $\alpha \to 0$, 
\begin{align*}
            \operatorname*{arg\,min}_{\mathbf{A}\in\mathbf{S}^d_{+}}
        \sup_{||\bm{\beta}||_2<c} \left| \Delta_{\bm{\beta}}(\mathbf{A}) \right| = \{\mathbf{A} : \mathbf{\Sigma}^{\nicefrac{1}{2}} \mathbf{A} \mathbf{\Sigma}^{\nicefrac{1}{2}} \propto \mathbf{\Sigma}\}.
\end{align*}
\end{theorem}
Note that $\mathbf{A} = \mathbf{I}_d$ satisfies the proportionality requirement for minimax optimality since the eigenvalues of $\mathbf{\Sigma}^{\nicefrac{1}{2}} \mathbf{I}_d \mathbf{\Sigma}^{\nicefrac{1}{2}}$ are $\lambda_1, \ldots, \lambda_d$. We assume $||\bm{\beta}||_2 \leq c$ because by \cref{diff_in_variance_theorem_eq} if $\bm{\beta}$ is unbounded the difference in variances may be infinite. Thus, \cref{euc_minimax} suggests that Euclidean Rerandomization is a ``safe'' choice, in the sense that the variance of the difference-in-means estimator is never too far from the optimal choice. We can provide some intuition behind this result by considering the variance reduction factors under Euclidean Rerandomization, $\nu_{j, \lambda}$, and under the optimal $\mathbf{A}^*$, which we define as $\nu^*_{1, \beta}$:
\begin{align*}
    \nu_{j, \lambda} = \mathbb{E}\left[\mathcal{Z}^2_j \mid \sum^d_{\ell=1} \lambda_\ell \mathcal{Z}^2_\ell \leq a \right] \:\text{and} \: \: \:   \nu^*_{1, \beta} = \mathbb{E}\left[\mathcal{Z}^2_1 \mid \left(\bm{\beta}^T \mathbf{\Sigma} \bm{\beta}\right) \! \mathcal{Z}^2_1 \leq a \right]
\end{align*}
where $\mathcal{Z}_1, \ldots, \mathcal{Z}_d \sim \mathcal{N}(0, 1)$. Informally, these results suggest that if $\beta^2_{Z, 1}, \ldots, \beta^2_{Z, d}$ are viewed as unknown weights applied to the eigenvalues of $\mathbf{\Sigma}$, Euclidean Rerandomization ensures that each weight is only off from the optimal choice by at most $\beta^2_{Z, j}$. Since $\mathbf{I}_d$ is positive-definite, these variance reduction factors are distributed across all $d$ principal components, as opposed to $\nu^*_{1, \beta}$, which is able to concentrate the variance reduction onto a single component.

\subsection{Quadratic Form Rerandomization on Principal Components} \label{pca_k_section}

All of the previous optimality results used all $d$ covariates in the balance metric $Q_{\mathbf{A}}(\sqrt{n}\widehat{\bm{\tau}}_{\mathbf{X}})$; however, \cite{zhang2023pca} showed that there can be benefits to dropping unimportant covariates during rerandomization. In this section, we consider implementing Quadratic Form Rerandomization on the top $k$ principal components, such that the remaining $d-k$ components receive zero weight. Thus, we extend the PCA Rerandomization methodology of \cite{zhang2023pca} to quadratic forms; we show that this is analogous to considering positive semi-definite $\mathbf{A}$ of rank $k$. Similar to the definitions given in \cref{pcarerand}, let $Q^k_\mathbf{A}(\sqrt{n}\widehat{\bm{\tau}}_{\mathbf{Z}}) := \sqrt{n}(\bar{\mathbf{Z}}^{(k)}_{T} - \bar{\mathbf{Z}}^{(k)}_{C})^T \mathbf{A} \sqrt{n}(\bar{\mathbf{Z}}^{(k)}_{T} - \bar{\mathbf{Z}}^{(k)}_{C})$ be the quadratic form associated with the top $k$ principal components, where $k \leq d$. The following theorem establishes the covariance matrix of the covariate mean differences under Quadratic Form Rerandomization on the principal components and under a positive semi-definite $\mathbf{A} \in \mathbb{R}^{k \times k}$.
\begin{theorem} \label{PCA_QFR}
    Suppose $\mathbf{A} \in \mathbf{S}^{k}_{+}$. Then, under \cref{asymptotic_norm_condition} and \cref{general_balance_condition}, as $n \to \infty$,
    \begin{align*}
        \text{Cov}(\sqrt{n} \widehat{\bm{\tau}}_{\mathbf{X}} \mid \mathbf{X}, Q^k_\mathbf{A}(\sqrt{n} \widehat{\bm{\tau}}_{\mathbf{Z}}) \leq a) &=  \mathbf{\Sigma}^{\nicefrac{1}{2}} \mathbf{V}  \begin{pmatrix}
        \mathbf{\Omega}_k \Big( \text{diag}(\nu_{j, \eta}(k))_{1 \leq j \leq k} \Big) \mathbf{\Omega}^T_k   & \mathbf{0} \\
        \mathbf{0} & \mathbf{I}_{d-k}
    \end{pmatrix} \mathbf{V}^T  \mathbf{\Sigma}^{\nicefrac{1}{2}}  
    \end{align*}
    where $\nu_{j, \eta}(k) = \mathbb{E}\left[\mathcal{Z}^2_j \mid \sum^k_{\ell =1} \eta_\ell \mathcal{Z}^2_\ell \leq a \right]$ with $\mathcal{Z}_1, \ldots, \mathcal{Z}_k \overset{iid}{\sim} \mathcal{N}(0, 1)$ such that $\eta_1, \ldots, \eta_k$ are the eigenvalues of $\mathbf{\Lambda}^{\nicefrac{1}{2}}_k \mathbf{A} \mathbf{\Lambda}^{\nicefrac{1}{2}}_k$, $\mathbf{\Omega}_k$ is the matrix of eigenvectors of $\mathbf{\Lambda}^{\nicefrac{1}{2}}_k \mathbf{A} \mathbf{\Lambda}^{\nicefrac{1}{2}}_k$, $\mathbf{\Lambda}_k$ is the diagonal matrix of the top $k$ eigenvalues of $\mathbf{\Sigma}$, and $\mathbf{V}$ is the matrix of singular vectors.
\end{theorem}
If we compare choosing a positive semi-definite $\mathbf{A} \in \mathbb{R}^{d \times d}$ using \cref{Theorem1} with rerandomization using the top $k$ principal components, as defined in \cref{PCA_QFR}, we can see that the covariance structures are analogous; both produce variance reductions $\nu_{1, \eta}(k), \ldots, \nu_{d, \eta}(k)$, where here $k$ is either the rank of $\mathbf{A}$ or the number of principal components included, but where $\nu_{j, \eta}(k) = 1$ for $j > k$. This is because when $\mathbf{A}$ has rank $k < d$, then $\eta_j = 0$ for all $j > k$, and thus $\nu_{j, \eta}(k) = 1$. Importantly, the variance reduction factors $\nu_{1,\eta}(k), \ldots, \nu_{k, \eta}(k)$ from \cref{PCA_QFR} are smaller than the reduction factors $\nu_{1,\eta}, \dots, \nu_{k, \eta}$ in \cref{Theorem1} under a full $d$-dimensional quadratic form; this introduces a trade-off --- practitioners gain more reduction to the top $k$ eigenvectors of $\mathbf{\Sigma}$ at the cost of applying no reduction to the bottom $d - k$ eigenvectors. 

A natural question raised by \cref{PCA_QFR} is how one should choose the number of principal components $k$. Although \cite{zhang2023pca} discuss decision rules for $k$, their rules are only based on the percentage of variance explained. Instead, we consider selecting principal components via a cost-benefit calculation: we drop a principal component only if the benefit in added variance reduction outweighs the cost of omitting this principal component. The following proposition establishes necessary and sufficient conditions under which the reduced quadratic form has a smaller variance than the full quadratic form.
\begin{proposition} \label{drop_pcs}
    Suppose $\mathbf{A}_k \in \mathbf{S}^k_+$, $\mathbf{A}_d \in \mathbf{S}^d_+$, $\mathbf{\Sigma}$ and $\mathbf{\Sigma}^{\nicefrac{1}{2}} \mathbf{A}_d \mathbf{\Sigma}^{\nicefrac{1}{2}}$ share an eigenbasis, and $\mathbf{\Lambda}_k$ and $\mathbf{\Lambda}_k^{\nicefrac{1}{2}} \mathbf{A}_k \mathbf{\Lambda}_k^{\nicefrac{1}{2}}$ share an eigenbasis. Then, under \cref{asymptotic_norm_condition} and \cref{general_balance_condition} and a common acceptance probability $\alpha$, it follows that as $n \to \infty$,
    \begin{align*}
        \text{Var}(\sqrt{n}(\widehat{\tau} - \tau) \mid \mathbf{X}, Q^d_\mathbf{A}(\sqrt{n} \widehat{\bm{\tau}}_{\mathbf{X}}) \leq a_d) \geq \text{Var}(\sqrt{n}(\widehat{\tau} - \tau) \mid \mathbf{X}, Q^k_\mathbf{A}(\sqrt{n} \widehat{\bm{\tau}}_{\mathbf{Z}}) \leq a_k)
    \end{align*}
    if and only if
    \begin{align*}
    \sum^k_{j=1} \beta^2_{Z, j} \lambda_j (\nu_{j, \eta}(d) - \nu_{j, \eta}(k)) \geq \sum^d_{j=k+1} \beta^2_{Z, j} \lambda_j (1 - \nu_{j, \eta}(d)).
\end{align*}
\end{proposition}
As we show in the supplementary material, the eigenbasis assumption can be removed, in which case the inequality will be notationally similar to \cref{diff_in_variance_theorem_eq} in \cref{diff_in_vars}. We invoke the eigenbasis assumption to simplify notation for our decision rules here. \cref{drop_pcs} shows that it is preferable to drop the bottom $d - k$ principal components if $\beta^2_{Z, j} \lambda_j$ is very small for $j \in \{k+1, \ldots, d\}$. When $\bm{\beta}$ is known, we propose the following decision rule,
\begin{align} \label{sb_rule_beta_known}
   k = \underset{j}{\text{arg max}} \left\{  \sum^j_{i=1}\beta^2_{Z, i} \lambda_i (\nu_{i, \eta}(d) - \nu_{i, \eta}(j)) - \sum^d_{i=j+1} \beta^2_{Z, i}\lambda_i (1 - \nu_{i, \eta}(d))\right\}
\end{align}
Meanwhile, for unknown $\bm{\beta}$, an analogous decision rule is:
\begin{align} \label{sb_rule}
   k = \underset{j}{\text{arg max}} \left\{  \sum^j_{i=1} \lambda_i (\nu_{i, \eta}(d) - \nu_{i, \eta}(j)) - \sum^d_{i=j+1} \lambda_i (1 - \nu_{i, \eta}(d))\right\}.
\end{align} 
Henceforth, we refer to \cref{sb_rule} as the Weighted Eigenvalue rule. \cref{sb_rule_beta_known} suggests that, when $\bm{\beta}$ is known, it can be preferable in terms of precision to implement Quadratic Form Rerandomization using only the first $k$ principal components, rather than the full set of components $d$. However, $\bm{\beta}$ is typically not known at the start of the experiment. In this case, \cref{sb_rule} selects $k$ such that the variance reduction benefits applied to the top principal components maximally outweigh the costs incurred to the bottom principal components. This may result in less precision for the difference-in-means estimator if the bottom principal components are strongly related to the potential outcomes. If information about $\bm{\beta}$ is known before the start of an experiment, e.g.\ from a pilot study \citep{liu2023bayesian}, then \cref{sb_rule_beta_known} could be applied using estimates of $\bm{\beta}$. 

\cref{sb_rule} requires searching all $j = 1, \ldots, d$ possible values, which can be computationally intensive. As discussed earlier, an alternative way to upweight the top principal components is via Euclidean Rerandomization, i.e.\ setting $\mathbf{A} = \mathbf{I}_d$. In fact, one could place further preference on the top principal components by setting $\mathbf{A} = \mathbf{\Sigma}$, because in this case the eigenvalues of $\mathbf{\Sigma}^{\nicefrac{1}{2}} \mathbf{A} \mathbf{\Sigma}^{\nicefrac{1}{2}}$ will be $\eta_j = \lambda_j^2$ for $j = 1, \ldots, d$. This method is fast to implement, and in \cref{simulations} we find that it performs particularly well in high-dimensional settings, although it may not be as robust as Euclidean Rerandomization. This method can be generalized to any power of $\lambda_j$ by choosing $\mathbf{A} = \mathbf{\Sigma}^{c}$ for some $c \geq 0$. In later sections, we refer to this choice of $\mathbf{A}$ as Squared Euclidean Rerandomization when $c=1$.

\subsection{Conducting Inference Post-Rerandomization}

The previous sections discuss the design stage of a randomized experiment with Quadratic Form Rerandomization. Here, we discuss the analysis stage --- i.e., how one conducts inference for the ATE after the experiment has been conducted. There are two primary ways inference can be conducted. The first is with a randomization-based confidence interval \citep[Chapter 2]{rosenbaum2002}, which allows us to obtain valid, finite-population inference at the cost of computation time. Randomization-based confidence intervals are often constructed by inverting a sharp null hypothesis that specifies the relationship between potential outcomes, such as $H_0^{\tau}: Y_i(1) = Y_i(0) + \tau$ for all $i = 1, \ldots, n$. Under $H_0^{\tau}$, we can compute the potential outcomes for any hypothetical randomization and set of observed outcomes. We can conduct a randomization test for $H_0^{\tau}$ as follows:
\begin{enumerate}
    \item[1.] Generate $M$ hypothetical randomizations, where each $\mathbf{w}^{(m)}$ for $m = 1,\ldots, M$ is chosen by rerandomizing $\mathbf{W}$ until $Q_{\mathbf{A}}(\sqrt{n} \widehat{\bm{\tau}}_{\mathbf{X}}) \leq a$ for some prespecified $a > 0$.
    \item[2.] Compute a test-statistic $t(\mathbf{w}, \mathbf{X}, \mathbf{y})$ across all $M$ rerandomizations assuming that $H^\tau_0$ is true where $\mathbf{y}$ is the vector of observed outcomes.
    \item[3.] Compute the randomization-based $p$-value, defined as $ p = \frac{1 + \sum^M_{m=1} \mathbbm{1}(| t(\mathbf{w}^{(m)}, \mathbf{X}, \mathbf{y})| > |t^{\text{obs}}|)}{M + 1}$,
\end{enumerate}
where $t^{\text{obs}}$ is the observed test-statistic \citep{phipson2010permutation}. Then, the randomization-based confidence interval is given by the set of $\tau$ such that we fail to reject $H_0^{\tau}$.

The second way of conducting inference is by utilizing the asymptotic distribution of $\widehat{\tau}$, as shown in the supplementary material. In short, one can find that under \cref{asymptotic_norm_condition} and \cref{general_balance_condition}, as $n \to \infty$, 
$\sqrt{n}(\widehat{\tau} - \tau) \mid \mathbf{X}, Q_{\mathbf{A}}(\sqrt{n} \widehat{\bm{\tau}}_{\mathbf{X}}) \leq a \sim \varepsilon + \bm{\beta}^T \bm{\xi} \mid \mathbf{X},  Q_{\mathbf{A}}(\bm{\xi}) \leq a$ \citep{li2018asymptotic}. Here, $\bm{\xi} \sim \mathcal{N}(\mathbf{0}, \mathbf{\Sigma})$ and $\varepsilon \sim \mathcal{N}(0, V_{\tau \tau}(1 - R^2))$ where $V_{\tau \tau} = \frac{1}{p} S^2_{Y(1)} + \frac{1}{1-p} S^2_{Y(0)} - S^2_\tau$ such that $S^2_{Y(\cdot)}$ are the finite population variances of the outcomes in treatment and control and $S^2_\tau$ is the finite-population variance of the individual treatment effects. In practice, $S^2_{Y(1)}$, $S^2_{Y(0)}$, and $R^2$ are estimated with sample analogs. However, note that $S^2_{\tau} = \frac{1}{n-1} \sum^n_{i=1}(\tau_i - \tau )^2$ is not estimable, so estimated intervals will be conservative. From here, we can easily obtain confidence intervals for $\tau$ by computing the quantiles of $\sqrt{n} \left(\widehat{\tau} - \tau \right) \mid \mathbf{X}, Q_{\mathbf{A}}(\sqrt{n} \widehat{\bm{\tau}}_{\mathbf{X}}) \leq a$, which can be done via Monte Carlo simulation.

\section{Simulations and Application} \label{simulations}

To better understand how different rerandomization methods compare in practice in terms of covariate balance and treatment effect estimation, we consider a simulation study and real data application that compares several Quadratic Form Rerandomization methods: Euclidean, Mahalanobis, PCA (using the Mahalanobis distance on the top $k$ principal components with $k$ chosen by both the Kaiser rule and the weighted eigenvalue rule in \cref{sb_rule}), Squared Euclidean Rerandomization, and the oracle $\mathbf{A}^* = \bm{\beta} \bm{\beta}^T$.

\subsection{Design of the simulated data} \label{sim_data_def}

As discussed in \cref{QFR_section}, the eigenstructure of $\mathbf{\Sigma}$ plays an important role in determining the performance of rerandomization methods. To better understand how the eigenstructure impacts the performance of rerandomization methods, we simulate the eigenvalues of $\mathbf{\Sigma}$ from a symmetric Dirichlet distribution with concentration parameter $\gamma$, scaled by the number of covariates $d$: i.e., $\lambda_1, \dots, \lambda_d \sim d \cdot \text{Dirichlet}(\gamma)$. As a result, $\sum_{j=1}^d \lambda_j = d = \text{tr}(\mathbf{\Sigma})$. As $\gamma \to 0$, the eigenvalues become sparser and as $\gamma \to \infty$, the eigenvalues are more uniform; we consider concentration values $\gamma \in \{0.05, 0.5, 1\}$. In what follows, we refer to $\gamma = 0.05$ as ``sparse'' concentration, $\gamma = 0.5$ as ``average'', and $\gamma = 1$ as ``uniform.''

We can use this assumed distribution on the eigenvalues to gain intuition about the variance reduction factors $\nu_{1, \eta}, \ldots, \nu_{d, \eta}$ for a given $\mathbf{A}$. When the eigenvalues of $\mathbf{\Sigma}^{\nicefrac{1}{2}} \mathbf{A} \mathbf{\Sigma}^{\nicefrac{1}{2}}$ are not all equal, there is a differential percentage variance reduction applied to the covariates. As an illustration when $d = 50$, \cref{varreduxfactors} shows that methods such as Euclidean and Squared Euclidean Rerandomization provide a much greater variance reduction to the top eigenvectors, but much less to the bottom eigenvectors. Meanwhile, PCA Rerandomization applies two levels of variance reduction --- one level to the top $k$ principal components, and no reduction to the bottom $d - k$ principal components. Mahalanobis Rerandomization applies a constant level of variance reduction across all eigenvalues. 

\begin{figure}[h]
    \centering
    \begin{tikzpicture}
    \begin{axis}[
    width=0.55\textwidth,
    height=0.45\textwidth,
    xlabel={Ordered Eigenvalue},
    ylabel={Variance reduction factor},
    ymin=0,
    legend pos=outer north east,
    grid=both,
    grid style={line width=0.1pt, draw=gray!20, opacity=0.5},
    legend cell align=left,
    xtick pos=bottom,
    ytick pos=left,
  ]

  \definecolor{ggred}{HTML}{f8766d}
  \definecolor{ggblue}{HTML}{00b0f6}
  \definecolor{ggpink}{HTML}{e76bf3}
  \definecolor{ggyellow}{HTML}{a3a500}
  \definecolor{gggreen}{HTML}{00bf7c}

  \pgfplotstableread[col sep=comma]{files/averaged_qs.csv}\datatable

  \addplot+[color=ggred, mark options={fill=ggred, mark size=1.25}] table[x="idx", y="Mahalanobis", col sep=comma] {\datatable};
  \addlegendentry{Mahalanobis}

  \addplot+[color=ggblue, mark options={fill=ggblue, mark size=1.25}] table[x="idx", y="Euclidean", col sep=comma] {\datatable};
  \addlegendentry{Euclidean}
  
  \addplot+[color=ggpink, mark options={fill=ggpink, mark size=1.25}] table[x="idx", y="Kaiser", col sep=comma] {\datatable};
  \addlegendentry{Kaiser}
  
  \addplot+[color=ggyellow, mark options={fill=ggyellow, mark size=1.25}] table[x="idx", y="Weighted", col sep=comma] {\datatable};
  \addlegendentry{Wtd. Eigenvalue}
  
  \addplot+[color=gggreen, mark options={fill=gggreen, mark size=1.25}] table[x="idx", y="Squared_Euclidean", col sep=comma] {\datatable};
  \addlegendentry{Sq. Euclidean}

  \end{axis}
\end{tikzpicture}
    \caption{$\nu_{j, \eta}$ averaged over 10,000 draws of $\lambda_1, \ldots, \lambda_d \sim d \cdot \text{Dirichlet}(\gamma = 1)$.}
    \label{varreduxfactors}
\end{figure} \vspace{-0.2in}

After simulating the eigenvalues $\lambda_1, \ldots, \lambda_d$ for a given $d$, we simulate our covariates as $\mathbf{X} \sim \mathcal{N}(\mathbf{0}, \mathbf{\Sigma})$ where $\mathbf{\Sigma} = \mathbf{P}( \text{diag}\{\lambda_1, \ldots, \lambda_d\}) \mathbf{P}^T$ and $\mathbf{P} \in \mathbb{R}^{d \times d}$ is a random orthogonal matrix. After simulating $\mathbf{X}$, we center and standardize each column such that each has mean zero and variance one. For each simulation, we generate $n = 500$ observations where 250 units are assigned to treatment and 250 units are assigned to control. We define the potential outcomes $Y_i(0)$ and $Y_i(1)$ for unit $i$ as $Y_i(0) \sim \mathcal{N}(\mathbf{Z}\bm{\beta}_Z, 1)$ and $Y_i(1) = Y_i(0) + \tau$ where $\tau = 1$, $\bm{\beta}_Z \in \mathbb{R}^d$, and $\mathbf{Z} = \mathbf{X} \mathbf{V}$ is the matrix of principal components of $\mathbf{X}$. For each choice of $\gamma \in \{0.05, 0.5, 1\}$ and $d \in \{5, 25, 50, 75, 100, 150, 200, 250\}$ we simulate 10,000 data sets and implement Mahalanobis, Euclidean, PCA, Squared Euclidean, and the oracle rerandomization for each. We set $\alpha = 0.01$ and  consider three choices of $\bm{\beta}_Z$:
\begin{enumerate}
    \item[(i)] $\bm{\beta}_Z = (\lambda_d, \ldots, \lambda_1 )^T$, which preferentially weights the bottom principal components.
    \item[(ii)] $\bm{\beta}_Z = (1, \ldots, 1)^T$, which equally prioritizes each principal component.
    \item[(iii)] $\bm{\beta}_Z = (\lambda^{\frac{1}{2}}_1, \ldots, \lambda^{\frac{1}{2}}_d)^T$, which preferentially weights the top principal components.
\end{enumerate}

\subsection{Results: Variance Reduction of $\widehat{\tau}$} \label{variance_reduction_sim_section}
\cref{cov_alpha_1} displays the ratio between the standard deviation of $\widehat{\tau}$ ($y$-axis) under rerandomization vs complete randomization for $(i)$, $(ii)$, and $(iii)$ (rows) across 10,000 simulated data sets, for different numbers of covariates ($x$-axis) and concentration $\gamma$ (columns).

\begin{figure}[ht]
    \centering
    \begin{tikzpicture}
      \definecolor{ggred}{HTML}{f8766d}
  \definecolor{ggblue}{HTML}{00b0f6}
  \definecolor{ggpurple}{HTML}{e76bf3} 
  \definecolor{ggyellow}{HTML}{a3a500}
  \definecolor{ggteal}{HTML}{00bfc4}
  \definecolor{gggreen}{HTML}{00bf7c}
  \definecolor{ggdarkgreen}{HTML}{55b503}
  \definecolor{ggpink}{HTML}{fb71db} 
  \begin{groupplot}[
    group style={
      group size=3 by 3,
      horizontal sep=1.05cm,
      vertical sep=1.25cm
    },
    width=0.35\textwidth,
    height=0.20\textheight,
    grid=both,
    grid style={line width=0.1pt, draw=gray!20, opacity=0.5},
    legend style={
                at={(-1, -0.4)},
                anchor=north,
                draw=none,
                legend columns=2,
                transpose legend=true,
                /tikz/every even column/.append style={column sep=0.5cm}, 
            },
    legend cell align=left
  ]

  \pgfplotstableread[col sep=comma]{files/uniform_005.csv}\frobdt
  \pgfplotstableread[col sep=comma]{files/uniform_05.csv}\frobdtt
  \pgfplotstableread[col sep=comma]{files/uniform_1.csv}\frobdttt

  \pgfplotstableread[col sep=comma]{files/sqrt_005.csv}\opdt
  \pgfplotstableread[col sep=comma]{files/sqrt_05.csv}\opdtt
  \pgfplotstableread[col sep=comma]{files/sqrt_1.csv}\opdttt

  \pgfplotstableread[col sep=comma]{files/inv_005.csv}\mtvrdt
  \pgfplotstableread[col sep=comma]{files/inv_05.csv}\mtvrdtt
  \pgfplotstableread[col sep=comma]{files/inv_1.csv}\mtvrdttt

\nextgroupplot[title={$(i)$ $\gamma=0.05$}, 
  title style={yshift=-0.15cm},
  xtick pos=bottom,
  ytick pos=left,
  ymax=1.075,
  ymin=-0.075,
  xlabel=\empty,
  ylabel=\empty]
\addplot+[color=ggred, mark options={fill=ggred, mark size=1.5}] table[x="d", y="Mahalanobis", col sep=comma] {\mtvrdt};
\addplot+[color=ggblue, mark options={fill=ggblue, mark size=1.5}] table[x="d", y="Euclidean", col sep=comma] {\mtvrdt};
\addplot+[color=ggpurple, mark options={fill=ggpurple, mark size=1.5}] table[x="d", y="Kaiser", col sep=comma] {\mtvrdt};
\addplot+[color=ggyellow, mark options={fill=ggyellow, mark size=1.5}] table[x="d", y="Weighted", col sep=comma] {\mtvrdt};
\addplot+[color=gggreen, mark options={fill=gggreen, mark size=1.5}] table[x="d", y="Sq_Euclidean", col sep=comma] {\mtvrdt};
\addplot+[color=ggpink, mark options={fill=ggpink, mark size=1.5}] table[x="d", y="Optimal", col sep=comma] {\mtvrdt};

\nextgroupplot[title={$(i)$ $\gamma=0.5$}, 
  title style={yshift=-0.15cm},
  xtick pos=bottom,
  ytick pos=left,
  ymax=1.075,
  ymin=-0.075,
  xlabel=\empty,
  ylabel=\empty]
\addplot+[color=ggred, mark options={fill=ggred, mark size=1.5}] table[x="d", y="Mahalanobis", col sep=comma] {\mtvrdtt};
\addplot+[color=ggblue, mark options={fill=ggblue, mark size=1.5}] table[x="d", y="Euclidean", col sep=comma] {\mtvrdtt};
\addplot+[color=ggpurple, mark options={fill=ggpurple, mark size=1.5}] table[x="d", y="Kaiser", col sep=comma] {\mtvrdtt};
\addplot+[color=ggyellow, mark options={fill=ggyellow, mark size=1.5}] table[x="d", y="Weighted", col sep=comma] {\mtvrdtt};
\addplot+[color=gggreen, mark options={fill=gggreen, mark size=1.5}] table[x="d", y="Sq_Euclidean", col sep=comma] {\mtvrdtt};
\addplot+[color=ggpink, mark options={fill=ggpink, mark size=1.5}] table[x="d", y="Optimal", col sep=comma] {\mtvrdtt};

\nextgroupplot[title={$(i)$ $\gamma=1$}, 
  title style={yshift=-0.15cm},
  xtick pos=bottom,
  ytick pos=left,
  ymax=1.075,
  ymin=-0.075,
  xlabel=\empty]
\addplot+[color=ggred, mark options={fill=ggred, mark size=1.5}] table[x="d", y="Mahalanobis", col sep=comma] {\mtvrdttt};
\addplot+[color=ggblue, mark options={fill=ggblue, mark size=1.5}] table[x="d", y="Euclidean", col sep=comma] {\mtvrdttt};
\addplot+[color=ggpurple, mark options={fill=ggpurple, mark size=1.5}] table[x="d", y="Kaiser", col sep=comma] {\mtvrdttt};
\addplot+[color=ggyellow, mark options={fill=ggyellow, mark size=1.5}] table[x="d", y="Weighted", col sep=comma] {\mtvrdttt};
\addplot+[color=gggreen, mark options={fill=gggreen, mark size=1.5}] table[x="d", y="Sq_Euclidean", col sep=comma] {\mtvrdttt};
\addplot+[color=ggpink, mark options={fill=ggpink, mark size=1.5}] table[x="d", y="Optimal", col sep=comma] {\mtvrdttt};

\nextgroupplot[title={$(ii)$ $\gamma=0.05$}, 
  title style={yshift=-0.15cm},
  ylabel style={yshift=0.1cm},
  xtick pos=bottom,
  ytick pos=left,
  ymax=1.075,
  ymin=-0.075,
  xlabel=\empty,
  ylabel=\empty]
\addplot+[color=ggred, mark options={fill=ggred, mark size=1.5}] table[x="d", y="Mahalanobis", col sep=comma] {\frobdt};
\addplot+[color=ggblue, mark options={fill=ggblue, mark size=1.5}] table[x="d", y="Euclidean", col sep=comma] {\frobdt};
\addplot+[color=ggpurple, mark options={fill=ggpurple, mark size=1.5}] table[x="d", y="Kaiser", col sep=comma] {\frobdt};
\addplot+[color=ggyellow, mark options={fill=ggyellow, mark size=1.5}] table[x="d", y="Weighted", col sep=comma] {\frobdt};
\addplot+[color=gggreen, mark options={fill=gggreen, mark size=1.5}] table[x="d", y="Sq_Euclidean", col sep=comma] {\frobdt};
\addplot+[color=ggpink, mark options={fill=ggpink, mark size=1.5}] table[x="d", y="Optimal", col sep=comma] {\frobdt};

\nextgroupplot[title={$(ii)$ $\gamma=0.5$},
  title style={yshift=-0.15cm},
  xtick pos=bottom,
  ytick pos=left,
  ymax=1.075,
  ymin=-0.075,
  xlabel=\empty,
  ylabel=\empty]
\addplot+[color=ggred, mark options={fill=ggred, mark size=1.5}] table[x="d", y="Mahalanobis", col sep=comma] {\frobdtt};
\addplot+[color=ggblue, mark options={fill=ggblue, mark size=1.5}] table[x="d", y="Euclidean", col sep=comma] {\frobdtt};
\addplot+[color=ggpurple, mark options={fill=ggpurple, mark size=1.5}] table[x="d", y="Kaiser", col sep=comma] {\frobdtt};
\addplot+[color=ggyellow, mark options={fill=ggyellow, mark size=1.5}] table[x="d", y="Weighted", col sep=comma] {\frobdtt};
\addplot+[color=gggreen, mark options={fill=gggreen, mark size=1.5}] table[x="d", y="Sq_Euclidean", col sep=comma] {\frobdtt};
\addplot+[color=ggpink, mark options={fill=ggpink, mark size=1.5}] table[x="d", y="Optimal", col sep=comma] {\frobdtt};

\nextgroupplot[title={$(ii)$ $\gamma=1$}, 
  title style={yshift=-0.15cm},
  xtick pos=bottom,
  ytick pos=left,
  ymax=1.075,
  ymin=-0.075,
  xlabel=\empty,
  ylabel=\empty]
\addplot+[color=ggred, mark options={fill=ggred, mark size=1.5}] table[x="d", y="Mahalanobis", col sep=comma] {\frobdttt};
\addplot+[color=ggblue, mark options={fill=ggblue, mark size=1.5}] table[x="d", y="Euclidean", col sep=comma] {\frobdttt};
\addplot+[color=ggpurple, mark options={fill=ggpurple, mark size=1.5}] table[x="d", y="Kaiser", col sep=comma] {\frobdttt};
\addplot+[color=ggyellow, mark options={fill=ggyellow, mark size=1.5}] table[x="d", y="Weighted", col sep=comma] {\frobdttt};
\addplot+[color=gggreen, mark options={fill=gggreen, mark size=1.5}] table[x="d", y="Sq_Euclidean", col sep=comma] {\frobdttt};
\addplot+[color=ggpink, mark options={fill=ggpink, mark size=1.5}] table[x="d", y="Optimal", col sep=comma] {\frobdttt};

\nextgroupplot[title={$(iii)$ $\gamma=0.05$}, 
  title style={yshift=-0.15cm},
  xtick pos=bottom,
  ytick pos=left,
  ymax=1.075,
  ymin=-0.075,
  xlabel=\empty,
  ylabel=\empty]
\addplot+[color=ggred, mark options={fill=ggred, mark size=1.5}] table[x="d", y="Mahalanobis", col sep=comma] {\opdt};
\addplot+[color=ggblue, mark options={fill=ggblue, mark size=1.5}] table[x="d", y="Euclidean", col sep=comma] {\opdt};
\addplot+[color=ggpurple, mark options={fill=ggpurple, mark size=1.5}] table[x="d", y="Kaiser", col sep=comma] {\opdt};
\addplot+[color=ggyellow, mark options={fill=ggyellow, mark size=1.5}] table[x="d", y="Weighted", col sep=comma] {\opdt};
\addplot+[color=gggreen, mark options={fill=gggreen, mark size=1.5}] table[x="d", y="Sq_Euclidean", col sep=comma] {\opdt};
\addplot+[color=ggpink, mark options={fill=ggpink, mark size=1.5}] table[x="d", y="Optimal", col sep=comma] {\opdt};

\nextgroupplot[title={$(iii)$ $\gamma=0.5$},
  title style={yshift=-0.15cm},
  xtick pos=bottom,
  ytick pos=left,
  ymax=1.075,
  ymin=-0.075,
  xlabel=\empty,
  ylabel=\empty]
\addplot+[color=ggred, mark options={fill=ggred, mark size=1.5}] table[x="d", y="Mahalanobis", col sep=comma] {\opdtt};
\addplot+[color=ggblue, mark options={fill=ggblue, mark size=1.5}] table[x="d", y="Euclidean", col sep=comma] {\opdtt};
\addplot+[color=ggpurple, mark options={fill=ggpurple, mark size=1.5}] table[x="d", y="Kaiser", col sep=comma] {\opdtt};
\addplot+[color=ggyellow, mark options={fill=ggyellow, mark size=1.5}] table[x="d", y="Weighted", col sep=comma] {\opdtt};
\addplot+[color=gggreen, mark options={fill=gggreen, mark size=1.5}] table[x="d", y="Sq_Euclidean", col sep=comma] {\opdtt};
\addplot+[color=ggpink, mark options={fill=ggpink, mark size=1.5}] table[x="d", y="Optimal", col sep=comma] {\opdtt};

\nextgroupplot[title={$(iii)$ $\gamma=1$}, 
  title style={yshift=-0.15cm},
  xtick pos=bottom,
  ytick pos=left,
  ymax=1.075,
  ymin=-0.075,
  xlabel=\empty,
  ylabel=\empty]
\addplot+[color=ggred, mark options={fill=ggred, mark size=1.5}] table[x="d", y="Mahalanobis", col sep=comma] {\opdttt};
\addplot+[color=ggblue, mark options={fill=ggblue, mark size=1.5}] table[x="d", y="Euclidean", col sep=comma] {\opdttt};
\addplot+[color=ggpurple, mark options={fill=ggpurple, mark size=1.5}] table[x="d", y="Kaiser", col sep=comma] {\opdttt};
\addplot+[color=ggyellow, mark options={fill=ggyellow, mark size=1.5}] table[x="d", y="Weighted", col sep=comma] {\opdttt};
\addplot+[color=gggreen, mark options={fill=gggreen, mark size=1.5}] table[x="d", y="Sq_Euclidean", col sep=comma] {\opdttt};
\addplot+[color=ggpink, mark options={fill=ggpink, mark size=1.5}] table[x="d", y="Optimal", col sep=comma] {\opdttt};

\legend{
    Mahalanobis, Euclidean, Kaiser, Wtd. Eigenvalue, Sq. Euclidean, Oracle
}

  \end{groupplot}
    \node[below=0.5cm] at ($(group c1r3.south)!0.5!(group c3r3.south)$) {Number of covariates};
     \node[rotate=90] at ($(group c1r1.west)!0.5!(group c1r3.west) + (-1.1cm, 0)$) {Variance Reduction};
\end{tikzpicture}
    \caption{Variance reduction across different numbers of covariates and concentrations.}
    \label{cov_alpha_1} 
\end{figure}

We can see that Mahalanobis Rerandomization is the best (besides the oracle) in terms of variance reduction in $(i)$, whereas Euclidean Rerandomization is the best at constraining the variance in $(ii)$. In setting $(iii)$, the Mahalanobis distance performs the worst, in the sense that it yields the largest relative standard error, whereas Euclidean Rerandomization is still a close second in terms of variance reduction, reflecting the robustness of Euclidean Rerandomization. PCA Rerandomization using the decision rule in \cref{sb_rule} performs better than the Kaiser rule when eigenvalues are sparse, but the two methods are approximately equivalent when the eigenvalues become more uniform. Furthermore, as the eigenvalues become more uniform there is less information that can be derived from the covariance matrix, and as a result all rerandomization methods are closer to each other. Finally, we note that the oracle is relatively more precise as $d$ increases, which follows due to the increasing $R^2$ and decreasing $\nu^*_{1, \beta}$ across $d$, as illustrated by \cref{optimal_a_outcomes}.

Importantly, we can see that practitioners should take great care when the distribution of eigenvalues is sparse. In this case, there are large differences among rerandomization methods; Mahalanobis Rerandomization does better than all other methods in setting $(i)$, but significantly worse in settings $(ii)$ and $(iii)$. Euclidean Rerandomization is a good hedge against this risk, as it constrains the covariance matrix the most, and produces the second greatest amount of variance reduction across all settings. This aligns with \cref{euc_minimax} which showed that Euclidean Rerandomization is minimax optimal in terms of the variance of the difference-in-means estimator $\widehat{\tau}$. Furthermore, it is notable that Squared Euclidean Rerandomization and Euclidean Rerandomization are nearly identical in performance across simulations. Thus, methods that upweight top eigenvalues like Squared Euclidean Rerandomization may outperform Euclidean Rerandomization when the top principal components are more strongly related to the potential outcomes, but may underperform if the bottom principal components have significant importance.

\subsection{Application}

So far, our simulations have only considered the setting where the covariates are linearly related to the potential outcomes and the treatment effect is additive. Here, we relax those assumptions using a real experiment and constructing synthetic outcomes. In our real data example, we consider a randomized controlled trial in which 308 low-socioeconomic-status German adolescents were randomly assigned mentors to help them prepare for professional life, where randomization followed a pairwise matching design with rerandomization to ensure covariate balance \citep{mentoringpaper}. We include $d=21$ covariates --- these include continuous variables such as age and personality scores, and categorical variables such as grade level, sex, or migrant status. We recode all categorical variables with binary indicator variables for each category. Furthermore, if any categorical variables were missing, we denoted their missingness with indicators. The experiment yielded outcomes $Y_i$, a continuous measure of labor market outcomes. In reality, we only know the outcomes $Y_i$ and not the full potential outcomes $Y_i(1), Y_i(0)$. For the purposes of comparing rerandomization methods, we generate potential outcomes in a way that induces non-linearity and non-additivity. For each subject $i$, we define the potential outcomes $Y_i(0) = \widehat{\mathbb{E}}\left[Y_i \mid X_i, W_i = 0\right]$ and $Y_i(1) = \widehat{\mathbb{E}}\left[Y_i \mid X_i, W_i = 1\right]$ where we use flexible, non-linear models for estimating the regression function via the \texttt{SuperLearner} package in \texttt{R}. Specifically, we used an ensemble modeling approach with a candidate library that included random forests, generalized additive models, gradient boosted models, LASSO, and generalized linear models. This yields data $(\mathbf{X}_i, W_i, Y_i(0), Y_i(1))$ such that the data are non-linearly related to the potential outcomes and the treatment effect is not additive. 

\begin{figure}[h]
    \centering
\begin{tikzpicture}
  \begin{axis}[
    width=0.6\textwidth,
    height=0.375\textwidth,
    grid=both,
    grid style={line width=0.1pt, draw=gray!20, opacity=0.5},
    xlabel={Component Number},
    ymax=0.55,
    ylabel={Variance explained},
    ylabel style={yshift=0cm},
    legend cell align=left,
    xtick pos=bottom,
    ytick pos=left,
    legend image post style={sharp plot,-, very thick},
  ]

    \pgfplotstableread[col sep=comma]{files/mentor_scree_plot.csv}\msm

    \addplot+[color=black,  mark options={fill=black, mark size=1.5}] 
      table[x="d", y="var_explained", col sep=comma] {\msm};

  \end{axis}
\end{tikzpicture}
    \caption{Scree plot for the mentoring data set.}
    \label{scree_plot_mentoring}
\end{figure}

Once covariates and outcomes are defined, we ran each Quadratic Form Rerandomization method 10,000 times with acceptance probability $\alpha = 0.01$, and computed the difference-in-means estimator $\widehat{\tau}$ for each replication. For comparison, we did the same for 10,000 complete randomizations. \cref{real_data_table} shows how the precision of the difference-in-means estimator compares across rerandomization methods, relative to complete randomization. 

\begin{table}[h]
\centering
\begin{tabular}{r|c}
\textit{Method} & $\text{Var}(\widehat{\tau} \mid \mathbf{X}, Q_{\mathbf{A}}(\sqrt{n} \widehat{\bm{\tau}}_{\mathbf{X}}) \leq a) / \text{Var}(\widehat{\tau} \mid \mathbf{X})$\\ \hline
Mahalanobis & 0.611  \\
Euclidean & \textbf{0.565}   \\
Kaiser $(k = 8)$ & 0.585 \\
Weighted Eigenvalue $(k = 16)$ & 0.608  \\
Sq. Euclidean & 0.576
\end{tabular}
\caption{Precision of $\widehat{\tau}$ under rerandomization relative to complete randomization.}
\label{real_data_table}
\end{table}

Although performance was close across methods, Euclidean Rerandomization yielded the smallest variance. The reason for the close comparison between methods likely stems from the eigenstructure of the covariates, as illustrated by the scree plot in \cref{scree_plot_mentoring}. Since the distribution of eigenvalues is very uniform, there is less information from the covariance matrix that can be leveraged during rerandomization. Informally, this is a setting close to the one described in \cref{variance_reduction_sim_section} when $\gamma = 1$. Regardless, the results in \cref{real_data_table} align with \cref{euc_minimax}, which shows that Euclidean Rerandomization is minimax optimal and thus, its variance reduction is more robust than other methods across many settings.

\section{Discussion and Conclusion} \label{conclusion}

Much of the rerandomization literature has focused on applying the Mahalanobis distance to different experimental settings \citep{morgan2012rerandomization, morgan2015rerandomization, branson2016improving, li2018asymptotic, zhou2018sequential, lu2023design}. However, the Mahalanobis distance can perform poorly in high-dimensional settings. Recent works such as \cite{branson2021ridge} and \cite{zhang2023pca} have attempted to address this problem by using quadratic forms other than the Mahalanobis distance as a balance metric for rerandomization. In our analysis, we derive general results for rerandomization using any quadratic form and establish which method is optimal under different conditions.

Under Quadratic Form Rerandomization, defined by a positive semi-definite matrix $\mathbf{A}$, we show that the variance reduction applied to each principal component of $\mathbf{X}$ is weighted by the eigenvalues of $\mathbf{\Sigma}^{\nicefrac{1}{2}} \mathbf{A} \mathbf{\Sigma}^{\nicefrac{1}{2}}$ where $\mathbf{\Sigma} = \text{Cov}(\sqrt{n}(\bar{\mathbf{X}}_T - \bar{\mathbf{X}}_C) \mid \mathbf{X})$. In terms of covariate balance, we found in \cref{opt_no_info} that $\mathbf{A} = \mathbf{I}_d$ minimizes the Frobenius norm of the covariance of the covariate mean differences after rerandomization, and $\mathbf{A} = \mathbf{\Sigma}^{-1}$ maximizes the total variance reduction applied to the principal components. Meanwhile, in \cref{var_reduction_section}, we showed that the precision of the difference-in-means estimator after rerandomization depends on the variance explained by each principal component and their relationship with the potential outcomes. We found that Euclidean Rerandomization is minimax optimal in the sense that the difference-in-means estimator's precision is never too far from its precision under the optimal choice of $\mathbf{A}$. Our results establish properties about a general class of rerandomization methods and provide guidance for navigating this class in practice.

We validated these theoretical results via simulation and a real application, and found that---while the optimal rerandomization method depends on the unknown relationship between the principal components and outcomes---Euclidean Rerandomization yielded fairly robust results across simulation settings. That said, if the practitioner has prior information that the top principal components are strongly related to the outcomes, they can employ methods like Squared Euclidean or PCA Rerandomization to achieve greater variance reduction than Euclidean Rerandomization, but at the expense of robustness.

Although this work establishes properties for a broad class of rerandomization methods, we only focus on methods that can be expressed by quadratic forms. Future work could establish methods for comparing rerandomization with quadratic forms to other rerandomization methods and experimental design strategies. Furthermore, although we focus on two-arm, completely randomized experiments, the quadratic form framework is not tied to this setting. For example, in multi-armed experiments one could stack imbalances across arms and apply quadratic form acceptance rules to this higher-dimensional covariate imbalance vector. Furthermore, Quadratic Form Rerandomization can be layered on top of other randomization schemes (e.g.\ stratified/block, cluster, or split-plot) by extending the approaches in \cite{wang2023rerandomization, lu2023design, shi2022rerandomization}; these are all natural directions for future work. In particular, sequential rerandomization could be especially useful, as practitioners may be able to use outcome information in prior experiments to inform future experiments. This outcome information could be useful not only for choosing which covariates to balance, but also which quadratic form to use when rerandomizing.

\section*{Acknowledgments}
The authors report no conflicts of interest. This research was supported through a fellowship provided by the Novartis Pharmaceuticals Corporation. We especially thank Joel Greenhouse for his many insightful comments and suggestions.

\section*{References}
\vspace{-1cm}
\bibliographystyle{abbrvnat}
\bibliography{references}

\setlength{\parindent}{0cm}
\appendix

\begin{center}
{\large\bf SUPPLEMENTARY MATERIAL}
\end{center}

\begin{description}

\item \cref{asymptotic_theory_sec}: Contains regularity conditions for asymptotic distributions.

\item \cref{geometry_section}: Contains a discussion of the geometry and eigenstructure of the covariance of the covariate mean differences under Quadratic Form Rerandomization.

\item \cref{quad_form_dist_section}: Contains a discussion on the distribution of quadratic forms and how to simulate acceptance thresholds in practice.

\item \cref{proofs_section}: Contains all proofs for our theoretical results, including:
\begin{description}
    \item \cref{theorem_1_proof}: Proof of \cref{Theorem1}.
    \item \cref{opnorm_proof}: Proof of \cref{opnorm}.
    \item \cref{totalvarredux_proof}: Proof of \cref{totalvarredux}.
    \item \cref{diff_in_vars_proof}: Proof of \cref{diff_in_vars}.
    \item \cref{optimal_a_outcomes_proof}: Proof of \cref{optimal_a_outcomes}.
    \item \cref{euc_minimax_proof}: Proof of \cref{euc_minimax}.
    \item \cref{PCA_QFR_proof}: Proof of \cref{PCA_QFR}.
    \item \cref{drop_pcs_proof}: Proof of \cref{drop_pcs}.
    \item \cref{eigenvalues_qfr_proof}: Proof of \cref{eigenvalues_qfr}.
\end{description}

\end{description}

\section{Asymptotic Theory} \label{asymptotic_theory_sec}

Here, we discuss the asymptotic distribution of the difference-in-means estimator $\widehat{\tau}$ and the covariate mean differences $\bar{\mathbf{X}}_T - \bar{\mathbf{X}}_C$ after Quadratic Form Rerandomization. To begin, we introduce some definitions and results derived in \cite{li2018asymptotic}. Let $\bar{Y}(w) = \frac{1}{n} \sum^n_{i=1} Y_i(w)$ denote the finite population average of the potential outcomes under treatment arm $w \in \{0, 1\}$ and similarly, define $S^2_{Y(w)} = \frac{1}{n-1} \sum^n_{i=1} (Y_i(w) - \bar{Y}(w))^2$ to be the finite population variance. Next, let $S^2_\tau = \frac{1}{n-1} \sum^n_{i=1}(\tau_i - \tau)^2$ be the finite population variance of the individual treatment effects where $\tau_i = Y_i(1) - Y_i(0)$, and $\mathbf{S}_{Y(w), \mathbf{X}} = \mathbf{S}^T_{\mathbf{X}, Y(w)} = \frac{1}{n-1} \sum^n_{i=1} (Y_i(w) - \bar{Y}(w))(\mathbf{X}_i - \bar{\mathbf{X}})^T$ be the finite population covariance between the potential outcomes and covariates, where $\bar{\mathbf{X}} = \frac{1}{n}\sum^n_{i=1} \mathbf{X}_i$. Finally, let $\mathbf{S}^2_{\mathbf{X}} = \frac{1}{n-1} \sum^n_{i=1} (\mathbf{X}_i - \bar{\mathbf{X}})(\mathbf{X}_i - \bar{\mathbf{X}})^T$ denote the sample covariance. In what follows, we use the notation $\sim$ to denote two sequences of random vectors converging weakly to the same distribution. With these definitions in place, \cite{li2018asymptotic} establish the following as a sufficient condition for the asymptotic normality of $ \sqrt{n}( \widehat{\tau} - \tau, (\bar{\mathbf{X}}_T - \bar{\mathbf{X}}_C)^T)$ under complete randomization: 

\begin{condition} \label{asymptotic_norm_condition}
For $w = 0, 1$, as $n \to \infty$,
    \begin{enumerate}
    \item[(i)] The proportion of treated units under treatment arm $w$ has positive limits.
    \item[(ii)] The finite population variances and covariances $S^2_{Y(w)}$, $S^2_\tau$, $\mathbf{S}^2_{\mathbf{X}}$, $\mathbf{S}^2_{\mathbf{X}, Y(w)}$ have finite limiting values.
    \item[(iii)] $\underset{1 \leq i \leq n}{\text{max}} \left\{ \frac{1}{n} \left| Y_i(w) - \bar{Y}(w)\right|^2 \right\} \to 0$ and $\underset{1 \leq i \leq n}{\text{max}} \frac{1}{n} \left\{ \left| \left| \mathbf{X}_i - \bar{\mathbf{X}} \right|\right|^2_2 \right\} \to 0$.
\end{enumerate}
\end{condition}
Then, by \cite{li2018asymptotic}, under the assumption that \cref{asymptotic_norm_condition} holds,  
\begin{align*}
     \sqrt{n}\begin{pmatrix}
        \widehat{\tau} - \tau \\
        \bar{\mathbf{X}}_T - \bar{\mathbf{X}}_C
    \end{pmatrix} \mid \mathbf{X} \sim \begin{pmatrix}
        T_\tau \\
        \mathbf{T}_{\mathbf{X}}
    \end{pmatrix}
\end{align*}
where $(T_\tau, \mathbf{T}^T_{\mathbf{X}})$ is a random vector from $\mathcal{N}(\mathbf{0}, \mathbf{\Sigma}_{\mathbf{V}})$ where
\begin{align*}
    \mathbf{\Sigma}_{\mathbf{V}} = \begin{bmatrix}
        V_{\tau \tau} & \mathbf{V}_{\tau x} \\
        \mathbf{V}_{x \tau} & \mathbf{V}_{xx}
    \end{bmatrix}
\end{align*}
such that $V_{\tau \tau} = \frac{1}{p} S^2_{Y(1)} + \frac{1}{1-p} S^2_{Y(0)} - S^2_\tau$, $\mathbf{V}_{\tau x} = \frac{1}{p} \mathbf{S}_{Y(1), \mathbf{X}} + \frac{1}{1-p} \mathbf{S}_{Y(0), \mathbf{X}}$, and $\mathbf{V}_{xx} = \frac{\mathbf{S}^2_{\mathbf{X}}}{p(1-p)}  = \mathbf{\Sigma}$. After establishing the sufficient conditions for asymptotic normality, \cite{li2018asymptotic} also derive the asymptotic distribution of $ \sqrt{n}( \widehat{\tau} - \tau, (\bar{\mathbf{X}}_T - \bar{\mathbf{X}}_C)^T)$ under general covariate balance criteria that depend only on $\sqrt{n}(\bar{\mathbf{X}}_T - \bar{\mathbf{X}}_C)$ and $\mathbf{\Sigma}$, i.e., $\phi(\sqrt{n}(\bar{\mathbf{X}}_T - \bar{\mathbf{X}}_C), \mathbf{\Sigma})$, where $\phi(\cdot, \cdot)$ is an indicator function that determines whether or not a treatment allocation is acceptable. The following are sufficient conditions for the asymptotic conditional distribution under general covariate balance criteria:
\begin{condition} \label{general_balance_condition} The covariate balance criterion $\phi(\cdot, \cdot)$ satisfies:
    \begin{enumerate}
    \item[(i)] $\phi(\cdot, \cdot)$ is almost surely continuous.
    \item[(ii)] For $\bm{\xi} \sim \mathcal{N}(\mathbf{0}, \mathbf{\Sigma})$, $\mathbb{P}(\phi(\bm{\xi}, \mathbf{\Sigma}) = 1) > 0$ for any $\mathbf{\Sigma} > 0$ and $\text{Var}(\bm{\xi} \mid \phi(\bm{\xi}, \mathbf{\Sigma}) = 1)$ is a continuous function of $\mathbf{\Sigma}$.
    \item[(iii)] $\phi(\bm{\mu}, \mathbf{\Sigma}) = \phi(- \bm{\mu},\mathbf{\Sigma})$, for all $\bm{\mu}$, $\mathbf{\Sigma} > 0$.
\end{enumerate}
where $\mathbf{\Sigma} > 0$ means that $\mathbf{\Sigma}$ is positive-definite.
\end{condition}
Let $\widehat{\bm{\tau}}_{\mathbf{X}} = \bar{\mathbf{X}}_T - \bar{\mathbf{X}}_C$. Then, note that Quadratic Form Rerandomization, defined as
\begin{align*}
    \phi(\sqrt{n}(\bar{\mathbf{X}}_T - \bar{\mathbf{X}}_C), \mathbf{\Sigma}) = \begin{cases}
        1 & \text{if $Q_{\mathbf{A}}(\sqrt{n} \widehat{\bm{\tau}}_{\mathbf{X}}) = (\sqrt{n} \widehat{\bm{\tau}}_{\mathbf{X}})^T \mathbf{A} (\sqrt{n} \widehat{\bm{\tau}}_{\mathbf{X}}) \leq a$} \\
        0 & \text{otherwise}
    \end{cases}
\end{align*}
satisfies \cref{general_balance_condition} under the condition that $\mathbf{\Sigma}$ is positive-definite. Let $\mathcal{G}$ denote the event that $\phi(\sqrt{n}(\bar{\mathbf{X}}_T - \bar{\mathbf{X}}_C), \mathbf{\Sigma}) = 1$ for some treatment assignment vector $\mathbf{W}$ and define $\mathfrak{G} = \{ \bm{\mu} \ : \ \phi(\bm{\mu}, \mathbf{\Sigma}) = 1 \}$ to be the acceptance region. Then, when \cref{asymptotic_norm_condition} and \cref{general_balance_condition} hold, by Corollary $A1$ of \cite{li2018asymptotic}, it follows that
\begin{align*}
    \sqrt{n}\begin{pmatrix}
        \widehat{\tau} - \tau \\
        \bar{\mathbf{X}}_T - \bar{\mathbf{X}}_C
    \end{pmatrix} \mid \mathbf{X}, \mathcal{G} \sim \begin{pmatrix}
        T_\tau \\
        \mathbf{T}_{\mathbf{X}}
    \end{pmatrix} \mid \mathfrak{G}
\end{align*}
where $(T_\tau, \mathbf{T}^T_{\mathbf{X}})^T \sim \mathcal{N}(\mathbf{0}, \mathbf{\Sigma}_{\mathbf{V}})$. In particular, we are interested in the conditional distribution $\sqrt{n}(\widehat{\tau} - \tau) \mid \mathbf{X}, Q_{\mathbf{A}}(\sqrt{n} \widehat{\bm{\tau}}_{\mathbf{X}}) \leq a$; under Corollary $A2$ of \cite{li2018asymptotic} it follows that
\begin{align} \label{tau_dist}
    \sqrt{n}(\widehat{\tau} - \tau) \mid \mathbf{X}, Q_{\mathbf{A}}(\sqrt{n} \widehat{\bm{\tau}}_{\mathbf{X}}) \leq a \sim \varepsilon + \bm{\beta}^T \bm{\xi} \mid  \bm{\xi}^T \mathbf{A} \bm{\xi} \leq a
\end{align}
where $\varepsilon \sim \mathcal{N}(0, V_{\tau \tau}(1 - R^2))$, $\bm{\xi} \sim \mathcal{N}(\mathbf{0}, \mathbf{\Sigma})$, $R^2 = \frac{1}{V_{\tau \tau }} (\mathbf{V}_{\tau x} \mathbf{\Sigma}^{-1} \mathbf{V}_{x \tau})$, and $\bm{\beta} = \mathbf{\Sigma}^{-1} \mathbf{V}_{x \tau}$. 

\section{Geometry of Quadratic Form Rerandomization} \label{geometry_section}

One interesting implication of rerandomization using quadratic forms relates to the geometry of the covariance matrix, as defined in \cref{Theorem1}. As previously discussed, Quadratic Form Rerandomization considers the set of all ellipsoidal constraints on the covariance of the covariate mean differences. However, more intuition can be built by directly solving for the eigenvalues of this covariance matrix. The following result establishes that the variance reduction factors $\nu_{1, \eta}, \ldots, \nu_{d, \eta}$ represent the generalized eigenvalues of $\text{Cov}\left(\sqrt{n}\left(\bar{\mathbf{X}}_T - \bar{\mathbf{X}}_C\right) \mid \mathbf{X},  Q_\mathbf{A}(\sqrt{n} \widehat{\bm{\tau}}_{\mathbf{X}}) \leq a\right)$ and $\mathbf{\Sigma}$.

\begin{corollary}\label{eigenvalues_qfr}
    Suppose $\mathbf{A} \in \mathbf{S}^d_{+}$. Then, under \cref{asymptotic_norm_condition} and \cref{general_balance_condition}, as $n \to \infty$, $\nu_{1, \eta}, \ldots, \nu_{d, \eta}$ are the generalized eigenvalues of $\text{Cov}\left(\sqrt{n} \widehat{\bm{\tau}}_{\mathbf{X}} \mid \mathbf{X}, Q_\mathbf{A}(\sqrt{n} \widehat{\bm{\tau}}_{\mathbf{X}}) \leq a\right)$ and $\mathbf{\Sigma}$. That is, for all $j = 1, \ldots, d$,
    \begin{align*}
    \text{det}\left(\text{Cov}\left(\sqrt{n} \widehat{\bm{\tau}}_{\mathbf{X}} \mid \mathbf{X}, Q_\mathbf{A}(\sqrt{n} \widehat{\bm{\tau}}_{\mathbf{X}}) \leq a\right) - \nu_{j, \eta} \mathbf{\Sigma} \right) = 0.
    \end{align*}
    Additionally, if $\mathbf{\Sigma}$ and $\mathbf{\Sigma}^{\nicefrac{1}{2}} \mathbf{A} \mathbf{\Sigma}^{\nicefrac{1}{2}}$ share an eigenbasis, then it can be shown that the eigenvalues of $\text{Cov}\left(\sqrt{n}\left(\bar{\mathbf{X}}_T - \bar{\mathbf{X}}_C\right) \mid \mathbf{X},  Q_\mathbf{A}(\sqrt{n} \widehat{\bm{\tau}}_{\mathbf{X}}) \leq a\right)$ are given by $\nu_{1, \eta} \lambda_1, \ldots, \nu_{d, \eta} \lambda_d$.
\end{corollary}

For generalized eigenvalue problems $\mathbf{A}\mathbf{v} = \lambda \mathbf{B} \mathbf{v}$, the generalized eigenvalue describes how the linear transformation $\mathbf{A}$ scales $\mathbf{v}$ relative to the transformation $\mathbf{B}$. Thus, \cref{eigenvalues_qfr} shows that balancing on a quadratic form scales the space defined by the transformation $\mathbf{\Sigma}$ by $\nu_{1, \eta}, \ldots, \nu_{d, \eta}$. In the case that $\mathbf{\Sigma}$ and $\mathbf{\Sigma}^{\nicefrac{1}{2}} \mathbf{A} \mathbf{\Sigma}^{\nicefrac{1}{2}}$ share an eigenbasis, this relationship simplifies, and now the eigenvalues of $\text{Cov}\left(\sqrt{n} \widehat{\bm{\tau}}_{\mathbf{X}} \mid \mathbf{X}, Q_\mathbf{A}(\sqrt{n} \widehat{\bm{\tau}}_{\mathbf{X}}) \leq a\right)$ are given by $\nu_{1, \eta} \lambda_1, \ldots, \nu_{d, \eta} \lambda_d$. In this setting, Quadratic Form Rerandomization directly manipulates the eigenstructure of $\mathbf{\Sigma}$, as illustrated in \cref{diffineigen}, where the shape of the covariance matrix afterwards depends on the choice of $\mathbf{A}$. Notably, most traditional choices of $\mathbf{A}$ satisfy the requirement that $\mathbf{\Sigma}$ and $\mathbf{\Sigma}^{\nicefrac{1}{2}} \mathbf{A} \mathbf{\Sigma}^{\nicefrac{1}{2}}$ share an eigenbasis. For example, $\mathbf{A} \in \{\mathbf{\Sigma}^{-1}, \mathbf{I}_d, (\mathbf{\Sigma} + \lambda \mathbf{I}_d)^{-1}\}$ all satisfy the shared eigenbasis requirement. $\mathbf{A}^*$, as defined in \cref{optimal_a_outcomes}, satisfies the shared eigenbasis requirement when $\bm{\beta}$ is an eigenvector of $\mathbf{\Sigma}$.

    \begin{figure}[h]
  \begin{subfigure}[b]{0.5\textwidth}
    \centering
\begin{tikzpicture}
\tikzset{>={latex}, thick}
  \begin{scope}[rotate=30]

    \draw[thick] (0,0) ellipse (4.25cm and 1.5cm); 
    
    \draw[->] (0,0) -- node[right] {$\sqrt{\lambda_1}$} ++(2,0) --(4.25,0) ;
    \draw[->] (0,0) -- node[midway, left] {$\sqrt{\lambda_2}$} (0, 1) --(0,1.5) ;
  \end{scope}
\end{tikzpicture}
    \caption{Eigenvalues before rerandomization.}
  \end{subfigure}
  \hspace{0.005\textwidth} 
  \begin{subfigure}[b]{0.5\textwidth}
    \centering
\begin{tikzpicture}
\tikzset{>={latex}, thick}

  \begin{scope}[rotate=30]

    \draw[thick] (0,0) ellipse (2.5cm and 1.5cm); 
    
    \draw[->] (0,0) -- node[right] {$\sqrt{\nu_{1, \eta}\lambda_1}$} ++(1,0) --(2.5,0) ;
    \draw[->] (0,0) -- node[midway, left] {$\sqrt{\nu_{2, \eta}\lambda_2}$} (0, 1) --(0,1.5) ;
  \end{scope}
\end{tikzpicture}
    \caption{Eigenvalues after rerandomization.}
  \end{subfigure}
  \caption{Change in eigenstructure after hypothetical Quadratic Form Rerandomization when there are two covariates. Here, the covariate matrix $\mathbf{\Sigma}$ has eigenvalues $\lambda_1$ and $\lambda_2$ before rerandomization; after Quadratic Form Rerandomization, these eigenvalues are scaled by $\nu_{j,\eta}$ in \cref{q_definition}. Because the $\nu_{j, \eta}$ are not necessarily constant, Quadratic Form Rerandomization can change not just the size but also the shape of the covariance matrix.}
  \label{diffineigen}  
\end{figure}

\subsection{Discussion of \cref{a_table} and the choice of $\mathbf{A}$}

In this section, we provide an extended discussion of \cref{a_table} and how each rerandomization method can be written as $Q_{\mathbf{A}}(\sqrt{n} \widehat{\bm{\tau}}_{\mathbf{X}}) = (\sqrt{n} \widehat{\bm{\tau}}_{\mathbf{X}})^T \mathbf{A} (\sqrt{n} \widehat{\bm{\tau}}_{\mathbf{X}})$ for some positive semi-definite or positive-definite matrix $\mathbf{A} \in \mathbb{R}^{d \times d}$ where $\widehat{\bm{\tau}}_{\mathbf{X}} = \bar{\mathbf{X}}_T - \bar{\mathbf{X}}_C$. In the case of Mahalanobis and Ridge Rerandomization, the choice of $\mathbf{A}$ is self-evident, as no manipulation is required. In the case of PCA Rerandomization, covariate balance is measured using the Mahalanobis distance between covariate means in treatment and control using the top $k$ principal components. Let $\mathbf{E}_k = (\mathbf{I}_k, \mathbf{0})^T$. Then, since $\mathbf{Z} = \mathbf{V}^T(\bar{\mathbf{X}}_T - \bar{\mathbf{X}}_C)$ where $\mathbf{V}$ is the matrix of singular vectors of $\mathbf{X}$, it follows that 
\begin{align*}
    M_k &= \sqrt{n}(\bar{\mathbf{Z}}^{(k)}_T - \bar{\mathbf{Z}}^{(k)}_C)^T \mathbf{\Sigma}^{-1}_{Z,k} \sqrt{n}(\bar{\mathbf{Z}}^{(k)}_T - \bar{\mathbf{Z}}^{(k)}_C) \\[0.05in]
    &= \left\{\sqrt{n}(\bar{\mathbf{X}}_T - \bar{\mathbf{X}}_C) \right\}^T \mathbf{V} \Big(\mathbf{E}_k \mathbf{\Sigma}^{-1}_{Z,k}  \mathbf{E}^T_k \Big) \mathbf{V}^T \left\{\sqrt{n}(\bar{\mathbf{X}}_T - \bar{\mathbf{X}}_C) \right\} \\[0.05in]
    &= \left\{\sqrt{n}(\bar{\mathbf{X}}_T - \bar{\mathbf{X}}_C) \right\}^T \mathbf{V} \begin{pmatrix}
        \mathbf{\Sigma}_{Z,k}^{-1} & \mathbf{0} \\[0.05in]
        \mathbf{0} & \mathbf{0}
    \end{pmatrix} \mathbf{V}^T \left\{\sqrt{n}(\bar{\mathbf{X}}_T - \bar{\mathbf{X}}_C) \right\} \\
    &= C^{-1}_n \left\{\sqrt{n}(\bar{\mathbf{X}}_T - \bar{\mathbf{X}}_C) \right\}^T \underbrace{\mathbf{V} \begin{pmatrix}
        \mathbf{D}^{-2}_k & \mathbf{0} \\
        \mathbf{0} & \mathbf{0}
    \end{pmatrix} \mathbf{V}^T}_{\mathbf{A}} \left\{\sqrt{n}(\bar{\mathbf{X}}_T - \bar{\mathbf{X}}_C) \right\}.
\end{align*}
For the regression-based joint test under $p$-value-based rerandomization, we assume that the covariates have been centered such that their means are zero. When this is the case, observe that
\begin{align*}
    \bar{\mathbf{X}}_T - \bar{\mathbf{X}}_C &= \frac{1}{n_1} \Big( \mathbf{X}^T \mathbf{W} \Big) - \frac{1}{n_0} \Big( \mathbf{X}^T(\mathbf{1}_n - \mathbf{W}) \Big) \\
    &= \left(\frac{1}{n_1} + \frac{1}{n_0} \right) \mathbf{X}^T \mathbf{W} - \frac{1}{n_0} \mathbf{X}^T\mathbf{1}_n \\
    &= \left(\frac{1}{n_1} + \frac{1}{n_0} \right) \mathbf{X}^T \mathbf{W} 
\end{align*}
where the final equation follows since $\frac{1}{n_0} \mathbf{X}^T\mathbf{1}_n = \mathbf{0}$ after centering. Then, we can see that after regressing the treatment vector $\mathbf{W}$ onto the covariates,
\begin{align*}
    \widehat{\bm{\beta}} &= (\mathbf{X}^T \mathbf{X})^{-1} \mathbf{X}^T \mathbf{W} \\
    &=   \left(\frac{1}{n_1} + \frac{1}{n_0} \right)^{-1}(\mathbf{X}^T \mathbf{X})^{-1} (\bar{\mathbf{X}}_T - \bar{\mathbf{X}}_C) \\
    &= \left(\frac{n}{n_0 n_1} \right)^{-1} \left(\frac{1}{(n-1)p(1-p)} \right)\mathbf{\Sigma}^{-1} (\bar{\mathbf{X}}_T - \bar{\mathbf{X}}_C) \\
    &=  \left(\frac{n_0 n_1}{n(n-1)p(1-p)} \right)\mathbf{\Sigma}^{-1} (\bar{\mathbf{X}}_T - \bar{\mathbf{X}}_C).
\end{align*}
Thus, it follows that we can write the joint test as
\begin{align*}
    \widehat{\bm{\beta}}^T \widehat{\mathbf{V}}^{-1}\widehat{\bm{\beta}} = \left(\frac{n_0 n_1}{n^{\nicefrac{3}{2}}(n-1)p(1-p)} \right)^2 (\sqrt{n} \widehat{\bm{\tau}}_{\mathbf{X}})^T \mathbf{\Sigma}^{-1} \widehat{\mathbf{V}}^{-1} \mathbf{\Sigma}^{-1} (\sqrt{n} \widehat{\bm{\tau}}_{\mathbf{X}}).
\end{align*}
Note that if logistic regression were used to estimate $\widehat{\bm{\beta}}$, asymptotically we would still obtain an expression that can be written as $(\sqrt{n} \widehat{\bm{\tau}}_{\mathbf{X}})^T \mathbf{A} (\sqrt{n} \widehat{\bm{\tau}}_{\mathbf{X}})$ --- \cite{zhao2021no} show that the joint test under logistic regression is equivalent to Mahalanobis Rerandomization. Finally, in the case of Weighted Euclidean distances (such as those suggested in \cite{lu2023design}) we can directly plug in any diagonal matrix $\text{diag}\{a_1, \ldots, a_d\}$ as our choice of $\mathbf{A}$.

\section{On the distribution of quadratic forms} \label{quad_form_dist_section}

To establish distributional properties of rerandomization schemes involving quadratic forms, it is useful to note several classical results. Suppose that $\mathbf{v} \in \mathbb{R}^d$ is a random vector such that $\mathbf{v} \sim \mathcal{N}(\bm{\mu}, \mathbf{\Sigma})$ and $\mathbf{A} \in \mathbb{R}^{d \times d}$ is a fixed, symmetric matrix. Then,
\begin{align} \label{qfr_distribution}
    \mathbf{v}^T \mathbf{A} \mathbf{v} \sim \sum^d_{j=1} \eta_j \chi^2_1(\gamma_j)
\end{align}
where $\eta_1, \ldots, \eta_d$ are the eigenvalues of $\mathbf{\Sigma}^{\nicefrac{1}{2}} \mathbf{A} \mathbf{\Sigma}^{\nicefrac{1}{2}}$ and $\chi^2_1(\gamma_j)$ for $j = 1, \ldots, d$ are independent non-central $\chi^2$ random variables where $\gamma_j = (\mathbf{\Omega}^T \mathbf{\Sigma}^{-\nicefrac{1}{2}} \bm{\mu})^2_j$ and $\mathbf{\Omega}$ is the orthogonal matrix of eigenvectors \citep{mathai1992quadratic}. In the rerandomization literature, it is often assumed that $\sqrt{n}(\bar{\mathbf{X}}_T - \bar{\mathbf{X}}_C) \mid \mathbf{X} \sim \mathcal{N}(\mathbf{0}, \mathbf{\Sigma})$ which is justified by the finite population central limit theorem \citep{li2017general}. In this case, $\gamma_j = 0$ for all $j = 1, \dots, d$. Thus, by plugging in each associated matrix $\mathbf{A}$, we can immediately obtain the distributions of the quadratic forms under Mahalanobis Rerandomization, Ridge Rerandomization, and PCA Rerandomization. For example, Lemma 4.1 of \cite{branson2021ridge} states that $M_{\lambda}$ as defined in \cref{ridge_m} is distributed as $ M_{\lambda} \mid \mathbf{X} \sim \sum_{j=1}^d \frac{\lambda_j}{\lambda_j + \lambda} \mathcal{Z}^2_j$ where $\mathcal{Z}_1, \ldots, \mathcal{Z}_d \overset{iid}{\sim} \mathcal{N}(0, 1)$. This result can be readily recovered by noting that
\begin{align*}
    \mathbf{\Sigma}^{\nicefrac{1}{2}} \mathbf{A} \mathbf{\Sigma}^{\nicefrac{1}{2}} &= \mathbf{\Sigma}^{\nicefrac{1}{2}} (\mathbf{\Sigma} + \lambda \mathbf{I}_d)^{-1} \mathbf{\Sigma}^{\nicefrac{1}{2}} \\
    &= \mathbf{\Sigma}^{\nicefrac{1}{2}} \left(\mathbf{\Sigma}^{-\nicefrac{1}{2}} (\mathbf{I}_d + \lambda \mathbf{\Sigma}^{-1})^{-1} \mathbf{\Sigma}^{-\nicefrac{1}{2}} \right) \mathbf{\Sigma}^{\nicefrac{1}{2}} \\
    &= (\mathbf{I}_d + \lambda \mathbf{\Sigma}^{-1})^{-1}.
\end{align*}
It can be shown that the eigenvalues of $(\mathbf{I}_d + \lambda \mathbf{\Sigma}^{-1})^{-1}$ are given by $\lambda_j (\lambda_j + \lambda)^{-1}$ for $j = 1, \ldots, d$, which completes the lemma. Similarly, Theorem 2 of \cite{zhang2023pca} states that \cref{mk_def_pos_df} follows the distribution $M_k \mid \mathbf{X} \sim \chi^2_k$; this can be immediately recovered by noting that $\mathbf{\Sigma}^{\nicefrac{1}{2}}_Z \mathbf{\Sigma}^{-1}_Z \mathbf{\Sigma}^{\nicefrac{1}{2}}_Z = \mathbf{I}_k$ and thus $\eta_j = 1$ for all $j = 1,\ldots,k$. Thus, \cref{qfr_distribution} can serve as a starting point for establishing properties of any rerandomization method involving a quadratic form $Q_{\mathbf{A}}(\mathbf{v})$, including methods not considered previously. To our knowledge, this classical result has not been leveraged in most rerandomization works; the only exception we know of is \cite{lu2023design} and \cite{liu2023bayesian}, who considered rerandomization using quadratic forms within the context of cluster-based experiments and Bayesian designs, respectively.

Given the distribution of $Q_{\mathbf{A}}(\mathbf{v})$, we can specify some threshold $a > 0$ by which a randomization is deemed acceptable or not. Typically, this threshold $a$ is determined by setting some acceptance probability $\alpha \in (0, 1)$, such that $\mathbb{P}(Q_{\mathbf{A}}(\mathbf{v}) \leq a) = \alpha$. When the eigenvalues $\eta_1, \ldots, \eta_d$ are not all equal and $\gamma_j = 0$, $Q_{\mathbf{A}}(\mathbf{v})$ is distributed as the summation of independent Gamma random variables with probability density function 
\begin{align*}
    f_Q(v) = \left\{\prod^r_{j=1} \left(\frac{\eta_{*}}{ \eta_j} \right)^{1/2} \right\} \sum^\infty_{k=0} \delta_k \frac{v^{(r/2)+k-1}\exp(-v/(2\eta_*))}{(2\eta_*)^{r/2+k}\Gamma(r/2+k)}
\end{align*}
for $v > 0$, where $r =\text{rank}(\mathbf{\Sigma}^{\nicefrac{1}{2}}\mathbf{A} \mathbf{\Sigma}^{\nicefrac{1}{2}})$, $\eta_*= \min_{1\leq j\leq r}\eta_j$, and $\delta_k$ satisfies
\begin{align*}
    \delta_{k} = \frac{1}{k} \sum^{k}_{\ell=1} \left[ \frac{1}{2}\sum^r_{j=1}  \left(1 - \frac{\eta_{*}}{\eta_j} \right)^{\ell}\right] \delta_{k - \ell}
\end{align*}
where $\delta_0 = 1$ \citep{moschopoulos1985distribution}. Unfortunately, this distribution is analytically intractable. Nonetheless, there are two simple options for choosing $a$. First, we can quickly determine $a$ by Monte Carlo simulation; since $Q_{\mathbf{A}}(\mathbf{v}) \sim \sum \eta_j \mathcal{Z}^2_j$ where $\mathcal{Z}_1, \ldots, \mathcal{Z}_d \sim \mathcal{N}(0, 1)$, we can simulate from this distribution many times and then define $a$ as an empirical quantile of these draws. Alternatively, we can approximate the distribution of $Q_{\mathbf{A}}(\mathbf{v})$ using an extension of the Welch–Satterthwaite method that approximates the summation of independent Gamma random variables with a single Gamma random variable \citep{stewart2007simple}. To do so, we employ what is known as the ``moment-matching method'' to match the first and second moments of our random variables. Note that for $c > 0$ it follows that $c \chi^2_\nu \sim \text{Gamma}(\nicefrac{\nu}{2}, 2 c)$. Thus, if we let $X_1, \ldots, X_n$ be a sequence of independent Gamma random variables such that $X_i \sim \text{Gamma}(\alpha_i, \beta_i)$, we can see that
\begin{align*}
    \mathbb{E}\left[\sum^n_{i=1} X_i \right] = \sum^n_{i=1} \alpha_i \beta_i \qquad \text{and} \qquad 
    \text{Var}\left[\sum^n_{i=1} X_i \right] = \sum^n_{i=1} \alpha_i \beta^2_i.
\end{align*}
Then, matching moments with a single Gamma random variable $X_m \sim \text{Gamma}(\alpha_m, \beta_m)$ we can see that
\begin{align*}
    \alpha_m = \frac{\mu^2}{\sum^n_{i=1}\alpha_i \beta^2_i} \qquad \text{and} \qquad  \beta_m = \frac{\sum^n_{i=1}\alpha_i \beta^2_i}{\mu}
\end{align*}
where $\mu = \sum^n_{i=1} \alpha_i \beta_i$ \citep{covo2014novel}. When considering the distribution of a quadratic form, i.e.\ when $\mathbf{v} \sim \mathcal{N}(\mathbf{0}, \mathbf{\Sigma})$ and $\mathbf{A} \in \mathbb{R}^{d \times d}$ is a fixed, symmetric matrix, then $\mathbf{v}^T \mathbf{A} \mathbf{v}$ approximately follows a $\text{Gamma}(\alpha_m, \beta_m)$ distribution, where
\begin{align*}
    \alpha_m = \frac{\left( \sum^d_{j=1} \eta_j\right)^2}{2 \sum^d_{j=1} \eta^2_j} \qquad \text{and} \qquad \beta_m = \frac{ 2\sum^d_{j=1} \eta^2_j}{\sum^d_{j=1} \eta_j }
\end{align*}
since $\eta_j \chi^2_1 \sim \text{Gamma}(\nicefrac{1}{2}, 2 \eta_j)$, and where the $m$ subscript denotes the moment matching method. Notably, we can derive a slightly cleaner expression for $\alpha_m$ and $\beta_m$ by noting that $\sum^d_{j=1} \eta_j$ and $\sum^d_{j=1} \eta^2_j$ are equal to the trace of $\mathbf{\Sigma}^{\nicefrac{1}{2}} \mathbf{A} \mathbf{\Sigma}^{\nicefrac{1}{2}}$ and $(\mathbf{\Sigma}^{\nicefrac{1}{2}} \mathbf{A} \mathbf{\Sigma}^{\nicefrac{1}{2}})^2$, respectively. Then, following the definitions outlined in \cite{bartlettbenign}, we can write $\alpha_m = \frac{R(\mathbf{\Sigma}^{\nicefrac{1}{2}} \mathbf{A} \mathbf{\Sigma}^{\nicefrac{1}{2}})}{2}$ and $\beta_m = \frac{2 \text{tr}(\mathbf{\Sigma}^{\nicefrac{1}{2}} \mathbf{A} \mathbf{\Sigma}^{\nicefrac{1}{2}})}{R(\mathbf{\Sigma}^{\nicefrac{1}{2}} \mathbf{A} \mathbf{\Sigma}^{\nicefrac{1}{2}})}$ where
\begin{align*}
    R(\mathbf{\Sigma}) = \frac{\text{tr}(\mathbf{\Sigma})^2}{\text{tr}(\mathbf{\Sigma}^2)} 
\end{align*}
is known as the ``effective rank'' of a covariance matrix. From here, we can define $a$ as the $\alpha$-quantile of the distribution of $\text{Gamma}(\alpha_m, \beta_m)$. Other approximations for the summation of independent Gamma random variables can be found in \cite{bodenham2016comparison}.

\section{Proofs from the Main Text} \label{proofs_section}

\subsection{Proof of \texorpdfstring{\cref{Theorem1}}{Theorem 1}} \label{theorem_1_proof}
\begin{proof}[\textbf{Proof:}]
Recall that, as discussed in \cref{asymptotic_theory_sec}, under \cref{asymptotic_norm_condition} and \cref{general_balance_condition}, as $n \to \infty$, it follows that
\begin{align*}
    \sqrt{n} \left(
        \bar{\mathbf{X}}_T - \bar{\mathbf{X}}_C \right) \mid \mathbf{X}, Q_{\mathbf{A}}(\sqrt{n} \widehat{\bm{\tau}}_{\mathbf{X}}) \leq a
    \sim \mathbf{Z} \mid Q_{\mathbf{A}}(\mathbf{Z}) \leq a
\end{align*}
where $\mathbf{Z} \sim \mathcal{N}(\mathbf{0}, \mathbf{\Sigma})$ and we use the notation $\sim$ to denote two sequences of random vectors converging weakly to the same distribution. From here, expanding out the definition of $Q_{\mathbf{A}}(\mathbf{Z})$ we can see that
\begin{align*}
    Q_{\mathbf{A}}(\mathbf{Z}) &= \mathbf{Z}^T \mathbf{A} \mathbf{Z}\\
    &=\big(\mathbf{\Sigma}^{-\nicefrac{1}{2}} \mathbf{Z} \big)^T \mathbf{\Sigma}^{\nicefrac{1}{2}} \mathbf{A} \mathbf{\Sigma}^{\nicefrac{1}{2}} \big(\mathbf{\Sigma}^{-\nicefrac{1}{2}} \mathbf{Z} \big) \\
    &= (\mathbf{\Omega}^T \widetilde{\mathbf{Z}})^T \mathbf{\Omega}^T \mathbf{\Sigma}^{\nicefrac{1}{2}} \mathbf{A} \mathbf{\Sigma}^{\nicefrac{1}{2}} \mathbf{\Omega} (\mathbf{\Omega}^T \widetilde{\mathbf{Z}}) \\
    &= (\mathbf{\Omega}^T \widetilde{\mathbf{Z}})^T \Big(\text{diag}\{\eta_1, \ldots, \eta_d \} \Big) (\mathbf{\Omega}^T \widetilde{\mathbf{Z}})
\end{align*}
where $\widetilde{\mathbf{Z}} = \mathbf{\Sigma}^{-\nicefrac{1}{2}} \mathbf{Z}$, $\mathbf{\Omega}$ is the orthogonal matrix of eigenvectors of $\mathbf{\Sigma}^{\nicefrac{1}{2}} \mathbf{A} \mathbf{\Sigma}^{\nicefrac{1}{2}}$, and $\eta_1, \ldots, \eta_d$ are the eigenvalues of $\mathbf{\Sigma}^{\nicefrac{1}{2}} \mathbf{A} \mathbf{\Sigma}^{\nicefrac{1}{2}}$. Note that we only require $\mathbf{A}$ to be a real, symmetric matrix in order to diagonalize $\mathbf{\Sigma}^{\nicefrac{1}{2}} \mathbf{A} \mathbf{\Sigma}^{\nicefrac{1}{2}}$. This follows, since any real symmetric matrix can be diagonalized, and due to the symmetry of $\mathbf{\Sigma}$ and $\mathbf{A}$, clearly $(\mathbf{\Sigma}^{\nicefrac{1}{2}} \mathbf{A} \mathbf{\Sigma}^{\nicefrac{1}{2}})^T = \mathbf{\Sigma}^{\nicefrac{1}{2}} \mathbf{A} \mathbf{\Sigma}^{\nicefrac{1}{2}}$. Then, $Q_{\mathbf{A}}(\mathbf{Z}) \sim \sum^d_{j=1} \eta_j (\mathbf{\Omega}^T \widetilde{\mathbf{Z}})^2_j$. From here, observe that
\begin{align*}
     \text{Cov}( \mathbf{Z} \mid Q_{\mathbf{A}}(\mathbf{Z}) \leq a)  &= \text{Cov}\left(\mathbf{\Sigma}^{\nicefrac{1}{2}} \left(\mathbf{\Sigma}^{-\nicefrac{1}{2}} \mathbf{Z} \right) \mid  \sum^d_{j=1} \eta_j (\mathbf{\Omega}^T \widetilde{\mathbf{Z}})^2_j \leq a\right) \\
     &= \text{Cov}\left(\mathbf{\Sigma}^{\nicefrac{1}{2}} \mathbf{\Omega} \mathbf{\Omega}^T \widetilde{\mathbf{Z}} \mid  \sum^d_{j=1} \eta_j (\mathbf{\Omega}^T \widetilde{\mathbf{Z}})^2_j \leq a\right) \\
     &= \mathbf{\Sigma}^{\nicefrac{1}{2}} \mathbf{\Omega} \text{Cov}\left( \mathbf{\Omega}^T \widetilde{\mathbf{Z}} \mid \sum^d_{j=1} \eta_j (\mathbf{\Omega}^T \widetilde{\mathbf{Z}})^2_j \leq a\right)\mathbf{\Omega}^T  \mathbf{\Sigma}^{\nicefrac{1}{2}} \\
     &= \mathbf{\Sigma}^{\nicefrac{1}{2}} \mathbf{\Omega} \text{Cov}\left( \bm{\mathcal{Z}} \mid  \sum^d_{j=1} \eta_j \mathcal{Z}^2_j \leq a\right)\mathbf{\Omega}^T  \mathbf{\Sigma}^{\nicefrac{1}{2}}.
\end{align*}
where the last equality follows by the orthogonality of $\mathbf{\Omega}$, since $\mathbf{\Omega}^T \widetilde{\mathbf{Z}} \sim \mathcal{N}(\mathbf{0}, \mathbf{\Omega}^T \mathbf{\Omega}) \sim \bm{\mathcal{Z}}$ for $\bm{\mathcal{Z}} \sim \mathcal{N}(\mathbf{0}, \mathbf{I}_d)$. From here, we need to calculate the conditional covariance of $\bm{\mathcal{Z}}$. Note that the symmetry of the normal distribution ensures that $\bm{\mathcal{Z}} \sim - \bm{\mathcal{Z}}$, which implies that
\begin{align*}
    \mathbb{E}\left[\mathcal{Z}_i \mid   \sum^d_{j=1} \eta_j \mathcal{Z}^2_j \leq a \right] = \mathbb{E}\left[-\mathcal{Z}_i \mid   \sum^d_{j=1} \eta_j \mathcal{Z}^2_j \leq a \right] = - \mathbb{E}\left[\mathcal{Z}_i \mid   \sum^d_{j=1} \eta_j \mathcal{Z}^2_j \leq a \right]
\end{align*}
for all $i = 1, \ldots, d$, which yields
\begin{align*}
     \mathbb{E}\left[\mathcal{Z}_i \mid   \sum^d_{j=1} \eta_j \mathcal{Z}^2_j \leq a \right] = 0.
\end{align*}
Therefore, the diagonal elements of $\text{Cov}(\bm{\mathcal{Z}} \mid   \sum^d_{j=1} \eta_j \mathcal{Z}^2_j \leq a)$ are given by
\begin{align*}
    \nu_{j, \eta} := \text{Var}\left(\mathcal{Z}_j \mid   \sum^d_{\ell=1} \eta_\ell \mathcal{Z}^2_\ell \leq a\right) = \mathbb{E}\left(\mathcal{Z}^2_j \mid   \sum^d_{\ell=1} \eta_\ell \mathcal{Z}^2_\ell \leq a\right).
\end{align*}
Meanwhile, for the $(\ell, m)$-elements of the covariance matrix where $\ell \neq m$, we can use the symmetry of the normal distribution to see that $\bm{\mathcal{Z}} \sim \bm{\mathcal{Z}}^*$ where $\mathcal{Z}^*_i = \mathcal{Z}_i$ for all $i \neq \ell$ and $\mathcal{Z}^*_\ell = - \mathcal{Z}_\ell$ so that
\begin{align*}
        \text{Cov}\left(\mathcal{Z}_\ell, \mathcal{Z}_m \mid   \sum^d_{j=1} \eta_j \mathcal{Z}^2_j \leq a \right) &= \text{Cov}\left(\mathcal{Z}^*_\ell, \mathcal{Z}^*_m \mid   \sum^d_{j=1} \eta_j (\mathcal{Z}^*_j)^2 \leq a \right) \\
        &= - \text{Cov}\left(\mathcal{Z}_\ell, \mathcal{Z}_m \mid   \sum^d_{j=1} \eta_j \mathcal{Z}^2_j \leq a \right)
\end{align*}
which implies
\begin{align*}
    \text{Cov}\left(\mathcal{Z}_\ell, \mathcal{Z}_m \mid   \sum^d_{j=1} \eta_j \mathcal{Z}^2_j \leq a \right) = 0
\end{align*}
for all $1 \leq \ell, m \leq d$ such that $\ell \neq m$. Putting everything together we can see that
\begin{align*}
     \text{Cov}\left(\bm{\mathcal{Z}} \mid   \sum^d_{j=1} \eta_j \mathcal{Z}^2_j \leq a\right) = \text{diag}\{(\nu_{j,\eta})_{1\le j\le d}\}.
\end{align*}
Therefore, under some regularity conditions we can see that
\begin{align*}
    \text{Cov}\left(\sqrt{n}\left(\bar{\mathbf{X}}_T - \bar{\mathbf{X}}_C\right) \mid \mathbf{X},  Q_\mathbf{A}(\sqrt{n} \widehat{\bm{\tau}}_{\mathbf{X}}) \leq a\right) &= \mathbf{\Sigma}^{\nicefrac{1}{2}} \mathbf{\Omega} \Big( \text{diag}\{(\nu_{j, \eta})_{1 \leq j \leq d} \} \Big) \mathbf{\Omega}^T \mathbf{\Sigma}^{\nicefrac{1}{2}}.
\end{align*}
Finally, we note that if $\mathbf{A}$ is positive semi-definite then $\mathbf{\Sigma}^{\nicefrac{1}{2}} \mathbf{A}\mathbf{\Sigma}^{\nicefrac{1}{2}}$ is positive semi-definite as well. This follows since for any $\mathbf{v} \in \mathbb{R}^d$, we have that $\mathbf{v}^T \mathbf{\Sigma}^{\nicefrac{1}{2}} \mathbf{A}\mathbf{\Sigma}^{\nicefrac{1}{2}} \mathbf{v} = \mathbf{u}^T \mathbf{A} \mathbf{u} \geq 0$ for $\mathbf{u} = \mathbf{\Sigma}^{\nicefrac{1}{2}} \mathbf{v}$. The same argument holds for positive definiteness. Therefore, when $\mathbf{A}$ is positive semi-definite it may be the case that some of the eigenvalues $\eta_1, \ldots, \eta_d$ are zero. Suppose that $\eta_k = 0$. Then, it follows that
\begin{align*}
    \nu_{k, \eta} = \mathbb{E}\left[\mathcal{Z}^2_k \mid   \sum^d_{j=1} \eta_j \mathcal{Z}^2_j \leq a\right] = \mathbb{E}\left[\mathcal{Z}^2_k \mid   \sum^d_{j \neq k} \eta_j \mathcal{Z}^2_j \leq a\right] = \mathbb{E}\left[\mathcal{Z}^2_k\right] = 1
\end{align*}
since $\mathcal{Z}_1, \ldots, \mathcal{Z}_d \sim \mathcal{N}(0, 1)$. Additionally, for positive semi-definite or positive-definite $\mathbf{A}$ it will be the case that $\nu_{j, \eta} \leq 1$ for all $j = 1, \ldots, d$. This follows since under the positive semi-definiteness assumption $\eta_1 \geq \cdots \geq \eta_d \geq 0$. Thus, by Lemma 4.2 of \cite{branson2021ridge} it follows that $\nu_{1, \eta}  \leq \cdots \leq \nu_{d, \eta} \leq 1$.
\end{proof}

\subsection{Proof of \texorpdfstring{\cref{opnorm}}{Theorem 2}} \label{opnorm_proof}
\begin{proof}[\textbf{Proof:}]
First, for notational convenience, let us define $\mathbf{\Psi} = \text{diag}\{(\nu_{j, \eta})_{1 \leq j \leq d} \}$. To begin, we work with the square root of the covariance under Quadratic Form Rerandomization,
\begin{align*}
    \text{Cov}\left(\sqrt{n}\widehat{\bm{\tau}}_{\mathbf{X}} \mid \mathbf{X}, Q_{\mathbf{A}}(\sqrt{n} \widehat{\bm{\tau}}_{\mathbf{X}}) \leq a\right)^{\nicefrac{1}{2}},
\end{align*}
as this will greatly simplify the notation used in the proof since for some $\mathbf{M} \in \mathbb{R}^{d \times d}$ the Frobenius norm can be written as $|| \mathbf{M} ||^2_F = \text{tr}(\mathbf{M}^T \mathbf{M})$. If we instead evaluate $\mathbf{M}^{\nicefrac{1}{2}}$, the Frobenius norm simplifies to  $||\mathbf{M}^{\nicefrac{1}{2}} ||^2_F = \text{tr}((\mathbf{M}^{\nicefrac{1}{2}})^T (\mathbf{M}^{\nicefrac{1}{2}})) = \text{tr}(\mathbf{M})$ when $\mathbf{M}$ is symmetric. Thus, under \cref{asymptotic_norm_condition} and \cref{general_balance_condition}, and as $n \to \infty$,
\begin{align*}
     \left|\left|\text{Cov}\left(\sqrt{n}\widehat{\bm{\tau}}_{\mathbf{X}} \mid \mathbf{X}, Q_{\mathbf{A}}(\sqrt{n} \widehat{\bm{\tau}}_{\mathbf{X}}) \leq a\right)^{\nicefrac{1}{2}}\right|\right|^2_F &= \text{tr}\left(\mathbf{\Sigma}^{\nicefrac{1}{2}} \mathbf{\Omega} \mathbf{\Psi} \mathbf{\Omega}^T \mathbf{\Sigma}^{\nicefrac{1}{2}}  \right)
\end{align*}
where $\mathbf{\Omega}$ is the orthogonal matrix of eigenvectors of $\mathbf{\Sigma}^{\nicefrac{1}{2}} \mathbf{A} \mathbf{\Sigma}^{\nicefrac{1}{2}}$. Then, using the fact that $\mathbf{\Gamma}^T \mathbf{\Sigma} \mathbf{\Gamma} = \mathbf{\Lambda}$, we can simplify our expression to
\begin{align*}
    \text{tr}\left(\mathbf{\Sigma}^{\nicefrac{1}{2}} \mathbf{\Omega} \mathbf{\Psi} \mathbf{\Omega}^T \mathbf{\Sigma}^{\nicefrac{1}{2}}  \right) &= \text{tr}\left(\mathbf{\Gamma} \mathbf{\Lambda}^{\nicefrac{1}{2}} \mathbf{\Gamma}^T \mathbf{\Omega} \mathbf{\Psi} \mathbf{\Omega}^T \mathbf{\Gamma} \mathbf{\Lambda}^{\nicefrac{1}{2}} \mathbf{\Gamma}^T  \right) \\
    &= \text{tr}\left( \mathbf{\Lambda} \mathbf{\Gamma}^T \mathbf{\Omega} \mathbf{\Psi} \mathbf{\Omega}^T \mathbf{\Gamma}  \right) \\
    &= \text{tr}\left(  \mathbf{\Omega}^T \mathbf{\Gamma} \mathbf{\Lambda} \mathbf{\Gamma}^T \mathbf{\Omega} \mathbf{\Psi}   \right)
\end{align*}
where the second and third equalities follow by the cyclic property of the trace. Then, by the AM-GM inequality it follows that
\begin{align*}
    \text{tr}\left(  \mathbf{\Omega}^T \mathbf{\Gamma} \mathbf{\Lambda} \mathbf{\Gamma}^T \mathbf{\Omega} \mathbf{\Psi}   \right) \geq d \cdot \text{det}\left(  \mathbf{\Omega}^T \mathbf{\Gamma} \mathbf{\Lambda} \mathbf{\Gamma}^T \mathbf{\Omega} \mathbf{\Psi}   \right)^{\nicefrac{1}{d}},
\end{align*}
where equality holds if and only if $\zeta_1 = \cdots = \zeta_d$ where $\zeta_1, \ldots, \zeta_d$ are the eigenvalues of $\mathbf{\Omega}^T \mathbf{\Gamma} \mathbf{\Lambda} \mathbf{\Gamma}^T \mathbf{\Omega} \mathbf{\Psi}$. Thus, our goal is to choose $\mathbf{A}$ (which has the effect of determining $\mathbf{\Psi}$) such that we obtain the lower bound.

From here, we make use of the approximation of $\nu_{j, \eta}$ derived in \cite{lu2023design} and discussed in \cref{q_definition_main_text}. That is, for positive-definite $\mathbf{A}$,
\begin{align*}
    \nu_{j, \eta} = \frac{p_d}{\eta_j} \text{det}(\mathbf{\Sigma}^{\nicefrac{1}{2}}\mathbf{A}\mathbf{\Sigma}^{\nicefrac{1}{2}})^{1/d} \alpha^{2/d} + o(\alpha^{2/d}).
\end{align*}
 Without loss of generality, suppose that $\text{det}(\mathbf{\Sigma}^{\nicefrac{1}{2}} \mathbf{A} \mathbf{\Sigma}^{\nicefrac{1}{2}}) = \prod^d_{j=1} \eta_j = 1$. This follows, since Quadratic Form Rerandomization is invariant to scaling, in the sense that for some scalar $\omega > 0$,
\begin{align*}
    \text{Cov}\left(\sqrt{n}\widehat{\bm{\tau}}_{\mathbf{X}} \mid \mathbf{X}, Q_{\mathbf{A}}(\sqrt{n} \widehat{\bm{\tau}}_{\mathbf{X}}) \leq a\right) &= \text{Cov}\left(\sqrt{n}\widehat{\bm{\tau}}_{\mathbf{X}} \mid \mathbf{X}, \omega Q_{\mathbf{A}}(\sqrt{n} \widehat{\bm{\tau}}_{\mathbf{X}}) \leq \omega a\right) \\
    &:= \text{Cov}\left(\sqrt{n}\widehat{\bm{\tau}}_{\mathbf{X}} \mid \mathbf{X}, \omega Q_{\mathbf{A}}(\sqrt{n} \widehat{\bm{\tau}}_{\mathbf{X}}) \leq a^{\prime}\right),
\end{align*}
so scaling $Q_{\mathbf{A}}(\sqrt{n} \widehat{\bm{\tau}}_{\mathbf{X}})$ only shifts the threshold selected for determining an acceptable randomization. Thus, we can simply re-scale each quadratic form by whatever constant makes their respective determinants one. Then, ignoring the remainder for the time being, we can see that $\mathbf{\Psi} \propto \bm{\eta}^{-1}$, since 
\begin{align*}
    \mathbf{\Psi} = \begin{pmatrix}
        \nu_{1, \eta} & & 0\\[-0.1in]
        & \ddots & \\[-0.1in]
        0 & & \nu_{d, \eta}
    \end{pmatrix} = p_d \alpha^{\nicefrac{2}{d}}\begin{pmatrix}
        1 / \eta_1 & & 0\\[-0.1in]
        & \ddots & \\[-0.1in]
        0 & & 1/ \eta_d
    \end{pmatrix} = p_d \alpha^{\nicefrac{2}{d}} \bm{\eta}^{-1}.
\end{align*}
Thus, we can see that our inequality simplifies to
\begin{align*}
    p_d \alpha^{\nicefrac{2}{d}} \text{tr}\left(  \mathbf{\Omega}^T \mathbf{\Gamma} \mathbf{\Lambda} \mathbf{\Gamma}^T \mathbf{\Omega} \bm{\eta}^{-1}   \right) \geq   p_d \alpha^{\nicefrac{2}{d}} d\cdot \text{det}\left(  \mathbf{\Omega}^T \mathbf{\Gamma} \mathbf{\Lambda} \mathbf{\Gamma}^T \mathbf{\Omega} \bm{\eta}^{-1}   \right)^{\nicefrac{1}{d}}.
\end{align*}
From here, we show that when $\mathbf{A} = \mathbf{I}_d$, we achieve the lower bound. To see this, we first note that now $\mathbf{\Sigma}^{\nicefrac{1}{2}} \mathbf{A} \mathbf{\Sigma}^{\nicefrac{1}{2}} = \mathbf{\Sigma}$, and thus $\bm{\eta} = \mathbf{\Lambda}$. Furthermore, since $\mathbf{\Omega}$ diagonalizes $\mathbf{\Sigma}^{\nicefrac{1}{2}} \mathbf{A} \mathbf{\Sigma}^{\nicefrac{1}{2}}$, it now follows that $\mathbf{\Omega} = \mathbf{\Gamma}$. Therefore,
\begin{align*}
    \mathbf{\Omega}^T \mathbf{\Gamma} \mathbf{\Lambda} \mathbf{\Gamma}^T \mathbf{\Omega} \bm{\eta}^{-1} = \mathbf{\Gamma}^T \mathbf{\Gamma} \mathbf{\Lambda} \mathbf{\Gamma}^T \mathbf{\Gamma} \bm{\Lambda}^{-1} = \mathbf{I}_d
\end{align*}
due to the orthogonality of $\mathbf{\Gamma}$. Clearly, we have achieved the lower bound since $\zeta_1 = \cdots = \zeta_d$. Finally, to account for the remainder term we note that for any $\mathbf{A}$, in summation form we may write
\begin{align*}
    \left|\left|\text{Cov}\left(\sqrt{n}\widehat{\bm{\tau}}_{\mathbf{X}} \mid \mathbf{X}, Q_{\mathbf{A}}(\sqrt{n} \widehat{\bm{\tau}}_{\mathbf{X}}) \leq a\right)^{\frac{1}{2}}\right|\right|^2_F = \text{tr}\left(  \mathbf{\Omega}^T \mathbf{\Gamma} \mathbf{\Lambda} \mathbf{\Gamma}^T \mathbf{\Omega} \mathbf{\Psi}   \right) = \sum^d_{i=1} \lambda_i \left( \sum^d_{k = 1}  \mathbf{P}^2_{i k } \nu_{k, \eta} \right)
\end{align*}
where we define $\mathbf{P} = \mathbf{\Omega}^T \mathbf{\Gamma}$. Therefore, we can now see that when we set $\mathbf{A} = \mathbf{I}_d$, for some $a^\prime$ such that the acceptance probability is also $\alpha$,
\begin{align*}
        \left|\left|\text{Cov}\left(\sqrt{n}\widehat{\bm{\tau}}_{\mathbf{X}} \mid \mathbf{X}, Q_{\mathbf{I}_d}(\sqrt{n} \widehat{\bm{\tau}}_{\mathbf{X}}) \leq a^\prime \right)^{\nicefrac{1}{2}}\right|\right|^2_F &= \sum^d_{i=1} \lambda_i \left( \sum^d_{k = 1}  \mathbf{P}^2_{i k } \nu_{k, \lambda} \right) \\
        &\overset{(i)}{=}  p_d \alpha^{2/d}  \sum^d_{i=1} \lambda_i \left( \sum^d_{k = 1}  \frac{\mathbf{P}^2_{i k }}{\lambda_i} \right)  + o(\alpha^{2/d}) \\
        &\overset{(ii)}{=} d\left( p_d \alpha^{2/d} \right)  + o(\alpha^{2/d}) \\
        &\overset{(iii)}{\leq} p_d \alpha^{2/d}  \sum^d_{i=1} \lambda_i \left( \sum^d_{k = 1}  \frac{\mathbf{P}^2_{i k }}{\eta_k} \right)  + o(\alpha^{2/d}).
\end{align*}
where $(i)$ follows by plugging in the approximation for $\nu_{k, \lambda}$ for all $k = 1, \ldots, d$ and $(ii)$ follows since we have shown that when $\mathbf{A} = \mathbf{I}_d$,
\begin{align*}
    \sum^d_{i=1} \lambda_i \left( \sum^d_{k = 1}  \frac{\mathbf{P}^2_{i k }}{\lambda_i} \right) = \text{tr}\left(  \mathbf{\Gamma}^T \mathbf{\Gamma} \mathbf{\Lambda} \mathbf{\Gamma}^T \mathbf{\Gamma} \bm{\Lambda}^{-1}   \right) = \text{tr}(\mathbf{I}_d) = d.
\end{align*}
Finally, $(iii)$ follows since we have shown that for any other choice of positive-definite $\mathbf{A}$,
\begin{align*}
    d\left( p_d \alpha^{2/d} \right) \leq p_d \alpha^{\nicefrac{2}{d}} \text{tr}\left(  \mathbf{\Omega}^T \mathbf{\Gamma} \mathbf{\Lambda} \mathbf{\Gamma}^T \mathbf{\Omega} \bm{\eta}^{-1}   \right) =  p_d \alpha^{2/d}  \sum^d_{i=1} \lambda_i \left( \sum^d_{k = 1}  \frac{\mathbf{P}^2_{i k }}{\eta_k} \right).
\end{align*}
From here, we recognize that for any choice of positive-definite  $\mathbf{A}$,
\begin{align*}
    p_d \alpha^{2/d}  \sum^d_{i=1} \lambda_i \left( \sum^d_{k = 1}  \frac{\mathbf{P}^2_{i k }}{\eta_k} \right) + o(\alpha^{2/d}) = \left|\left|\text{Cov}\left(\sqrt{n}\widehat{\bm{\tau}}_{\mathbf{X}} \mid \mathbf{X}, Q_{\mathbf{A}}(\sqrt{n} \widehat{\bm{\tau}}_{\mathbf{X}}) \leq a\right)^{\nicefrac{1}{2}}\right|\right|^2_F.
\end{align*}
Note that the summation of finitely many $o(\alpha^{2/d})$ terms is still $o(\alpha^{2/d})$. Thus, we may aggregate remainder terms to see that
\begin{align*}
    \left(p_d \alpha^{2/d}  \sum^d_{i=1} \lambda_i \left( \sum^d_{k = 1}  \frac{\mathbf{P}^2_{i k }}{\eta_k} \right) \right) + o(\alpha^{2/d}) &=  \left(p_d \alpha^{2/d}  \sum^d_{i=1} \lambda_i \left( \sum^d_{k = 1}  \frac{\mathbf{P}^2_{i k }}{\eta_k} \right) + o(\alpha^{2/d}) \right) + o(\alpha^{2/d}) \\
    &= \left|\left|\text{Cov}\left(\sqrt{n}\widehat{\bm{\tau}}_{\mathbf{X}} \mid \mathbf{X}, Q_{\mathbf{A}}(\sqrt{n} \widehat{\bm{\tau}}_{\mathbf{X}}) \leq a\right)^{\nicefrac{1}{2}}\right|\right|^2_F  + o(\alpha^{2/d}).
\end{align*}
Putting everything together, it follows that for all positive-definite $\mathbf{A}$,
    \begin{align*}
        ||\text{Cov}(\sqrt{n} \widehat{\bm{\tau}}_{\mathbf{X}} \mid \mathbf{X}, Q_{\mathbf{I}_d}(\sqrt{n} \widehat{\bm{\tau}}_{\mathbf{X}}) \leq a^\prime)^{\nicefrac{1}{2}}||^2_{F} \leq ||\text{Cov}(\sqrt{n} \widehat{\bm{\tau}}_{\mathbf{X}} \mid \mathbf{X}, Q_{\mathbf{A}}(\sqrt{n} \widehat{\bm{\tau}}_{\mathbf{X}}) \leq a)^{\nicefrac{1}{2}}||^2_{F} + o(\alpha^{\nicefrac{2}{d}})
    \end{align*}
where the inequality holds up to the size of the remainder for a sufficiently small $\alpha$. From here, we can easily extend this result to $||\text{Cov}\left(\sqrt{n}\widehat{\bm{\tau}}_{\mathbf{X}} \mid \mathbf{X}, Q_{\mathbf{A}}(\sqrt{n} \widehat{\bm{\tau}}_{\mathbf{X}}) \leq a\right)||_F$. Because
\begin{align*}
    \text{Cov}\left(\sqrt{n}\widehat{\bm{\tau}}_{\mathbf{X}} \mid \mathbf{X}, Q_{\mathbf{A}}(\sqrt{n} \widehat{\bm{\tau}}_{\mathbf{X}}) \leq a\right)
\end{align*}
is positive semi-definite, it follows that
\begin{align*}
    ||\text{Cov}(\sqrt{n}\widehat{\bm{\tau}}_{\mathbf{X}} \mid \mathbf{X}, Q_{\mathbf{A}}(\sqrt{n} \widehat{\bm{\tau}}_{\mathbf{X}}) \leq a)||_F \geq \frac{
    \text{tr}(\text{Cov}\left(\sqrt{n}\widehat{\bm{\tau}}_{\mathbf{X}} \mid \mathbf{X}, Q_{\mathbf{A}}(\sqrt{n} \widehat{\bm{\tau}}_{\mathbf{X}}) \leq a\right)) }{\sqrt{d}}.
\end{align*}
Consequently,
\begin{align*}
    ||\text{Cov}(\sqrt{n}\widehat{\bm{\tau}}_{\mathbf{X}} \mid \mathbf{X}, Q_{\mathbf{A}}(\sqrt{n} \widehat{\bm{\tau}}_{\mathbf{X}}) \leq a)||_F \geq \sqrt{d} \,  p_d \text{det}(\mathbf{\Sigma})^{1/d} \alpha^{\nicefrac{2}{d}} + o(\alpha^{\nicefrac{2}{d}}).
\end{align*}
Then, because under Euclidean Rerandomization, $\mathbf{\Omega}=\mathbf{\Gamma}$ and $\eta_j=\lambda_j$, it follows that 
\begin{align*}
    \text{Cov}(\sqrt{n}\widehat{\bm{\tau}}_{\mathbf{X}} \mid \mathbf{X}, Q_{\mathbf{I}_d}(\sqrt{n} \widehat{\bm{\tau}}_{\mathbf{X}}) \leq a^\prime) &= \mathbf{\Sigma}^{\nicefrac{1}{2}} \mathbf{\Gamma}
    \text{diag}\{(\nu_{j,\lambda})_{1\leq j\leq d}\} \mathbf{\Gamma}^T \mathbf{\Sigma}^{\nicefrac{1}{2}} \\
    &= p_d \operatorname{det}(\mathbf{\Sigma})^{1/d} \alpha^{\nicefrac{2}{d}} \mathbf{I}_d + o(\alpha^{\nicefrac{2}{d}})
\end{align*}
and therefore, $||\text{Cov}\left(\sqrt{n}\widehat{\bm{\tau}}_{\mathbf{X}} \mid \mathbf{X}, Q_{\mathbf{I}_d}(\sqrt{n} \widehat{\bm{\tau}}_{\mathbf{X}}) \leq a^\prime\right)||_F = \sqrt{d}\, p_d \text{det}(\mathbf{\Sigma})^{1/d} \alpha^{\nicefrac{2}{d}} + o(\alpha^{\nicefrac{2}{d}})$. Thus, putting everything together, for every positive-definite $\mathbf{A}$,
    \begin{align*}
        ||\text{Cov}(\sqrt{n} \widehat{\bm{\tau}}_{\mathbf{X}} \mid \mathbf{X}, Q_{\mathbf{I}_d}(\sqrt{n} \widehat{\bm{\tau}}_{\mathbf{X}}) \leq a^\prime)||_{F} \leq ||\text{Cov}(\sqrt{n} \widehat{\bm{\tau}}_{\mathbf{X}} \mid \mathbf{X}, Q_{\mathbf{A}}(\sqrt{n} \widehat{\bm{\tau}}_{\mathbf{X}}) \leq a)||_{F} + o(\alpha^{\nicefrac{2}{d}}).
    \end{align*}

\noindent \textbf{Remark:} In some sense, we have implicitly used the assumption that $\mathbf{A}$ is positive-definite by applying the approximation $\nu_{j, \eta} = (p_d / \eta_j) \text{det}(\mathbf{\Sigma}^{\nicefrac{1}{2}}\mathbf{A}\mathbf{\Sigma}^{\nicefrac{1}{2}})^{1/d} \alpha^{2/d} + o(\alpha^{2/d})$ for all $j = 1, \ldots, d$ since for positive semi-definite $\mathbf{A}$, some eigenvalues will be zero (and thus the ratio $1 / \eta_j$ will be undefined). However, this result holds across all positive semi-definite $\mathbf{A}$ as well. To see this, let $\mathcal C_{\mathbf{A}} = \text{Cov}(\sqrt n\widehat{\bm{\tau}}_{\mathbf{X}} \mid \mathbf{X},
Q_{\mathbf{A}}(\sqrt{n} \widehat{\bm{\tau}}_{\mathbf{X}}) \leq a)$. Suppose that $\mathbf{A} \neq \mathbf{0}$ has rank $k < d$. Then, observe that we may write 
\begin{align*}
    \mathbf{\Sigma}^{\nicefrac{1}{2}} \mathbf{A} \mathbf{\Sigma}^{\nicefrac{1}{2}} = \mathbf{\Omega} \begin{pmatrix} \text{diag}(\eta_1,\ldots,\eta_k) & 0\\
    0 & \mathbf{0}_{d-k}
    \end{pmatrix}
\mathbf{\Omega}^T,
\end{align*}
so we may create the partition $\mathbf{\Omega}=(\mathbf{\Omega}_1,\mathbf{\Omega}_0)$, where $\mathbf{\Omega}_0\in\mathbb{R} ^{d\times(d-k)}$ spans the nullspace of $\mathbf{\Sigma}^{\nicefrac{1}{2}} \mathbf{A} \mathbf{\Sigma}^{\nicefrac{1}{2}}$. Since the acceptance event depends only on the first $k$ coordinates, for the remaining $d - k$ coordinates it follows that $\nu_{j,\eta}=1$. Thus, by \cref{Theorem1}, we know that
\begin{align*}
    \mathcal{C}_{\mathbf{A}} = \mathbf{\Sigma}^{\nicefrac{1}{2}}\mathbf{\Omega}_1 \Big( \text{diag}(\nu_{1,\eta},\ldots,\nu_{k,\eta}) \Big) \mathbf{\Omega}_1^T\mathbf{\Sigma}^{\nicefrac{1}{2}} + \mathbf{\Sigma}^{\nicefrac{1}{2}}\mathbf{\Omega}_0\mathbf{\Omega}_0^T \mathbf{\Sigma}^{\nicefrac{1}{2}}.
\end{align*}
From here, note the first term is positive semi-definite and therefore its trace is non-negative. Therefore, we can focus on the remainder term to see that
\begin{align*}
    \text{tr}(\mathcal C_{\mathbf{A}}) \geq \text{tr}\left( \mathbf{\Sigma}^{\nicefrac{1}{2}}\mathbf{\Omega}_0\mathbf{\Omega}_0^T \mathbf{\Sigma}^{\nicefrac{1}{2}} \right) = \text{tr}\left(\mathbf{\Omega}_0^T\mathbf{\Sigma}\mathbf{\Omega}_0\right) \geq (d-k)\lambda_{\min}(\mathbf{\Sigma})
\end{align*}
where we say $\lambda_{\min}(\mathbf{\Sigma})$ is the minimum eigenvalue of $\mathbf{\Sigma}$. Consequently,
\begin{align*}
    ||\mathcal{C}_{\mathbf{A}}||_F \geq \frac{\text{tr}(\mathcal{C}_{\mathbf{A}})}{\sqrt{d}} \geq \frac{(d-k)\lambda_{\min}(\mathbf{\Sigma})}{\sqrt{d}}
\end{align*}
Thus every singular nonzero $\mathbf{A} \succeq 0$ leaves a nonvanishing amount of covariance, uniformly in $\alpha$. On the other hand, under Euclidean Rerandomization,
\begin{align*}
    ||\mathcal{C}_{\mathbf{I}_d}||_F = \sqrt{d} \,p_d\text{det}(\mathbf{\Sigma})^{1/d}\alpha^{2/d} + o(\alpha^{2/d}) = o(1).
\end{align*}
Consequently, for all sufficiently small $\alpha$, $||\mathcal{C}_{\mathbf{I}_d}||_F \leq ||\mathcal{C}_{\mathbf{A}}||_F$ for every singular nonzero $\mathbf{A} \succeq 0$. Combining this with the
positive-definite case proves the result over $\mathbf{S}^d_{+}$. 

\noindent \textbf{Remark:} Although we show that the Frobenius norm of $\text{Cov}\left(\sqrt{n}\widehat{\bm{\tau}}_{\mathbf{X}} \mid \mathbf{X}, Q_{\mathbf{A}}(\sqrt{n} \widehat{\bm{\tau}}_{\mathbf{X}}) \leq a\right)$ is minimized when $\mathbf{A} = \mathbf{I}_d$, this result actually holds for all $\mathbf{A} \propto \mathbf{I}_d$. To see this, we must be careful, since in our proof we rescale $Q_{\mathbf{A}}(\sqrt{n} \widehat{\bm{\tau}}_{\mathbf{X}})$ such that $\text{det}(\mathbf{\Sigma}^{\nicefrac{1}{2}} \mathbf{A} \mathbf{\Sigma}^{\nicefrac{1}{2}}) = 1$.

Let $c_1 = \text{det}( \mathbf{\Sigma}^{\nicefrac{1}{2}}\mathbf{A} \mathbf{\Sigma}^{\nicefrac{1}{2}})^{-1/d}$, let $c_2 > 0$ be some scalar, and let $c = c_1 c_2$. Furthermore, recall that the eigenvalues of $\mathbf{\Sigma}^{\nicefrac{1}{2}} (c\mathbf{A}) \mathbf{\Sigma}^{\nicefrac{1}{2}}$ are given by $c \eta_1, \ldots, c\eta_d$. Then, we can see that the squared Frobenius norm of $\text{Cov}\left(\sqrt{n}\widehat{\bm{\tau}}_{\mathbf{X}} \mid \mathbf{X}, c Q_{\mathbf{A}}(\sqrt{n} \widehat{\bm{\tau}}_{\mathbf{X}}) \leq a\right)^{\nicefrac{1}{2}}$ is given by
\begin{align*}
     \text{tr}\left(  \mathbf{\Omega}^T \mathbf{\Gamma} \mathbf{\Lambda} \mathbf{\Gamma}^T \mathbf{\Omega} \mathbf{\Psi}   \right) &= \left(p_d \alpha^{\nicefrac{2}{d}}\text{det}(\mathbf{\Sigma}^{\nicefrac{1}{2}}(c\mathbf{A})\mathbf{\Sigma}^{\nicefrac{1}{2}})^{1/d} \right) \text{tr}\left(  \mathbf{\Omega}^T \mathbf{\Gamma} \mathbf{\Lambda} \mathbf{\Gamma}^T \mathbf{\Omega} (c \bm{\eta})^{-1} \right) \\
    &= \left(p_d \alpha^{\nicefrac{2}{d}} c_2 \right) \text{tr}\left(  \mathbf{\Omega}^T \mathbf{\Gamma} \mathbf{\Lambda} \mathbf{\Gamma}^T \mathbf{\Omega}  (c \bm{\eta})^{-1}  \right) \\
    &\geq   \left(p_d \alpha^{\nicefrac{2}{d}} c_2 d \right)  \text{det}\left(  \mathbf{\Omega}^T \mathbf{\Gamma} \mathbf{\Lambda} \mathbf{\Gamma}^T \mathbf{\Omega} (c \bm{\eta})^{-1}   \right)^{\nicefrac{1}{d}}
\end{align*}
since $\text{det}(\mathbf{\Sigma}^{\nicefrac{1}{2}} (c\mathbf{A}) \mathbf{\Sigma}^{\nicefrac{1}{2}})^{1/d} = c_2$. Then, we can see that plugging in $\mathbf{A} = \mathbf{I}_d$, it follows that $c =  \text{det}(\mathbf{\Sigma})^{-1/d} c_2$. Consequently, 
\begin{align*}
    \left(p_d \alpha^{\nicefrac{2}{d}} c_2 \right) \text{tr}\left(  \mathbf{\Omega}^T \mathbf{\Gamma} \mathbf{\Lambda} \mathbf{\Gamma}^T \mathbf{\Omega}  (c \bm{\Lambda})^{-1}  \right) =  \frac{p_d \alpha^{\nicefrac{2}{d}} c_2 d}{c} = p_d \alpha^{\nicefrac{2}{d}} d \, \text{det}(\mathbf{\Sigma})^{1/d}
\end{align*}
and 
\begin{align*}
    \left(p_d \alpha^{\nicefrac{2}{d}} c_2 d \right)  \text{det}\left(  \mathbf{\Omega}^T \mathbf{\Gamma} \mathbf{\Lambda} \mathbf{\Gamma}^T \mathbf{\Omega} (c \mathbf{\Lambda})^{-1}   \right)^{\nicefrac{1}{d}} = \frac{p_d \alpha^{\nicefrac{2}{d}} c_2 d}{c} =  p_d \alpha^{\nicefrac{2}{d}} d \, \text{det}(\mathbf{\Sigma})^{1/d}.
\end{align*}
Thus, we achieve the lower bound even under arbitrary scaling of $\mathbf{I}_d$.

\end{proof}

\subsection{Proof of \texorpdfstring{\cref{totalvarredux}}{Theorem 3}} \label{totalvarredux_proof}
\begin{proof}[\textbf{Proof:}]
To begin, note that maximizing $\sum^d_{j=1}(1 - \nu_{j, \eta})$ is equivalent to minimizing $\sum^d_{j=1} \nu_{j, \eta}$. Then, the remainder of the proof will follow similarly to that of \cref{opnorm}. Under \cref{asymptotic_norm_condition} and \cref{general_balance_condition}, and as $n \to \infty$, by the AM-GM inequality it follows that
\begin{align*}
    \sum^d_{j=1} \nu_{j, \eta} = \text{tr}(\mathbf{\Psi})  \geq d  \cdot \text{det}(\mathbf{\Psi})^{\nicefrac{1}{d}}
\end{align*}
where equality holds if and only if the eigenvalues of $\mathbf{\Psi}$ are all equal. Next, recall that (up to a remainder term), after applying the approximation derived in \cite{lu2023design} it follows that when $\mathbf{A}$ is positive-definite, $\mathbf{\Psi} \propto \bm{\eta}^{-1}$, where without loss of generality we suppose that $\text{det}(\mathbf{\Sigma}^{\nicefrac{1}{2}}\mathbf{A}\mathbf{\Sigma}^{\nicefrac{1}{2}}) = 1$. Then, we have that
\begin{align*}
     \text{tr}(\mathbf{\Psi}) =  p_d \alpha^{\nicefrac{2}{d}} \text{tr}\left( \bm{\eta}^{-1}   \right) \geq p_d \alpha^{\nicefrac{2}{d}} d \cdot \text{det}(\bm{\eta}^{-1})^{\nicefrac{1}{d}}
\end{align*}
where equality holds if and only if $1/\eta_1 = \cdots = 1 / \eta_d$. Clearly, we must choose $\mathbf{A}$ such that $\bm{\eta}^{-1} = \mathbf{I}_d$. This is satisfied when $\mathbf{A} = \mathbf{\Sigma}^{-1}$, since then $\mathbf{\Sigma}^{\nicefrac{1}{2}}\mathbf{A}\mathbf{\Sigma}^{\nicefrac{1}{2}} = \mathbf{\Sigma}^{\nicefrac{1}{2}} \mathbf{\Sigma}^{-1}\mathbf{\Sigma}^{\nicefrac{1}{2}} = \mathbf{I}_d$. Then, to account for the remainder term we can see that when $\mathbf{A} = \mathbf{\Sigma}^{-1}$,
    \begin{align*}
        \sum^d_{j=1} v_a &= d p_d \alpha^{2/d} + o(\alpha^{2/d}) \leq  p_d \alpha^{2/d} \sum^d_{j=1} \frac{1}{\eta_j} + o(\alpha^{2/d})
    \end{align*}
where again (following the steps outlined in \cref{opnorm}) we aggregate remainder terms and the inequality holds up to the remainder, for a sufficiently small $\alpha$. Therefore, for all positive-definite $\mathbf{A}$,
    \begin{align*}
         \sum^d_{j=1} v_a  \leq \sum^d_{j=1} \nu_{j, \eta} + o(\alpha^{2/d}).
    \end{align*}
Note that following the remarks in \cref{opnorm}, this result can be extended to hold across all positive semi-definite $\mathbf{A} \in \mathbb{R}^{d \times d}$, and for all $\mathbf{A} \propto \mathbf{\Sigma}^{-1}$.
\end{proof}
\subsection{Proof of \texorpdfstring{\cref{diff_in_vars}}{Theorem 4}} \label{diff_in_vars_proof}
\begin{proof}[\textbf{Proof:}]
Recall from our discussion in \cref{asymptotic_theory_sec} that (under \cref{asymptotic_norm_condition} and \cref{general_balance_condition}, as $n \to \infty$) we can leverage \cref{tau_dist}, and use the fact that $\varepsilon$ is independent of $\mathbf{Z}$ for $\mathbf{Z} \sim \mathcal{N}(\mathbf{0}, \mathbf{\Sigma})$ and $\varepsilon \sim \mathcal{N}(0, V_{\tau \tau}(1 - R^2))$ to write the variance of $\sqrt{n} \left(\widehat{\tau} - \tau \right)$ under Quadratic Form Rerandomization as
\begin{align*}
    \text{Var}\left(  \sqrt{n} \left(\widehat{\tau} - \tau \right) \mid \mathbf{X}, Q_{\mathbf{A}}(\sqrt{n} \widehat{\bm{\tau}}_{\mathbf{X}}) \leq a \right) &= \text{Var}\left(\varepsilon\right) + \text{Var}\left(\mathbf{V}_{\tau x} \mathbf{\Sigma}^{-1} \mathbf{Z} \mid  \mathbf{Z}^T \mathbf{A} \mathbf{Z} \leq a \right) \\
    &= \text{Var}\left(\varepsilon \right) + \mathbf{V}_{\tau x} \mathbf{\Sigma}^{-1} \text{Cov}\left( \mathbf{Z} \mid Q_{\mathbf{A}}(\mathbf{Z}) \leq a \right) \mathbf{\Sigma}^{-1} \mathbf{V}_{x \tau } \\
    &= \text{Var}\left(\varepsilon\right) + \mathbf{V}_{\tau x} \mathbf{\Sigma}^{-1} \left(\mathbf{\Sigma}^{\nicefrac{1}{2}} \mathbf{\Omega} \mathbf{\Psi} \mathbf{\Omega}^T \mathbf{\Sigma}^{\nicefrac{1}{2}} \right) \mathbf{\Sigma}^{-1} \mathbf{V}_{x \tau }
\end{align*}
where again $\mathbf{\Psi} = \text{diag}\{(\nu_{j, \eta})_{1 \leq j \leq d} \}$ and the last equality follows by applying \cref{Theorem1}. From here, we define $\bm{\beta} = \mathbf{\Sigma}^{-1} \mathbf{V}_{x \tau }$ and $\bm{\beta}_Z = \mathbf{V}^T \bm{\beta}$ to be the rotation of $\bm{\beta}$ along the principal components of $\mathbf{X}$. Then, it follows that the difference
\begin{align*}
    \Delta := \text{Var}\left(  \sqrt{n} \left(\widehat{\tau} - \tau \right) \mid \mathbf{X} \right) - \text{Var}\left(  \sqrt{n} \left(\widehat{\tau} - \tau \right) \mid \mathbf{X}, Q_{\mathbf{A}}(\sqrt{n} \widehat{\bm{\tau}}_{\mathbf{X}}) \leq a \right)
\end{align*}
is given by
\begin{align*}
    \Delta &= \bm{\beta}^T \mathbf{\Sigma} \bm{\beta} - \bm{\beta}^T \mathbf{\Sigma}^{\nicefrac{1}{2}} \mathbf{\Omega} \Big( \text{diag}\{(\nu_{j, \eta})_{1 \leq j \leq d} \} \Big) \mathbf{\Omega}^T \mathbf{\Sigma}^{\nicefrac{1}{2}} \bm{\beta} \\
    &= \bm{\beta}^T \mathbf{\Sigma}^{\nicefrac{1}{2}} \left(\mathbf{I}_d - \mathbf{\Omega} \Big( \text{diag}\{(\nu_{j, \eta})_{1 \leq j \leq d} \} \Big) \mathbf{\Omega}^T  \right) \mathbf{\Sigma}^{\nicefrac{1}{2}} \bm{\beta} \\
    &= \bm{\beta}^T_Z \mathbf{V}^T ( \mathbf{V} \mathbf{\Lambda}^{\nicefrac{1}{2}} \mathbf{V}^T) \left(\mathbf{\Omega} \mathbf{\Omega}^T - \mathbf{\Omega} \Big( \text{diag}\{(\nu_{j, \eta})_{1 \leq j \leq d} \} \Big) \mathbf{\Omega}^T  \right) ( \mathbf{V} \mathbf{\Lambda}^{\nicefrac{1}{2}} \mathbf{V}^T) \mathbf{V} \bm{\beta}_Z \\
    &= \bm{\beta}^T_Z  \mathbf{\Lambda}^{\nicefrac{1}{2}} \mathbf{V}^T\mathbf{\Omega} \Big( \mathbf{I}_d  - \text{diag}\{(\nu_{j, \eta})_{1 \leq j \leq d} \} 
 \Big)\mathbf{\Omega}^T   \mathbf{V} \mathbf{\Lambda}^{\nicefrac{1}{2}}  \bm{\beta}_Z.
\end{align*}
Thus, we may equivalently write $\Delta$ as the summation
\begin{align*}
    \Delta = \sum^d_{j=1} \left(\mathbf{\Omega}^T   \mathbf{V} \mathbf{\Lambda}^{\nicefrac{1}{2}}  \bm{\beta}_Z \right)^2_j (1 - \nu_{j, \eta}) = \sum^d_{j=1} \left((\mathbf{\Omega}^T   \mathbf{V})_j \mathbf{\Lambda}^{\nicefrac{1}{2}}  \bm{\beta}_Z \right)^2 (1 - \nu_{j, \eta}).
\end{align*}
where $(\mathbf{\Omega}^T   \mathbf{V})_j$ is the $j$th row of $\mathbf{\Omega}^T \mathbf{V}$. The second equality follows since
\begin{align*}
    \mathbf{\Lambda}^{\nicefrac{1}{2}}  \bm{\beta}_Z = \begin{pmatrix}
        \lambda^{\nicefrac{1}{2}}_1 & & 0\\
        & \ddots & \\
        0 & & \lambda^{\nicefrac{1}{2}}_d
    \end{pmatrix} \begin{pmatrix}
        \beta_{Z, 1} \\
        \vdots \\
        \beta_{Z, d}
    \end{pmatrix} = \begin{pmatrix}
        \lambda^{\nicefrac{1}{2}}_1 \beta_{Z, 1} \\
        \vdots \\
       \lambda^{\nicefrac{1}{2}}_d \beta_{Z, d}
    \end{pmatrix},
\end{align*}
so (if we say $\mathbf{P} = \mathbf{\Omega}^T   \mathbf{V}$) it is clear that
\begin{align*}
    \mathbf{P}\mathbf{\Lambda}^{\nicefrac{1}{2}}  \bm{\beta}_Z = \begin{pmatrix}
        \mathbf{P}_1 \\
        \vdots \\
        \mathbf{P}_d
    \end{pmatrix} \begin{pmatrix}
        \lambda^{\nicefrac{1}{2}}_1 \beta_{Z, 1} \\
        \vdots \\
       \lambda^{\nicefrac{1}{2}}_d \beta_{Z, d}
    \end{pmatrix} = \begin{pmatrix}
        \mathbf{P}_1 \mathbf{\Lambda}^{\nicefrac{1}{2}} \bm{\beta}_Z \\
        \vdots \\
       \mathbf{P}_d \mathbf{\Lambda}^{\nicefrac{1}{2}} \bm{\beta}_Z 
    \end{pmatrix}.
\end{align*}
Thus, the $j$th element of $\mathbf{\Omega}^T   \mathbf{V} \mathbf{\Lambda}^{\nicefrac{1}{2}}  \bm{\beta}_Z$ is simply $(\mathbf{\Omega}^T   \mathbf{V})_j \mathbf{\Lambda}^{\nicefrac{1}{2}}  \bm{\beta}_Z$. Furthermore, note that if $\mathbf{\Sigma}$ and $\mathbf{\Sigma}^{\nicefrac{1}{2}} \mathbf{A} \mathbf{\Sigma}^{\nicefrac{1}{2}}$ share an eigenbasis, then $\mathbf{\Omega} = \mathbf{\Gamma}$, where $\mathbf{\Gamma}$ is the orthogonal matrix of eigenvectors of $\mathbf{\Sigma}$. In this case, $\mathbf{V}$ and $\mathbf{\Gamma}$ are identical, up to a scaling of $\pm 1$ --- this follows since
\begin{align*}
   \mathbf{\Sigma} \propto  \mathbf{X}^T \mathbf{X} = (\mathbf{U}\mathbf{D}\mathbf{V}^T)^T(\mathbf{U}\mathbf{D}\mathbf{V}^T) = \mathbf{V} \mathbf{D}^2 \mathbf{V}^T.
\end{align*}
Therefore, $\mathbf{\Gamma}^T \mathbf{V}$ will be a diagonal matrix with $\pm 1$ on each diagonal entry. Thus, when $\mathbf{\Sigma}$ and $\mathbf{\Sigma}^{\nicefrac{1}{2}} \mathbf{A} \mathbf{\Sigma}^{\nicefrac{1}{2}}$ share an eigenbasis,
\begin{align*}
    \text{Var}\left(\sqrt{n} \left(\widehat{\tau} - \tau \right) \mid \mathbf{X} \right) - \text{Var}\left(\sqrt{n} \left(\widehat{\tau} - \tau \right) \mid \mathbf{X}, Q_{\mathbf{A}}(\sqrt{n} \widehat{\bm{\tau}}_{\mathbf{X}}) \leq a \right) &= \sum^d_{j=1} \lambda_j \beta^2_{Z, j} (1 - \nu_{j, \eta}).
\end{align*}
Furthermore, note that since $\nu_{j, \eta} \leq 1$ for all $j = 1, \ldots, d$ it follows that 
\begin{align*}
    \sum^d_{j=1} \left((\mathbf{\Omega}^T   \mathbf{V})_j \mathbf{\Lambda}^{\nicefrac{1}{2}}  \bm{\beta}_Z \right)^2 (1 - \nu_{j, \eta}) \geq 0,
\end{align*}
where equality holds if $\bm{\beta}_Z = 0$, thereby implying that Quadratic Form Rerandomization always weakly decreases the variance of $\widehat{\tau}$. The inequality is strict if and only if there exists at least one $j$ such that $(\mathbf{\Omega}^T\mathbf{\Sigma}^{\nicefrac{1}{2}}\bm{\beta})_j\neq 0$ and $\nu_{j,\eta} < 1$.
\end{proof}

\subsection{Proof of \texorpdfstring{\cref{optimal_a_outcomes}}{Theorem 5}} \label{optimal_a_outcomes_proof}
\begin{proof}[\textbf{Proof:}]
Recall that by \cref{diff_in_vars} we know that under \cref{asymptotic_norm_condition} and \cref{general_balance_condition}, as $n \to \infty$, it follows that
 \begin{align*}
        \text{Var}(\sqrt{n} \left(\widehat{\tau} - \tau \right) \mid \mathbf{X}) - \text{Var}\left(  \sqrt{n} \left(\widehat{\tau} - \tau \right) \mid \mathbf{X}, Q_{\mathbf{A}}(\sqrt{n} \widehat{\bm{\tau}}_{\mathbf{X}}) \leq a \right)  &= \sum^d_{j=1} (\mathbf{P}_j \mathbf{\Lambda}^{\nicefrac{1}{2}} \bm{\beta}_Z)^2 (1 - \nu_{j, \eta}) 
    \end{align*}
    where $\bm{\beta}_Z = \mathbf{V}^T \bm{\beta} = \mathbf{V}^T \mathbf{\Sigma}^{-1} \mathbf{V}_{x \tau}$, $\mathbf{\Lambda}$ is the diagonal matrix of eigenvalues of $\mathbf{\Sigma}$, and $\mathbf{P} = \mathbf{\Omega}^T \mathbf{V}$. Therefore, in order to reduce the variance of $\widehat{\tau}$ as much as possible relative to complete randomization, we must focus on minimizing $\sum^d_{j=1} (\mathbf{P}_j \mathbf{\Lambda}^{\nicefrac{1}{2}} \bm{\beta}_Z)^2 \nu_{j, \eta}$. Observe that we may write this summation as
    \begin{align*}
        \sum^d_{j=1} (\mathbf{P}_j \mathbf{\Lambda}^{\nicefrac{1}{2}} \bm{\beta}_Z)^2 \nu_{j, \eta} &= \bm{\beta}^T_Z \mathbf{\Lambda}^{\nicefrac{1}{2}} \mathbf{P}^T \mathbf{\Psi} \mathbf{P} \mathbf{\Lambda}^{\nicefrac{1}{2}} \bm{\beta}_Z \\[-0.1in]
        &= \text{tr}\left(\bm{\beta}^T_Z \mathbf{\Lambda}^{\nicefrac{1}{2}} \mathbf{P}^T \mathbf{\Psi} \mathbf{P} \mathbf{\Lambda}^{\nicefrac{1}{2}} \bm{\beta}_Z \right) \\
        &= \text{tr}\left(\mathbf{P} \mathbf{\Lambda}^{\nicefrac{1}{2}} \bm{\beta}_Z \bm{\beta}^T_Z \mathbf{\Lambda}^{\nicefrac{1}{2}} \mathbf{P}^T \mathbf{\Psi}    \right)
    \end{align*}
    where the second equality follows since $\bm{\beta}^T_Z \mathbf{\Lambda}^{\nicefrac{1}{2}} \mathbf{P}^T \mathbf{\Psi} \mathbf{P} \mathbf{\Lambda}^{\nicefrac{1}{2}} \bm{\beta}_Z$ is a scalar, and the third equality follows by the cyclic property of the trace. Importantly, note that 
    \begin{align*}
        \mathbf{P} \mathbf{\Lambda}^{\nicefrac{1}{2}} \bm{\beta}_Z \bm{\beta}^T_Z \mathbf{\Lambda}^{\nicefrac{1}{2}} \mathbf{P}^T &= (\mathbf{\Omega}^T \mathbf{V}) \mathbf{\Lambda}^{\nicefrac{1}{2}} (\mathbf{V}^T \bm{\beta}) (\mathbf{V}^T \bm{\beta})^T \mathbf{\Lambda}^{\nicefrac{1}{2}} (\mathbf{\Omega}^T \mathbf{V})^T \\
        &= \mathbf{\Omega}^T \mathbf{\Sigma}^{\nicefrac{1}{2}} \bm{\beta} \bm{\beta}^T \mathbf{\Sigma}^{\nicefrac{1}{2}} \mathbf{\Omega}
    \end{align*}
    where we have used the fact that $\mathbf{V} \mathbf{\Lambda}^{\nicefrac{1}{2}} \mathbf{V}^T = \mathbf{\Sigma}^{\nicefrac{1}{2}}$. Furthermore, recall that $\mathbf{\Omega}$ diagonalizes $\mathbf{\Sigma}^{\nicefrac{1}{2}} \mathbf{A} \mathbf{\Sigma}^{\nicefrac{1}{2}}$. Thus, if we let $\mathbf{A} = \bm{\beta} \bm{\beta}^T$ it follows that $\mathbf{\Omega}^T \mathbf{\Sigma}^{\nicefrac{1}{2}} \bm{\beta} \bm{\beta}^T \mathbf{\Sigma}^{\nicefrac{1}{2}} \mathbf{\Omega} = \bm{\eta}$, where $\bm{\eta}$ is the diagonal matrix of eigenvalues of $\mathbf{\Sigma}^{\nicefrac{1}{2}} \bm{\beta} \bm{\beta}^T \mathbf{\Sigma}^{\nicefrac{1}{2}}$. From here, note that $\mathbf{\Sigma}^{\nicefrac{1}{2}} \bm{\beta} \bm{\beta}^T \mathbf{\Sigma}^{\nicefrac{1}{2}}$ is a rank-one matrix, and therefore it has only one positive eigenvalue (with all others being zero). Since the trace of a matrix is equal to the sum of its eigenvalues, we know that
    \begin{align*}
        \eta_1 = \text{tr}(\mathbf{\Sigma}^{\nicefrac{1}{2}} \bm{\beta} \bm{\beta}^T \mathbf{\Sigma}^{\nicefrac{1}{2}}) = \text{tr}( \bm{\beta}^T \mathbf{\Sigma} \bm{\beta}) = \bm{\beta}^T \mathbf{\Sigma} \bm{\beta}.
    \end{align*}
    where the second equality follows using the cyclic property of the trace. Thus, when $\mathbf{A} = \bm{\beta} \bm{\beta}^T$, 
    \begin{align*}
        \sum^d_{j=1} (\mathbf{P}_j \mathbf{\Lambda}^{\nicefrac{1}{2}} \bm{\beta}_Z)^2 \nu_{j, \eta} &= \text{tr}\left(\mathbf{P} \mathbf{\Lambda}^{\nicefrac{1}{2}} \bm{\beta}_Z \bm{\beta}^T_Z \mathbf{\Lambda}^{\nicefrac{1}{2}} \mathbf{P}^T \mathbf{\Psi} \right) \\
        &=\text{tr}\left(\mathbf{\Omega}^T \mathbf{\Sigma}^{\nicefrac{1}{2}} \bm{\beta} \bm{\beta}^T \mathbf{\Sigma}^{\nicefrac{1}{2}} \mathbf{\Omega}\mathbf{\Psi} \right) \\
        &= \text{tr}\left(\bm{\eta} \mathbf{\Psi} \right) \\
        &= \eta_1 \nu^*_{1, \beta}
    \end{align*}
    where the last equality follows since $\bm{\eta}$ is a diagonal matrix that is zero everywhere except the first element and we define $\nu^*_{1, \beta} = \mathbb{E}[Z^2 \mid (\bm{\beta}^T \mathbf{\Sigma} \bm{\beta})Z^2 < a^\prime]$ for $Z \sim N(0, 1)$, where $a^\prime$ is chosen such that $\mathbb{P}( (\bm{\beta}^T \mathbf{\Sigma} \bm{\beta}) Z^2 < a^\prime) = \alpha$. Our goal now is to show that the variance of $\widehat{\tau}$ under Quadratic Form Rerandomization when $\mathbf{A} = \bm{\beta} \bm{\beta}^T$ is smaller than any other choice of $\mathbf{A}$. Let $\widetilde{\mathbf{\Omega}}$ be the eigenvectors of $\mathbf{\Sigma}^{\nicefrac{1}{2}} \mathbf{A} \mathbf{\Sigma}^{\nicefrac{1}{2}}$ for any other positive-definite or positive semi-definite choice of $\mathbf{A}$. Then, observe that
    \begin{align*}
        \eta_1 &= \bm{\beta}^T \mathbf{\Sigma} \bm{\beta} \\
        &= (\mathbf{V} \bm{\beta}_Z)^T \mathbf{\Sigma} \mathbf{V} \bm{\beta}_Z \\
        &= \bm{\beta}^T_Z \mathbf{\Lambda} \bm{\beta}_Z \\
        &= \bm{\beta}^T_Z \mathbf{\Lambda}^{\nicefrac{1}{2}}(\widetilde{\mathbf{\Omega}}^T \mathbf{V})^T (\widetilde{\mathbf{\Omega}}^T \mathbf{V}) \mathbf{\Lambda}^{\nicefrac{1}{2}}\bm{\beta}_Z \\
        &= \sum^d_{j=1} ((\widetilde{\mathbf{\Omega}}^T \mathbf{V})_j \mathbf{\Lambda}^{\nicefrac{1}{2}} \bm{\beta}_Z)^2,
    \end{align*}
    where the fourth equality follows since $\widetilde{\mathbf{\Omega}}^T \mathbf{V}$ is an orthogonal matrix. Thus, the difference in variances between $\mathbf{A} = \bm{\beta} \bm{\beta}^T$ and any other positive semi-definite or positive-definite choice of $\mathbf{A}$ is given by
    \begin{align*}
        \sum^d_{j=1} ((\widetilde{\mathbf{\Omega}}^T \mathbf{V})_j \mathbf{\Lambda}^{\nicefrac{1}{2}} \bm{\beta}_Z)^2(\nu^*_{1, \beta} - \nu_{j, \eta}).
    \end{align*}
    Then, if we can show that $\nu^*_{1, \beta} \leq \nu_{j, \eta}$ for each $j = 1, \ldots, d$, the proof is complete. To show this, observe that we may write $\nu^*_{1, \beta} = \mathbb{E}[Z^2 \mid Z^2 < q_\alpha]$ where $q_\alpha = a^\prime / (\bm{\beta}^T \mathbf{\Sigma} \bm{\beta})$. Importantly, since $a^\prime$ is chosen to have acceptance probability $\alpha$, $q_\alpha$ is simply the $\alpha$-quantile of a $\chi^2_1$ random variable. Next, fix $j$ and let
    \begin{align*}
        \mathcal{E}_{\mathbf{A}} = \left\{ \sum^d_{\ell=1} \eta_\ell Z^2_\ell < a\right\}
    \end{align*}
    denote the acceptance event for some arbitrary $\mathbf{A} \in \mathbf{S}^d_{+}$ such that $\mathbb{P}(\mathcal{E}_{\mathbf{A}}) = \alpha$. Our goal is to show that among all events $\mathcal{E}_{\mathbf{A}}$ with acceptance probability $\alpha$,
    \begin{align*}
        \mathbb{E}[Z^2_j \mid \mathcal{E}_{\mathbf{A}}] \geq \mathbb{E}[Z^2_j \mid Z^2_j \leq q_\alpha].
    \end{align*}
    From here, observe that
    \begin{align} \label{expectation_inequality}
         \mathbb{E}\left[Z^2_j \mathbb{I}(\mathcal{E}_{\mathbf{A}})\right] - \mathbb{E}\left[Z^2_j\mathbb{I}(Z^2_j\leq q_\alpha)\right] \overset{(i)}{=} \mathbb{E}\left[(Z^2_j-q_\alpha) \left\{ \mathbb{I}(\mathcal{E}_{\mathbf{A}}) -\mathbb{I}(Z^2_j\leq q_\alpha) \right\} \right] \overset{(ii)}{\geq} 0
    \end{align}
    where $(i)$ follows since $\mathbb{E}\left[(Z^2_j-q_\alpha) \left\{ \mathbb{I}(\mathcal{E}_{\mathbf{A}}) -\mathbb{I}(Z^2_j\leq q_\alpha) \right\} \right]$ expands to
    \begin{align*}
         \mathbb{E}\left[Z^2_j \mathbb{I}(\mathcal{E}_{\mathbf{A}})\right] - \mathbb{E}\left[Z^2_j\mathbb{I}(Z^2_j\leq q_\alpha)\right] - q_\alpha \mathbb{P}(\mathcal{E}_{\mathbf{A}}) + q_\alpha \mathbb{P}(Z^2_j\leq q_\alpha)
    \end{align*}
    and the last two terms cancel since $\mathbb{P}(Z^2_j\leq q_\alpha) = \mathbb{P}(\mathcal{E}_{\mathbf{A}}) = \alpha$. Next, $(ii)$ follows because when $Z^2_j \leq q_\alpha$, both $Z^2_j -q_\alpha\leq0$ and $\mathbb{I}(\mathcal{E}_{\mathbf{A}})-\mathbb{I}(Z^2_j\leq q_\alpha)\leq0$, while when $Z^2_j > q_\alpha$, both $Z^2_j -q_\alpha\geq0$ and $\mathbb{I}(\mathcal{E}_{\mathbf{A}}) -\mathbb{I}(Z^2_j \leq q_\alpha)\geq0$. Finally, since
    \begin{align*}
        \mathbb{E}[Z^2_j \mid \mathcal{E}_{\mathbf{A}}] = \frac{\mathbb{E}[Z^2_j \mathbb{I}(\mathcal{E}_{\mathbf{A}})]}{\mathbb{P}(\mathcal{E}_{\mathbf{A}})} = \frac{\mathbb{E}[Z^2_j \mathbb{I}(\mathcal{E}_{\mathbf{A}})]}{\alpha}
    \end{align*}
    and
    \begin{align*}
        \mathbb{E}[Z^2_j \mid Z^2_j \leq q_\alpha] = \frac{\mathbb{E}[Z^2_j \mathbb{I}(Z^2_j \leq q_\alpha)]}{\mathbb{P}(Z^2_j \leq q_\alpha)} = \frac{\mathbb{E}[Z^2_j \mathbb{I}(Z^2_j \leq q_\alpha)]}{\alpha}
    \end{align*}
    dividing both sides of \cref{expectation_inequality} by $\alpha$ and rearranging yields
    \begin{align*}
        \mathbb{E}[Z^2_j \mid \mathcal{E}_{\mathbf{A}}] \geq \mathbb{E}[Z^2_j \mid Z^2_j \leq q_\alpha]. 
    \end{align*}
    Consequently, it follows that $\nu^*_{1, \beta} \leq \nu_{j, \eta}$ for each $j = 1, \ldots, d$ and thus,
    \begin{align*}
    \text{Var}(\sqrt{n}(\widehat{\tau} - \tau)\mid \mathbf{X}, Q_{\mathbf{A}^*}(\sqrt{n} \widehat{\bm{\tau}}_{\mathbf{X}}) \leq a^\prime) \leq \text{Var}(\sqrt{n}(\widehat{\tau} - \tau) \mid \mathbf{X}, Q_{\mathbf{A}}(\sqrt{n} \widehat{\bm{\tau}}_{\mathbf{X}}) \leq a)
    \end{align*}
    where $\mathbf{A}^* = \bm{\beta}\bm{\beta}^T$.
    
Furthermore, we can show that the variance reduction of $\widehat{\tau}$ under optimal Quadratic Form Rerandomization is
\begin{align*}
     \text{Var}\left(  \sqrt{n} \left(\widehat{\tau} - \tau \right) \mid \mathbf{X}, Q_{\mathbf{A}^*}(\sqrt{n} \widehat{\bm{\tau}}_{\mathbf{X}}) \leq a \right) &= \nu^*_{1, \beta} \sum^d_{j=1} \beta^2_{Z,j} \lambda_j +  \text{Var}(\varepsilon) \\
     &\overset{(i)}{=} \nu^*_{1, \beta} \Big( \text{Var}(\sqrt{n} \left(\widehat{\tau} - \tau \right) \mid \mathbf{X}) - \text{Var}(\varepsilon) \Big) +  \text{Var}(\varepsilon) \\
     &\overset{(ii)}{=} \nu^*_{1, \beta}\Big(V_{\tau \tau} - \text{Var}(\varepsilon) \Big) + \text{Var}( \varepsilon) \\
    &= \nu^*_{1, \beta} V_{\tau \tau} + (1 - \nu^*_{1, \beta}) \text{Var}( \varepsilon) \\
    &\overset{(iii)}{=} \nu^*_{1, \beta} V_{\tau \tau} + (1 - \nu^*_{1, \beta}) V_{\tau \tau}(1 - R^2) \\
    &=  (1 - (1 - \nu^*_{1, \beta})R^2) V_{\tau \tau}.
\end{align*}
where in $(i)$ we have used the fact that $\sum^d_{j=1} \beta^2_{Z,j} \lambda_j = \text{Var}(\sqrt{n} \left(\widehat{\tau} - \tau \right) \mid \mathbf{X}) - \text{Var}(\varepsilon)$, in $(ii)$ we have used
$\text{Var}(\sqrt{n} \left(\widehat{\tau} - \tau \right) \mid \mathbf{X}) = V_{\tau \tau}$, and in $(iii)$ that $\text{Var}(\varepsilon) = V_{\tau \tau}(1 - R^2)$. Putting these facts together, we can see that the optimal percentage reduction in variance under Quadratic Form Rerandomization is $100(1 - \nu^*_{1, \beta})R^2$.
\end{proof}

\subsection{Proof of \texorpdfstring{\cref{euc_minimax}}{Theorem 6}} \label{euc_minimax_proof}

First, we will prove a technical lemma regarding the conditional expectation of Gamma random variables. Although this lemma can be derived from other works such as \cite{chapman_truncated_gamma}, we provide a proof here for completeness and for the convenience of the reader.
\begin{lemma} \label{Gamma_cond_exp}
    Suppose $X \sim \text{Gamma}(\alpha, \beta)$. Then, for some $a > 0$,
    \begin{align*}
        \mathbb{E}\left[X \mid X < a \right] = \alpha \beta \left(\frac{F(a; \alpha + 1, \beta)}{F(a; \alpha, \beta)}\right)
    \end{align*}
    where $F(a; \cdot, \cdot)$ is the cumulative distribution function of $X \sim \text{Gamma}(\cdot, \cdot)$.
\end{lemma}
\begin{proof}[\textbf{Proof:}] First, note that we parameterize $X \sim \text{Gamma}(\alpha, \beta)$ such that $X$ has the density
\begin{align*}
    f(x; \alpha, \beta) = \frac{x^{\alpha - 1} e^{- x / \beta}}{\beta^\alpha \Gamma(\alpha)}.
\end{align*}
Then, observe that
\begin{align*}
    \mathbb{E}\left[X \mid X < a \right] &= \frac{\mathbb{E}\left[X \cdot \mathbbm{1}(X < a) \right]}{\mathbb{P}(X < a)} = \frac{\int^a_0 x \left( \frac{x^{\alpha - 1} e^{- x / \beta}}{\beta^\alpha \Gamma(\alpha)}  \right) dx}{\int^a_0 \frac{x^{\alpha - 1} e^{- x / \beta}}{\beta^\alpha \Gamma(\alpha)} dx} = \frac{\frac{1}{\Gamma(\alpha)}\int^a_0  \frac{x^{\alpha} e^{- x / \beta}}{\beta^\alpha }  dx}{\frac{1}{\Gamma(\alpha)}\int^a_0 \frac{x^{\alpha - 1} e^{- x / \beta}}{\beta^\alpha } dx}.
\end{align*}
Next, we can write this expression as a ratio of lower incomplete gamma functions, $\gamma(s, x) = \int^x_0 t^{s-1} e^{-t} dt$, since
\begin{align*}
  \frac{\int^a_0  \frac{x^{\alpha} e^{- x / \beta}}{\beta^\alpha }  dx}{\int^a_0 \frac{x^{\alpha - 1} e^{- x / \beta}}{\beta^\alpha } dx}  &=    \frac{\int^a_0 (x / \beta)^{\alpha} e^{- x / \beta}  dx}{\frac{1}{\beta}\int^a_0 (x / \beta)^{\alpha - 1} e^{- x / \beta} dx} = \beta \frac{\int^{a / \beta}_0 y^{\alpha} e^{- y}  dy}{\int^{a / \beta}_0 y^{\alpha - 1} e^{- y} dy} = \beta \cdot \frac{\gamma(\alpha + 1, a / \beta)}{\gamma(\alpha, a / \beta)}
\end{align*}
where in the second to last equality we let $y = x / \beta$. We can further simplify this expression by noting that $\gamma(\alpha + 1, a / \beta) = \alpha \gamma( \alpha, a / \beta) - (a / \beta)^\alpha e^{- a / \beta}$. Then, we can see that
\begin{align*}
    \beta \cdot \frac{\gamma(\alpha + 1, a / \beta)}{\gamma(\alpha, a / \beta)} = \beta \left(\alpha - \frac{(a / \beta)^\alpha e^{- a / \beta}}{\gamma(\alpha, a / \beta)}\right) = \alpha \beta \left(1 - \frac{(a / \beta)^\alpha e^{- a / \beta}}{\alpha \gamma(\alpha, a / \beta)}\right).
\end{align*}
Finally, using the fact that $\Gamma(\alpha + 1) = \alpha \Gamma(\alpha)$, it follows that
\begin{align*}
     \frac{(a / \beta)^\alpha e^{- a / \beta}}{\alpha \gamma(\alpha, a / \beta)} = \frac{\frac{a^\alpha e^{- a / \beta}}{\beta^\alpha \Gamma(\alpha + 1)} }{\frac{\alpha \gamma(\alpha, a / \beta)}{\Gamma(\alpha + 1)}} = \frac{\beta \left(\frac{a^\alpha e^{- a / \beta}}{\beta^{\alpha+1} \Gamma(\alpha + 1)} \right) }{\frac{\gamma(\alpha, a / \beta)}{\Gamma(\alpha)}} = \beta \left(\frac{f(a; \alpha + 1, \beta)}{F(a; \alpha, \beta)} \right).
\end{align*}
Putting everything together, we have that
\begin{align*}
    \mathbb{E}\left[X \mid X < a \right] = \alpha \beta \left(1 - \beta \cdot \frac{f(a; \alpha + 1, \beta)}{F(a; \alpha, \beta)}\right).
\end{align*}
Note that this makes sense intuitively, since $\mathbb{E}[X] = \alpha \beta$, so we can see that as $a \to \infty$ we recover the unconditional mean. We can also see that the mean of $X$ is reduced by $\beta \cdot \frac{f(a; \alpha + 1, \beta)}{F(a; \alpha, \beta)}$. However, it is still possible to simplify this conditional expectation. To do so, observe that
\begin{align*}
    F(a; \alpha, \beta) - \beta f(a ; \alpha + 1, \beta) &= \frac{\gamma(\alpha, a / \beta)}{\Gamma(\alpha)} - \beta \left(\frac{a^\alpha e^{- a / \beta}}{\beta^{\alpha+1} \Gamma(\alpha + 1)} \right) \\
    &= \frac{\alpha \gamma(\alpha, a / \beta)}{\Gamma(\alpha + 1)} - \frac{\beta a^\alpha e^{- a / \beta}}{\beta^{\alpha+1} \Gamma(\alpha + 1)} \\
    &= \frac{\beta^{\alpha + 1} \alpha \gamma(\alpha, a / \beta) - \beta a^\alpha e^{- a / \beta}}{\beta^{\alpha + 1}\Gamma(\alpha + 1)} \\
    &= \frac{\beta^{\alpha + 1} \left( \gamma(\alpha + 1, a / \beta) + (a / \beta)^\alpha e^{- a / \beta}\right) - \beta a^\alpha e^{- a / \beta}}{\beta^{\alpha + 1}\Gamma(\alpha + 1)} \\
    &= \frac{ \gamma(\alpha + 1, a / \beta)}{\Gamma(\alpha + 1)} + \frac{\beta a^\alpha e^{- a / \beta} - \beta a^\alpha e^{- a / \beta}}{\beta^{\alpha + 1}\Gamma(\alpha + 1)} \\
    &= F(a; \alpha + 1, \beta)
\end{align*}
where the fourth equality follows since $\alpha \gamma( \alpha, a / \beta) = \gamma(\alpha + 1, a / \beta) + (a / \beta)^\alpha e^{- a / \beta}$. Thus,
\begin{align*}
    1 - \beta \cdot \frac{f(a; \alpha + 1, \beta)}{F(a; \alpha, \beta)} = \frac{F(a; \alpha, \beta) - \beta f(a; \alpha + 1, \beta)}{F(a; \alpha, \beta)} = \frac{F(a; \alpha + 1, \beta)}{F(a; \alpha, \beta)}
\end{align*}
which yields a final conditional expectation of
\begin{align*}
    \mathbb{E}\left[X \mid X < a \right] = \alpha \beta \left(\frac{F(a; \alpha + 1, \beta)}{F(a; \alpha, \beta)}\right)
\end{align*}
\end{proof}

\begin{proof}[\textbf{Proof:}] To begin, suppose that $\bm{\beta} \neq \mathbf{0}$ and that $d \geq 2$; when $d = 1$ then the result holds trivially since any rerandomization criterion invokes the same acceptance event. With this in mind, we proceed by directly evaluating $\nu^*_{1, \beta }$ to obtain a better understanding of its properties. Note that for $Z \sim \mathcal{N}(0, 1)$,
\begin{align*}
    \nu^*_{1, \beta } &= \mathbb{E}\left[Z^2 \mid (\bm{\beta}^T \mathbf{\Sigma} \bm{\beta}) Z^2 < a \right] = \mathbb{E}\left[Z^2 \mid  Z^2 < q_{\alpha, \bm{\beta}} \right]
\end{align*}
where $q_{\alpha, \bm{\beta}} = a / (\bm{\beta}^T \mathbf{\Sigma} \bm{\beta})$. Note that because $a$ is specifically chosen to yield an acceptance probability of $\alpha$, it follows that $\alpha = \mathbb{P}( \bm{\beta}^T \mathbf{\Sigma} \bm{\beta}) Z^2 \leq a) = \mathbb{P}( Z^2 \leq q_{\alpha, \bm{\beta}})$; thus $q_{\alpha, \bm{\beta}} = q_{\alpha}$ is simply the $\alpha$-quantile of a $\chi^2_1$ random variable and does not depend on $\bm{\beta}$ after normalizing. By \cref{Gamma_cond_exp}, we know that 
\begin{align*}
    \nu^*_{1, \beta } &= \frac{F(q_{\alpha}; 1/2 + 1, 2)}{F(q_{\alpha}; 1/2, 2)}
    \end{align*}
    where $F(t; \kappa, \theta)$ is the cumulative distribution function of a Gamma random variable with shape and scale parameters $\kappa$ and $\theta$. Note that
\begin{align*}
    F(t;\kappa,2) = \frac{1}{\Gamma(\kappa)} \int_0^{t/2}x^{\kappa-1}e^{-x}dx.
\end{align*}
Thus, as $t\downarrow0$, the approximation $e^{-x}=1+O(x)$ holds uniformly over $0\leq x\leq t/2$. Therefore,
\begin{align*}
    F(t;\kappa,2) &= \frac{1}{\Gamma(\kappa)} \int_0^{t/2} x^{\kappa-1}\{1+O(x)\}dx \\
    &= \frac{1}{\Gamma(\kappa)} \left[ \int_0^{t/2}x^{\kappa-1}dx + O\!\left( \int_0^{t/2}x^\kappa dx\right) \right] \\
    &= \frac{(t/2)^\kappa}{\Gamma(\kappa+1)} + O(t^{\kappa+1}).
\end{align*}
Thus, after noting that $\Gamma(3/2)=\sqrt{\pi}/2$ and $\Gamma(5/2)=3\sqrt{\pi}/4$, we can simply plug in $\kappa = 1/2$ and $\kappa = 3/2$ to see that 
\begin{align*}
    \nu^*_{1,\beta} &= \frac{ \frac{1}{3} \sqrt{\frac{2}{\pi}} q_\alpha^{\nicefrac{3}{2}} + O(q_\alpha^{\nicefrac{5}{2}})}{ \sqrt{\frac{2}{\pi}} q_\alpha^{\nicefrac{1}{2}} + O(q_\alpha^{\nicefrac{3}{2}})} = \frac{\frac{1}{3}q_\alpha + O(q_\alpha^2) }{ 1+O(q_\alpha) } = \frac{q_\alpha}{3} + O(q_\alpha^2).
\end{align*}
Moreover, observe that
\begin{align*}
     \alpha = F(q_\alpha;1/2,2) = \sqrt{\frac{2}{\pi}}q_\alpha^{\nicefrac{1}{2}} + O(q_\alpha^{\nicefrac{3}{2}}) = q_\alpha^{\nicefrac{1}{2}} \left( \sqrt{\frac{2}{\pi}} + O(q_\alpha) \right),
\end{align*}
so it follows that $q_\alpha=O(\alpha^2)$, and consequently, $\nu^*_{1,\beta}=O(\alpha^2)$. Note that this $O(\alpha^2)$ bound is uniform over $||\bm{\beta}||_2<c$, because $\nu^*_{1,\beta}$ depends on $\bm{\beta}$ only through the normalized
threshold $q_\alpha$. As a result, we can now see that uniformly over $|| \bm{\beta}||_2 < c$,
\begin{align*}
    (\bm{\beta}^T \mathbf{\Sigma} \bm{\beta})\nu^*_{1,\beta} \leq \lambda_{\max}(\mathbf{\Sigma}) ||\bm{\beta}||_2^2 O(\alpha^2) \notag \leq c^2\lambda_{\max}(\mathbf{\Sigma})O(\alpha^2) = o(\alpha^{\nicefrac{2}{d}}),
\end{align*}
where the last equality uses $d\geq2$, since $\alpha^2/\alpha^{2/d}=\alpha^{2-2/d}\to0$. Now, recall that under \cref{asymptotic_norm_condition} and \cref{general_balance_condition} (and as $n \to \infty$),
\begin{align*}
    \text{Var}\left(\sqrt{n}(\widehat{\tau} - \tau) \mid \mathbf{X}, Q_{\mathbf{A}}(\sqrt{n} \widehat{\bm{\tau}}_{\mathbf{X}}) \leq a \right) = \bm{\beta}^T \mathcal{C}_{\mathbf{A}} \bm{\beta} + \text{Var}( \varepsilon )
\end{align*}
where we define $\mathcal{C}_{\mathbf{A}} = \text{Cov}\left( \sqrt{n} \widehat{\bm{\tau}}_{\mathbf{X}} \mid \mathbf{X}, Q_{\mathbf{A}}(\sqrt{n} \widehat{\bm{\tau}}_{\mathbf{X}}) \leq a  \right)$. Similarly, under the oracle,
\begin{align*}
    \text{Var}\left(\sqrt{n}(\widehat{\tau} - \tau) \mid \mathbf{X}, Q_{\mathbf{A}^*}(\sqrt{n} \widehat{\bm{\tau}}_{\mathbf{X}}) \leq a^\prime \right) = \nu^*_{1,\beta} \big( \bm{\beta}^T \mathbf{\Sigma} \bm{\beta}\big) + \text{Var}( \varepsilon ).
\end{align*}
For notational convenience, let
\begin{align*}
    \Delta_{\bm{\beta}}(\mathbf{A}) = \left| \text{Var}\!\left( \sqrt{n}(\widehat{\tau}-\tau) \mid \mathbf{X}, Q_{\mathbf{A}^*} (\sqrt{n}\widehat{\bm{\tau}}_{\mathbf{X}}) \leq
a^\prime\right) - \text{Var}\!\left( \sqrt{n}(\widehat{\tau}-\tau) \mid \mathbf{X}, Q_{\mathbf{A}} (\sqrt{n}\widehat{\bm{\tau}}_{\mathbf{X}}) \leq
a \right) \right|
\end{align*}
denote the absolute difference between the variance under the oracle and under any arbitrary quadratic form. Plugging in each decomposition, it is clear that $\Delta_{\bm{\beta}}(\mathbf{A})  = |\bm{\beta}^T \mathcal{C}_{\mathbf{A}} \bm{\beta} - \nu^*_{1,\beta} \big( \bm{\beta}^T \mathbf{\Sigma} \bm{\beta}\big)|$. From here, let $u_{\bm{\beta}} = \bm{\beta}^T \mathcal{C}_{\mathbf{A}} \bm{\beta}$ and $v_{\bm{\beta}} = \nu^*_{1,\beta} \big( \bm{\beta}^T \mathbf{\Sigma} \bm{\beta}\big)$. Since both $u_{\bm{\beta}}$ and $v_{\bm{\beta}}$ are nonnegative it follows that
\begin{align*}
    \left| |u_{\bm{\beta}}-v_{\bm{\beta}}| - u_{\bm{\beta}} \right| \leq v_{\bm{\beta}}.
\end{align*}
Thus, taking the supremum over $||\bm{\beta} ||_2<c$ yields
\begin{align*}
\left| \sup_{||\bm{\beta}||_2<c} |u_{\bm{\beta}}-v_{\bm{\beta}}| - \sup_{||\bm{\beta}||_2<c} u_{\bm{\beta}}
\right|  \leq \sup_{||\bm{\beta}||_2<c} \left| |u_{\bm{\beta}}-v_{\bm{\beta}}| - u_{\bm{\beta}} \right| \leq \sup_{||\bm{\beta}||_2<c} v_{\bm{\beta}} = o(\alpha^{\nicefrac{2}{d}}).
\end{align*}
Consequently,
\begin{align*}
\sup_{||\bm{\beta}||_2<c} \Delta_{\bm{\beta}}(\mathbf{A})   = \sup_{||\bm{\beta}||_2<c} \bm{\beta}^T\mathcal{C}_{\mathbf{A}}\bm{\beta} + o(\alpha^{\nicefrac{2}{d}}).
\end{align*}
From here, since $\mathcal{C}_{\mathbf{A}}$ is positive semi-definite,
\begin{align*}
    \sup_{||\bm{\beta}||_2<c} \bm{\beta}^T\mathcal{C}_{\mathbf{A}}\bm{\beta} = c^2\eta_{\max}(\mathcal{C}_{\mathbf{A}}),
\end{align*}
where we say $\eta_{\max}(\mathcal{C}_{\mathbf{A}})$ is the maximum eigenvalue of $\mathcal{C}_{\mathbf{A}}$. Thus, minimizing $\Delta_{\bm{\beta}}(\mathbf{A})$ across $\mathbf{A} \in \mathbf{S}^d_{+}$ is equivalent to minimizing $\eta_{\max}(\mathcal{C}_{\mathbf{A}})$. To proceed, we restrict our attention to positive-definite choices of $\mathbf{A}$ in order to apply the small-$\alpha$ approximation of $\nu_{j, \eta}$. Recall that
\begin{align*}
    \mathcal{C}_{\mathbf{A}} &= \mathbf{\Sigma}^{\nicefrac{1}{2}} \mathbf{\Omega} \Big(\text{diag}\{(\nu_{j, \eta})_{1 \leq j \leq d} \} \Big) \mathbf{\Omega}^T \mathbf{\Sigma}^{\nicefrac{1}{2}} \\
    &= p_d\alpha^{\nicefrac{2}{d}} \underbrace{\text{det}(\mathbf{\Sigma}^{\nicefrac{1}{2}} \mathbf{A} \mathbf{\Sigma}^{\nicefrac{1}{2}})^{\nicefrac{1}{d}} \mathbf{\Sigma}^{\nicefrac{1}{2}} (\mathbf{\Sigma}^{\nicefrac{1}{2}} \mathbf{A} \mathbf{\Sigma}^{\nicefrac{1}{2}})^{-1} \mathbf{\Sigma}^{\nicefrac{1}{2}}}_{:=\mathbf{M}_\mathbf{A}} + o(\alpha^{\nicefrac{2}{d}}),
\end{align*}
where the remainder is in operator norm for fixed $d$ and fixed
positive-definite $\mathbf{A}$. Note that
\begin{align*}
     \text{det}(\mathbf{M}_{\mathbf{A}}) = \text{det}(\mathbf{\Sigma}^{\nicefrac{1}{2}} \mathbf{A} \mathbf{\Sigma}^{\nicefrac{1}{2}}) \text{det}(\mathbf{\Sigma}^{\nicefrac{1}{2}} (\mathbf{\Sigma}^{\nicefrac{1}{2}} \mathbf{A} \mathbf{\Sigma}^{\nicefrac{1}{2}})^{-1} \mathbf{\Sigma}^{\nicefrac{1}{2}}) = \text{det}(\mathbf{\Sigma}).
\end{align*}
Because $\mathbf{M}_{\mathbf{A}}$ is positive-definite, it follows that
\begin{align*}
    \eta_{\max}(\mathbf{M}_{\mathbf{A}}) \geq \text{det}(\mathbf{M}_{\mathbf{A}})^{\nicefrac{1}{d}} =
    \text{det}(\mathbf{\Sigma})^{\nicefrac{1}{d}}
\end{align*}
where equality holds if and only if all of the eigenvalues of $\mathbf{M}_{\mathbf{A}}$ are equal. Equivalently, when $\mathbf{M}_{\mathbf{A}} \propto \text{det}(\mathbf{\Sigma})^{\nicefrac{1}{d}} \mathbf{I}_d$, which follows immediately if $\mathbf{\Sigma}^{\nicefrac{1}{2}} \mathbf{A} \mathbf{\Sigma}^{\nicefrac{1}{2}} \propto \mathbf{\Sigma}$, which occurs when $\mathbf{A} = \omega \mathbf{I}_d$ for some $\omega > 0$. Finally, we rule out singular choices of $\mathbf{A}$ to expand our optimization to $\mathbf{S}^d_{+}$. Suppose $\mathbf{A}$ is positive semi-definite with rank $k < d$. Then, recall from the discussion in \cref{pca_k_section} that $\nu_{j, \eta}(k) = 1$ for $j > k$. Thus, one could adversarially select $\bm{\beta}$ in directions that cannot be balanced by a positive semi-definite quadratic form, such that the worst case variance gap remains bounded away from zero even as $\alpha \to 0$. Putting everything together, it follows that for a sufficiently small $\alpha$,
\begin{align*}
            \operatorname*{arg\,min}_{\mathbf{A}\in\mathbf{S}^d_{+}}
        \sup_{||\bm{\beta}||_2<c} \left| \Delta_{\bm{\beta}}(\mathbf{A}) \right| = \{\omega\mathbf{I}_d:\omega>0\}.
\end{align*}
\end{proof}

\subsection{Proof of \texorpdfstring{\cref{PCA_QFR}}{Theorem 7}} \label{PCA_QFR_proof}
\begin{proof}[\textbf{Proof:}]
First, we will rewrite the covariance of the covariate mean differences under Principal Components Quadratic Form Rerandomization as
\begin{align*}
    \text{Cov}(\sqrt{n}(\bar{\mathbf{X}}_T - \bar{\mathbf{X}}_C) \mid \mathbf{X}, Q^k_\mathbf{A}(\sqrt{n} \widehat{\bm{\tau}}_{\mathbf{Z}}) \leq a) = \mathbf{V}\text{Cov}(\sqrt{n}(\bar{\mathbf{Z}}_T - \bar{\mathbf{Z}}_C) \mid \mathbf{X}, Q^k_\mathbf{A}( \sqrt{n}\widehat{\bm{\tau}}_{\mathbf{Z}}) \leq a) \mathbf{V}^T
\end{align*}
where $\bar{\mathbf{Z}}_T - \bar{\mathbf{Z}}_C = \mathbf{V}^T(\bar{\mathbf{X}}_T - \bar{\mathbf{X}}_C)$ and $Q^k_\mathbf{A}(\sqrt{n}\widehat{\bm{\tau}}_{\mathbf{Z}})$ is defined as in \cref{pca_k_section}. From here, let $\mathbf{Z}_k = \bar{\mathbf{Z}}^{(k)}_T - \bar{\mathbf{Z}}^{(k)}_C \in \mathbb{R}^k$ be the covariate mean differences for the first $k$ principal components and $\mathbf{Z}_{d-k} = \bar{\mathbf{Z}}^{(d - k)}_T - \bar{\mathbf{Z}}^{(d - k)}_C \in \mathbb{R}^{d-k}$ be the mean differences for the last $d - k$ components. Observe that under \cref{asymptotic_norm_condition}, $\sqrt{n}\mathbf{Z}_k \mid \mathbf{X} \sim \mathcal{N}(\mathbf{0}, \mathbf{\Lambda}_k)$ where $\mathbf{\Lambda}_{k}$ represents the diagonal matrix of the first $k$ eigenvalues of $\mathbf{\Lambda}$. Then, under \cref{asymptotic_norm_condition} and \cref{general_balance_condition}, we can apply \cref{Theorem1} to see that, as $n \to \infty$,
\begin{align*}
    \text{Cov}(\sqrt{n} \mathbf{Z}_k \mid \mathbf{X}, Q^k_\mathbf{A}(\sqrt{n} \widehat{\bm{\tau}}_{\mathbf{Z}}) \leq a) &= \mathbf{\Lambda}^{\nicefrac{1}{2}}_k \mathbf{\Omega}_k \Big( \text{diag}\{(\nu_{j, \eta}(k))_{1 \leq j \leq k} \} \Big) \mathbf{\Omega}^T_k \mathbf{\Lambda}^{\nicefrac{1}{2}}_k
\end{align*}
where $\nu_{j,\eta}(k) = \mathbb{E}\left[\mathcal{Z}^2_j \mid  \sum^k_{\ell=1} \eta_\ell \mathcal{Z}^2_\ell \leq a \right]$ and $\eta_1, \ldots, \eta_k$ are the eigenvalues of $\mathbf{\Lambda}^{\nicefrac{1}{2}}_k \mathbf{A} \mathbf{\Lambda}^{\nicefrac{1}{2}}_k$ and $\mathbf{\Omega}_k$ is the orthogonal matrix of eigenvectors of $\mathbf{\Lambda}^{\nicefrac{1}{2}}_k \mathbf{A} \mathbf{\Lambda}^{\nicefrac{1}{2}}_k$. Next, using the asymptotic normality of $\sqrt{n}\mathbf{Z} \mid \mathbf{X} \sim \mathcal{N}(\mathbf{0}, \mathbf{\Lambda})$, it is clear that $\mathbf{Z}_{d-k}$ and $\mathbf{Z}_k$ are independent since $\mathbf{\Lambda}$ is a diagonal matrix. Thus, it follows that
\begin{align*}
    \text{Cov}(\sqrt{n}\mathbf{Z}_{d-k} \mid \mathbf{X}, Q^k_\mathbf{A}(\sqrt{n}\widehat{\bm{\tau}}_{\mathbf{Z}}) \leq a) &= \text{Cov}(\sqrt{n}\mathbf{Z}_{d-k} \mid \mathbf{X}) = \mathbf{\Lambda}_{d-k}
\end{align*}
as there are no constraints imposed on the bottom $d - k$ principal components. Similarly, if we consider the conditional covariance between $\mathbf{Z}_k$ and $\mathbf{Z}_{d - k}$ where we let $Z_{k, i}$ be the $i$th element of $\mathbf{Z}_k$ and $Z_{d-k, j}$ be the $j$th element of $\mathbf{Z}_{d - k}$, then
\begin{align*}
    \mathbb{E}\left[Z_{k, i} Z_{d-k, j} \mid \mathbf{X},  Q^k_{\mathbf{A}}(\sqrt{n}\widehat{\bm{\tau}}_{\mathbf{Z}}) \leq a \right] &= \mathbb{E}\left[Z_{k, i} \mathbb{E}\left( Z_{d-k, j} \mid \mathbf{X} \right) \mid  \mathbf{X}, Q^k_{\mathbf{A}}( \sqrt{n}\widehat{\bm{\tau}}_{\mathbf{Z}}) \leq a \right] = 0.
\end{align*}
Putting everything together, it follows that under \cref{asymptotic_norm_condition} and \cref{general_balance_condition}, as $n \to \infty$,
\begin{align*}
    \text{Cov}(\sqrt{n}(\bar{\mathbf{X}}_T - \bar{\mathbf{X}}_C) \mid \mathbf{X}, Q^k_\mathbf{A}(\sqrt{n} \widehat{\bm{\tau}}_{\mathbf{Z}}) \leq a) &= \mathbf{V}\begin{pmatrix}
        \mathbf{\Lambda}^{\nicefrac{1}{2}}_k \mathbf{\Omega}_k \mathbf{\Psi}_k \mathbf{\Omega}^T_k \mathbf{\Lambda}^{\nicefrac{1}{2}}_k  & \mathbf{0} \\
        \mathbf{0} & \mathbf{\Lambda}_{d-k}
    \end{pmatrix} \mathbf{V}^T \\
     &= \mathbf{V} \mathbf{\Lambda}^{\nicefrac{1}{2}}  \begin{pmatrix}
        \mathbf{\Omega}_k \mathbf{\Psi}_k \mathbf{\Omega}^T_k   & \mathbf{0} \\
        \mathbf{0} & \mathbf{I}_{d-k}
    \end{pmatrix} \mathbf{\Lambda}^{\nicefrac{1}{2}} \mathbf{V}^T \\
    &= \mathbf{\Sigma}^{\nicefrac{1}{2}} \mathbf{V}  \begin{pmatrix}
        \mathbf{\Omega}_k \mathbf{\Psi}_k \mathbf{\Omega}^T_k   & \mathbf{0} \\
        \mathbf{0} & \mathbf{I}_{d-k}
    \end{pmatrix} \mathbf{V}^T  \mathbf{\Sigma}^{\nicefrac{1}{2}} 
\end{align*}
where $\mathbf{\Psi}_k = \text{diag}\{(\nu_{j, \eta}(k))_{1 \leq j \leq k} \}$, which completes the proof.
\end{proof}

\subsection{Proof of \texorpdfstring{\cref{drop_pcs}}{Proposition 2}} \label{drop_pcs_proof}
\begin{proof}[\textbf{Proof:}]
    First, observe that for $\mathbf{A}_d \in \mathbb{R}^{d \times d}$, under \cref{asymptotic_norm_condition} and \cref{general_balance_condition}, as $n \to \infty$ we know that
    \begin{align*}
        \text{Var}(\sqrt{n}(\widehat{\tau} - \tau) \mid \mathbf{X}, Q^d_\mathbf{A}(\sqrt{n} \widehat{\bm{\tau}}_{\mathbf{X}}) \leq a_d) &= \sum^d_{j=1} \left((\mathbf{\Omega}^T   \mathbf{V})_j \mathbf{\Lambda}^{\nicefrac{1}{2}}  \bm{\beta}_Z \right)^2 \nu_{j, \eta}(d) + V_{\tau \tau}(1 - R^2)
    \end{align*}
    where $V_{\tau \tau} = \text{Var}(\sqrt{n}(\widehat{\tau} - \tau) \mid \mathbf{X})$, $\nu_{j, \eta}(d)$ is defined as in \cref{q_definition}, and $R^2$ is the squared multiple correlation between the potential outcomes and $\mathbf{X}$. Next, following \cref{PCA_QFR}, for some $\mathbf{A}_k \in \mathbb{R}^{k \times k}$ with $k < d$ then $\text{Var}(\sqrt{n}(\widehat{\tau} - \tau) \mid \mathbf{X}, Q^k_\mathbf{A}(\sqrt{n} \widehat{\bm{\tau}}_{\mathbf{Z}}) \leq a_k)$ is given by
    \begin{align*}
         \bm{\beta}^T\mathbf{\Sigma}^{\nicefrac{1}{2}} \mathbf{V}  \begin{pmatrix}
        \mathbf{\Omega}_k \Big( \text{diag}(\nu_{j, \eta}(k))_{1 \leq j \leq k} \Big) \mathbf{\Omega}^T_k   & \mathbf{0} \\
        \mathbf{0} & \mathbf{I}_{d-k}
    \end{pmatrix} \mathbf{V}^T  \mathbf{\Sigma}^{\nicefrac{1}{2}} \bm{\beta}  + V_{\tau \tau}(1 - R^2),
    \end{align*}
    or, after simplifying following the proof of \cref{diff_in_vars}, by
    \begin{align*}
         \sum^k_{j=1}(\mathbf{\Omega}^T_k \mathbf{\Lambda}^{\nicefrac{1}{2}}_k \bm{\beta}^{(k)}_Z)^2_j \nu_{j, \eta}(k) + \sum^d_{j = k+1} \beta^2_{Z,j} \lambda_j  + V_{\tau \tau}(1 - R^2),
    \end{align*}
    where $\bm{\beta}^{(k)}_Z$ denotes the first $k$ values of $\bm{\beta}_Z$ and $\mathbf{\Lambda}_k$ contains the first $k$ eigenvalues of $\mathbf{\Sigma}$. Note that $a_d$ and $a_k$ are chosen to have common acceptance probability $\alpha$. Therefore, it follows that
    \begin{align*}
        \text{Var}(\sqrt{n}(\widehat{\tau} - \tau) \mid \mathbf{X}, Q^d_\mathbf{A}(\sqrt{n} \widehat{\bm{\tau}}_{\mathbf{X}}) \leq a_d) \geq \text{Var}(\sqrt{n}(\widehat{\tau} - \tau) \mid \mathbf{X}, Q^k_\mathbf{A}(\sqrt{n} \widehat{\bm{\tau}}_{\mathbf{Z}}) \leq a_k)
    \end{align*}
if and only if
    \begin{align*}
        \sum^d_{j=1} \left((\mathbf{\Omega}^T   \mathbf{V})_j \mathbf{\Lambda}^{\nicefrac{1}{2}}  \bm{\beta}_Z \right)^2 \nu_{j, \eta}(d)  \geq \sum^k_{j=1}(\mathbf{\Omega}^T_k \mathbf{\Lambda}^{\nicefrac{1}{2}}_k \bm{\beta}^{(k)}_Z)^2_j \nu_{j, \eta}(k) + \sum^d_{j = k+1} \beta^2_{Z,j} \lambda_j 
    \end{align*}
    If we assume that $\mathbf{\Sigma}$ and $\mathbf{\Sigma}^{\nicefrac{1}{2}} \mathbf{A}_d \mathbf{\Sigma}^{\nicefrac{1}{2}}$ share an eigenbasis, as well as that $\mathbf{\Lambda}^{\nicefrac{1}{2}}_k$ and $\mathbf{\Lambda}^{\nicefrac{1}{2}}_k \mathbf{A}_k \mathbf{\Lambda}^{\nicefrac{1}{2}}_k$ share an eigenbasis, this inequality simplifies nicely. Then, it follows that the variance of $\widehat{\tau}$ under the reduced quadratic form is less than under the full model if and only if
    \begin{align*}
        \sum^k_{j=1} \beta^2_{Z, j} \lambda_j (\nu_{j, \eta}(d) - \nu_{j, \eta}(k)) \geq \sum^d_{j=k+1} \beta^2_{Z, j} \lambda_j (1 - \nu_{j, \eta}(d)).
    \end{align*}
\end{proof}
\subsection{Proof of \texorpdfstring{\cref{eigenvalues_qfr}}{Corollary 1}} \label{eigenvalues_qfr_proof}
\begin{proof}[\textbf{Proof:}] We are looking to solve for $\lambda$ such that
\begin{align} \label{gen_eigen_eq}
    \text{det}\left(\text{Cov}\left(\sqrt{n}\left(\bar{\mathbf{X}}_T - \bar{\mathbf{X}}_C\right) \mid \mathbf{X}, Q_\mathbf{A}(\sqrt{n} \widehat{\bm{\tau}}_{\mathbf{X}}) \leq a\right) - \lambda \mathbf{\Sigma} \right) = 0.
\end{align}
Then, applying \cref{Theorem1}, we can see that under \cref{asymptotic_norm_condition} and \cref{general_balance_condition}, as $n \to \infty$,
    \begin{align*}
        \text{Cov}\left(\sqrt{n}\left(\bar{\mathbf{X}}_T - \bar{\mathbf{X}}_C \right) \mid \mathbf{X}, Q_{\mathbf{A}}(\sqrt{n} \widehat{\bm{\tau}}_{\mathbf{X}}) \leq a \right) &= \mathbf{\Sigma}^{\nicefrac{1}{2}} \mathbf{\Omega} \Big( \text{diag}\{(\nu_{j, \eta})_{1 \leq j \leq d} \} \Big) \mathbf{\Omega}^T \mathbf{\Sigma}^{\nicefrac{1}{2}}
    \end{align*}
    where $\mathbf{\Omega}$ is the orthogonal matrix of eigenvectors of $\mathbf{\Sigma}^{\nicefrac{1}{2}} \mathbf{A} \mathbf{\Sigma}^{\nicefrac{1}{2}}$. Thus, it follows that the determinant from \cref{gen_eigen_eq} can be written as
    \begin{align*}
        \text{det}\left(\mathbf{\Sigma}^{\nicefrac{1}{2}} \mathbf{\Omega} \mathbf{\Psi} \mathbf{\Omega}^T \mathbf{\Sigma}^{\nicefrac{1}{2}} - \lambda \mathbf{\Sigma} \right) &= \text{det}\left( \left(\mathbf{\Sigma}^{\nicefrac{1}{2}} \mathbf{\Omega} \mathbf{\Psi} \mathbf{\Omega}^T \mathbf{\Sigma}^{-\nicefrac{1}{2}} - \lambda \mathbf{I}_d \right) \mathbf{\Sigma} \right) \\
        &=  \text{det}\left( \mathbf{\Sigma}^{\nicefrac{1}{2}} \mathbf{\Omega} \mathbf{\Psi} \mathbf{\Omega}^T \mathbf{\Sigma}^{-\nicefrac{1}{2}} - \lambda \mathbf{I}_d \right) \text{det}\left(\mathbf{\Sigma} \right)
    \end{align*}
     where for notational convenience we have defined $\mathbf{\Psi} = \text{diag}\{(\nu_{j, \eta})_{1 \leq j \leq d} \}$ and the second equality follows by the property of the determinant that $\text{det}(\mathbf{A} \mathbf{B}) = \text{det}(\mathbf{A}) \text{det}(\mathbf{B})$ for any two matrices $\mathbf{A}$ and $\mathbf{B}$. Thus, solving for $\lambda$ such that $\text{det}\left( \mathbf{\Sigma}^{\nicefrac{1}{2}} \mathbf{\Omega} \mathbf{\Psi} \mathbf{\Omega}^T \mathbf{\Sigma}^{-\nicefrac{1}{2}} - \lambda \mathbf{I}_d \right) = 0$ equivalently solves for \cref{gen_eigen_eq}. Importantly, we are now solving for the eigenvalues of $\mathbf{\Sigma}^{\nicefrac{1}{2}} \mathbf{\Omega} \mathbf{\Psi} \mathbf{\Omega}^T \mathbf{\Sigma}^{-\nicefrac{1}{2}}$, which is a similar matrix to $\mathbf{\Psi}$ because there exists an invertible matrix $\mathbf{P}$ such that $\mathbf{P}^{-1} \mathbf{\Sigma}^{\nicefrac{1}{2}} \mathbf{\Omega} \mathbf{\Psi} \mathbf{\Omega}^T \mathbf{\Sigma}^{-\nicefrac{1}{2}} \mathbf{P} = \mathbf{\Psi}$, i.e.\ $\mathbf{P} = \mathbf{\Sigma}^{\nicefrac{1}{2}} \mathbf{\Omega}$. Therefore, $\mathbf{\Sigma}^{\nicefrac{1}{2}} \mathbf{\Omega} \mathbf{\Psi} \mathbf{\Omega}^T \mathbf{\Sigma}^{-\nicefrac{1}{2}}$ and $\mathbf{\Psi}$ share the same eigenvalues, which are given by $\nu_{1, \eta}, \ldots, \nu_{d, \eta}$ since $\mathbf{\Psi}$ is a diagonal matrix.

Next, we solve for the eigenvalues of $\text{Cov}\left(\sqrt{n}\widehat{\bm{\tau}}_{\mathbf{X}} \mid \mathbf{X}, Q_{\mathbf{A}}(\sqrt{n} \widehat{\bm{\tau}}_{\mathbf{X}}) \leq a \right)$. To do so, we make the simplifying assumption that $\mathbf{\Gamma}$ diagonalizes $\mathbf{\Sigma}^{\nicefrac{1}{2}} \mathbf{A} \mathbf{\Sigma}^{\nicefrac{1}{2}}$, i.e., that $\mathbf{\Sigma}$ and $\mathbf{\Sigma}^{\nicefrac{1}{2}} \mathbf{A} \mathbf{\Sigma}^{\nicefrac{1}{2}}$ share an eigenbasis. As an example, a sufficient condition for this assumption is that the matrix $\mathbf{A}$ can be written as $\mathbf{\Gamma} \mathbf{D} \mathbf{\Gamma}^T$ for any diagonal matrix $\mathbf{D}$. Under this condition,
    \begin{align*}
       \mathbf{\Gamma}^T \left( \mathbf{\Sigma}^{\nicefrac{1}{2}} \mathbf{A} \mathbf{\Sigma}^{\nicefrac{1}{2}} \right) \mathbf{\Gamma} = \mathbf{\Gamma}^T \left( \mathbf{\Sigma}^{\nicefrac{1}{2}} \mathbf{\Gamma} \mathbf{D} \mathbf{\Gamma}^T \mathbf{\Sigma}^{\nicefrac{1}{2}} \right) \mathbf{\Gamma} = \mathbf{\Lambda}^{\nicefrac{1}{2}} \mathbf{D} \mathbf{\Lambda}^{\nicefrac{1}{2}}.
    \end{align*}
    Therefore, in this setting we can let $\mathbf{\Omega} = \mathbf{\Gamma}$, which simplifies the covariance to
    \begin{align*}
        \text{Cov}\left(\sqrt{n}\left(\bar{\mathbf{X}}_T - \bar{\mathbf{X}}_C \right) \mid \mathbf{X}, Q_{\mathbf{A}}(\sqrt{n} \widehat{\bm{\tau}}_{\mathbf{X}}) \leq a \right) &= \Big( \mathbf{\Gamma} \mathbf{\Lambda}^{\nicefrac{1}{2}} \mathbf{\Gamma}^T \Big) \mathbf{\Gamma} \mathbf{\Psi} \mathbf{\Gamma}^T \Big( \mathbf{\Gamma} \mathbf{\Lambda}^{\nicefrac{1}{2}} \mathbf{\Gamma}^T \Big) \\
        &=  \mathbf{\Gamma} \mathbf{\Lambda}^{\nicefrac{1}{2}} \mathbf{\Psi} \mathbf{\Lambda}^{\nicefrac{1}{2}} \mathbf{\Gamma}^T.
    \end{align*}
Thus, plugging this expression into the determinant calculation yields
    \begin{align*}
       \text{det}\left( \mathbf{\Gamma} \mathbf{\Lambda}^{\nicefrac{1}{2}} \mathbf{\Psi} \mathbf{\Lambda}^{\nicefrac{1}{2}} \mathbf{\Gamma}^T   - \lambda \mathbf{I}_d \right) &= \text{det}\left( \mathbf{\Gamma} \mathbf{\Lambda}^{\nicefrac{1}{2}} \mathbf{\Psi} \mathbf{\Lambda}^{\nicefrac{1}{2}} \mathbf{\Gamma}^T  - \lambda \mathbf{\Gamma} \mathbf{\Gamma}^T \right) \\
        &= \text{det}\left(\mathbf{\Gamma} \Big(  \mathbf{\Lambda}^{\nicefrac{1}{2}} \mathbf{\Psi} \mathbf{\Lambda}^{\nicefrac{1}{2}}  - \lambda \mathbf{I}_d \Big) \mathbf{\Gamma}^T \right).
    \end{align*}
Then, using the property of the determinant that $\text{det}(\mathbf{A} \mathbf{B}) = \text{det}(\mathbf{A}) \text{det}(\mathbf{B})$ for any two matrices $\mathbf{A}$ and $\mathbf{B}$, we may equivalently solve for $\text{det}\left(  \mathbf{\Lambda}^{\nicefrac{1}{2}} \mathbf{\Psi} \mathbf{\Lambda}^{\nicefrac{1}{2}}  - \lambda \mathbf{I}_d \right) = 0$ since
    \begin{align*}
        \text{det}\left(\mathbf{\Gamma} \Big(  \mathbf{\Lambda}^{\nicefrac{1}{2}} \mathbf{\Psi} \mathbf{\Lambda}^{\nicefrac{1}{2}}  - \lambda \mathbf{I}_d \Big) \mathbf{\Gamma}^T \right) = \text{det}\left(\mathbf{\Gamma} \right) \text{det}\left(  \mathbf{\Lambda}^{\nicefrac{1}{2}} \mathbf{\Psi} \mathbf{\Lambda}^{\nicefrac{1}{2}}  - \lambda \mathbf{I}_d \right) \text{det}\left( \mathbf{\Gamma}^T \right), 
    \end{align*}
    so solving for $\lambda$ that makes $\text{det}\left(  \mathbf{\Lambda}^{\nicefrac{1}{2}} \mathbf{\Psi} \mathbf{\Lambda}^{\nicefrac{1}{2}}  - \lambda \mathbf{I}_d \right) = 0$ also solves for the eigenvalues of $\text{Cov}\left(\sqrt{n}\widehat{\bm{\tau}}_{\mathbf{X}} \mid \mathbf{X}, Q_{\mathbf{A}}(\sqrt{n} \widehat{\bm{\tau}}_{\mathbf{X}}) \leq a \right)$. From here, it is easy to see that
    \begin{align*}
        \mathbf{\Lambda}^{\nicefrac{1}{2}} \mathbf{\Psi} \mathbf{\Lambda}^{\nicefrac{1}{2}} = \begin{pmatrix}
            \lambda_1 \nu_{1, \eta} & & 0\\
            & \ddots & \\
           0 & & \lambda_d \nu_{d, \eta}
        \end{pmatrix}
    \end{align*}
    so letting $\lambda = \lambda_j \nu_{j, \eta}$ makes the determinant zero for all $j = 1, \ldots, d$. Thus, under \cref{asymptotic_norm_condition} and \cref{general_balance_condition}, the eigenvalues of $\text{Cov}\left(\sqrt{n}\widehat{\bm{\tau}}_{\mathbf{X}} \mid \mathbf{X}, Q_{\mathbf{A}}(\sqrt{n} \widehat{\bm{\tau}}_{\mathbf{X}}) \leq a \right)$ are given by $\lambda_1 \nu_{1, \eta}, \ldots, \lambda_d \nu_{d, \eta}$ when $\mathbf{\Sigma}$ and $\mathbf{\Sigma}^{\nicefrac{1}{2}} \mathbf{A} \mathbf{\Sigma}^{\nicefrac{1}{2}}$ share an eigenbasis.
\end{proof}

\end{document}